\newif\ifdraft
\newif\iffullversion

\fullversiontrue

\iffullversion
\documentclass[a4paper,11pt]{article}
\else
\documentclass[10pt, conference]{IEEEtran}
\fi
\usepackage[ascii]{inputenc} %
\usepackage[T1]{fontenc}
\usepackage{microtype,amssymb,paralist,bbm,version}
\usepackage[\iffullversion cmex10\fi]{amsmath}
\usepackage{multicol,xcolor,graphicx,ifpdf,etoolbox,relsize}
\usepackage{tikz}
\usepackage[noend,nofillcomment,boxed,linesnumbered]{algorithm2e}
\usepackage[font=footnotesize,labelfont=bf]{caption}
\iffullversion
\usepackage[nottoc]{tocbibind}
\usepackage[toc]{multitoc}
\usepackage{ae,aecompl} %
\usepackage{longtable}
\fi
\iffullversion\else
\usepackage{cite}
\fi
\usepackage[define,symbolindex]{declmath}
\usepackage{makeidx}\makeindex
\expandafter\ifx\csname opt@preview.sty\endcsname\relax
\usepackage[\iffullversion\else pass\fi]{geometry}
\fi
\usepackage[users=QARW]{uchanges}
\ChangesSetUserColor[Dominique Unruh]{Q}{red}
\ChangesSetUserColor[Andris Ambainis]{A}{blue}
\ChangesSetUserColor[Ansis Rosmanis]{R}{magenta}
\ChangesSetUserColor[Reviewer]{W}{pink}
\usepackage[\iffullversion\else draft,\fi 
pdfauthor={Andris Ambainis, Ansis Rosmanis, Dominique Unruh},
pdftitle={Quantum Attacks on Classical Proof Systems - The Hardness of Quantum Rewinding}]{hyperref}
\usepackage[envcountsame,multiautoref,delaytext,fixpreview]{latexhacks}

\input{Qcircuit}
\let\ket\undefined
\let\bra\undefined

\iffullversion
\includeversion{fullversion}
\excludeversion{shortversion}
\newcommand\fullshort[2]{#1}
\else
\includeversion{shortversion}
\excludeversion{fullversion}
\newcommand\fullshort[2]{#2}
\fi

\newcommand\shortonly[1]{\fullshort{}{#1}}
\newcommand\fullonly[1]{\fullshort{#1}{}}

\input{macros}

\newtheorem{definition}{Definition}

\newtheorem{corollary}{Corollary}

\newtheorem{lemma}{Lemma}
\newtheorem{claim}{Claim}
\newtheorem{fact}{Fact}
\newtheorem{theorem}{Theorem}

\newenvironment{proof}[1][]{\medskip\noindent
  \textit{Proof\ifx\proof#1\proof\else\
    #1\fi. }}{\qed\global\let\qed\qedmacro\par}
\newcommand\autonameref[1]{\autoref{#1} {\smaller\sffamily(\nameref{#1})}}

\shortonly{
  \renewcommand\paragraph[1]{\par\noindent\textbf{#1}}}

\newcommand\tqed{\shortonly{\hfill$\diamond$}}
\newcommand\mtqed{\shortonly{\tag*{$\diamond$}}}

\AtBeginDocument{
  \ifx\PreviewMacro\undefined\else
  \PreviewMacro*\For
  \PreviewMacro*\If
  \fi
}

\ifdraft\else
\let\include\input
\fi

\makeatletter
\DeclareRobustCommand\ch@sidecross{%
  {\tikz[remember picture,overlay] {
    \node (target) at (0,.2\baselineskip) {};
    \node (cross) [xshift=\ch@sidecrosspos pt] at (current page.west |- target) {\normalsize\sffamily\bfseries X};
    }}}

\begin{document}

\iffullversion\else
\makeatletter
\def\symbolindex@warn#1#2#3{}
\makeatother
\fi

\begin{shortversion}
  \date{}
\end{shortversion}

\title{Quantum Attacks on Classical Proof Systems\\[6pt]\Large The Hardness of Quantum Rewinding}

\iffullversion
\author{Andris Ambainis\\
\footnotesize University of Latvia and\\[-2pt]
\footnotesize Institute for Advanced Study\\[-2pt]
\footnotesize Princeton
\and 
Ansis Rosmanis\\
\footnotesize Institute for Quantum Computing\\[-2pt]
\footnotesize School of Computer Science\\[-2pt]
\footnotesize University of Waterloo
\and Dominique Unruh\\
\footnotesize University of Tartu}
\else
\author{\IEEEauthorblockN{Andris Ambainis}
\IEEEauthorblockA{
\footnotesize University of Latvia and\\
Institute for Advanced Study\\
Princeton}
\and
\IEEEauthorblockN{Ansis Rosmanis}
\IEEEauthorblockA{\footnotesize Institute for Quantum Computing\\
School of Computer Science\\
University of Waterloo}
\and
\IEEEauthorblockN{Dominique Unruh}
\IEEEauthorblockA{\footnotesize University of Tartu}
}
\fi

\maketitle

\fullonly{
\renewenvironment{abstract}{\begin{quote}\small\textbf{Abstract.}}{\end{quote}}
}

\begin{abstract}
  Quantum zero-knowledge proofs and quantum proofs of knowledge are
  inherently difficult to analyze because their security analysis uses
  rewinding. Certain cases of quantum rewinding are handled by the
  results by Watrous (SIAM J Comput, 2009) and Unruh (Eurocrypt 2012),
  yet in general the problem remains elusive. We show that this is not
  only due to a lack of proof techniques: relative to an oracle,
  we show that classically secure proofs and proofs of knowledge are
  insecure in the quantum setting.

  More specifically, sigma-protocols, the Fiat-Shamir construction,
  and Fischlin's proof system are quantum insecure under assumptions that are
  sufficient for classical security. Additionally, we show that for
  similar reasons, computationally binding commitments provide almost
  no security guarantees in a quantum setting.

  To show these results, we develop the ``pick-one trick'', a general
  technique that allows an adversary to find one value satisfying a
  given predicate, but not two.
\end{abstract}

\begin{shortversion}
  \begin{IEEEkeywords}
    quantum cryptography; quantum query complexity; rewinding; random oracles
  \end{IEEEkeywords}
\end{shortversion}

\begin{fullversion}
  \columnseprule.4pt
  \columnsep30pt
  \tableofcontents
\end{fullversion}

\section{Introduction}

Quantum computers threaten classical cryptography. With a quantum
computer, an attacker would be able to break all schemes based on the
hardness of factoring, or on the hardness of discrete logarithms
\cite{Shor:1994:Algorithms}, this would affect most public key encryption and
signature schemes is use today. For symmetric ciphers and hash
functions, longer key and output lengths will be required due to
considerable improvements in brute force attacks
\cite{grover,bht97collision}. These threats lead to the
question: how can classical cryptography be made secure against
quantum attacks? Much research has been done towards cryptographic schemes
based on hardness assumptions not known to be vulnerable to quantum
computers, e.g., lattice-based cryptography. (This is called \emph{post-quantum
cryptography}%
\index{post-quantum cryptography}; see~\cite{postqc-book} for a somewhat dated
survey.) Yet, identifying useful quantum-hard assumptions is only half of the
problem. Even if the underlying assumption holds against quantum
attackers, for many classically secure protocols it is not clear if
they also resist quantum attacks: the proof techniques used in the
classical setting often cannot be applied in the quantum world. This
raises the question whether it is just our proof techniques that are
insufficient, or whether the protocols themselves are quantum
insecure. The most prominent example are zero-knowledge proofs. To
show the security of a zero-knowledge proof
system,\fullonly{\footnote{\label{revfootnote:explain.proof}Reminder: a
  proof or proof system is a protocol where a prover convinces a verifier of 
  the validity of a statement $s$. It is zero-knowledge if the view of the verifier can be 
  simulated without knowing a witness $w$ for the statement $s$ (i.e., 
  the verifier learns nothing about the witness). A proof of knowledge
  is a proof which additionally convinces the verifier that the
  prover could provide a witness $w$ (i.e., not just the mere existence
  of $w$ is proven). Arguments and arguments of knowledge are like proofs 
  and proofs of 
  knowledge, except that they are secure only against computationally
  limited provers.}}
 one typically uses
rewinding. That is, in a hypothetical execution, the adversary's state
is saved, and the adversary is executed several times starting from
that state. In the quantum setting, we cannot do that: saving a
quantum state means cloning it, violating the no-cloning theorem
\cite{nocloning}. Watrous \cite{watrous-qzk} showed that for many
zero-knowledge proofs, security can be shown using a quantum version
of the rewinding technique. (Yet this technique is not as versatile as
classical rewinding. For example, the quantum security of the graph
non-isomorphism proof system \cite{Goldreich:1991:Proofs} is an open problem.)
Unruh \cite{qpok} noticed that Watrous' rewinding cannot be used to
show the security of proofs of knowledge; he developed a new rewinding
technique to show that so-called sigma-protocols are proofs of
knowledge.  Yet, in \cite{qpok} an unexpected condition was needed: their technique only
applies to proofs of knowledge with \emph{strict soundness} (which
roughly means that the last message in the interaction is determined
by the earlier ones); this condition is not needed in the classical
case. The security of sigma-protocols without strict soundness (e.g.,
graph isomorphism \cite{Goldreich:1991:Proofs}) was
left open. The
problem also applies to arguments as well (i.e., computationally-sound proof
systems, without ``of knowledge''), as these are often shown secure by proving that they
are actually arguments of knowledge. Further cases where new proof
techniques are needed in the quantum setting are schemes involving
random oracles. Various proof techniques were developed
\cite{boneh11quantumro,zhandry:quantum.ibe,qtc,boneh13quantumsigs,qpos}, but
all are restricted to specific cases, none of them matches the power of the
classical proof techniques.

To summarize: For many constructions that are easy to prove secure
classically, proofs in the quantum setting are much harder and come
with additional conditions limiting their applicability. The question
is: does this only reflect our lack of understanding of the quantum
setting, or are those additional conditions indeed necessary? 
Or could it be that those classically secure constructions are 
actually insecure quantumly?

\paragraph{Our contribution.}
We
show, relative to an oracle, that the answer is indeed \textbf{yes}:
\begin{compactitem}
\item\pagelabel{revpage:contrib.list} Sigma-protocols are
  not necessarily quantum proofs of knowledge, even if they are
  classical proofs of knowledge. In particular, the strict soundness
  condition from \cite{qpok} is necessary. (\autoref{theo:know.break.sigma})
\item In the computational setting, sigma-protocols are not
  necessarily quantum arguments, even if they are classical
  arguments.
  (\autoref{theo:break.sigma.comp})
\item The Fiat-Shamir construction \cite{Fiat:1987:How} for
  non-interactive proofs of knowledge in the random oracle model does
  not give rise to quantum proofs of knowledge. And in the
  computational setting, not even to quantum arguments.
  (\multiautoref{theo:know.break.fs,theo:break.fs.comp})
\item Fischlin's non-interactive proof of knowledge in the random
  oracle model \cite{fischlin05online} is not a quantum proof of
  knowledge. (This is remarkable because in contrast to Fiat-Shamir,
  the classical security proof of Fischlin's scheme does not use rewinding.)  And in the
  computational setting, it is not even an argument.
  (\multiautoref{theo:know.break.fischlin,theo:break.fischlin.comp})
\item Besides proof systems, we also have negative results for
  commitment schemes. The usual classical definition of
  computationally binding commitments is that the adversary cannot
  provide openings to two different values for the same
  commitment. Surprisingly, relative to an oracle, there are
  computationally binding commitments where a quantum adversary can
  open the commitment to any value he chooses (just not to two values
  \emph{simultaneously}).
  (\autoref{theo:com.attack})
\item The results on commitments in turn allow us to strengthen the
  above results for proof systems. While it is known that even in the
  quantum case, sigma-protocols with so-called ``strict soundness''
  (the third message is uniquely determined by the other two) are
  proofs and proofs of knowledge \cite{qpok}, using the computational
  variant of this property leads to schemes that are not even
  computationally secure.
  (\multiautoref{theo:know.break.sigma,theo:break.sigma.comp,theo:know.break.fs,theo:break.fs.comp,theo:know.break.fischlin,theo:break.fischlin.comp}.)
    \end{compactitem}
\begin{figure*}[t]
  \begin{center}
    \renewcommand\r[1]{\textsuperscript{\ref{#1}}}
    \renewcommand\c[1]{\textsuperscript{\cite{#1}}}
    \small\sffamily
    \begin{tabular}{|ccc|cc|cc|cc|}
      \hline
      \multicolumn{3}{|c|}{\bfseries Underlying sigma-protocol} & \multicolumn{2}{|c|}{\bfseries Sig.-pr.\,used\,directly} & 
      \multicolumn{2}{|c|}{\bfseries Fiat-Shamir} & \multicolumn{2}{|c|}{\bfseries Fischlin} \\
      \hline
      zero- & special & strict     &           &       &     &       &   &       \\
      knowledge &  soundness &  soundness & PoK       & proof & PoK & proof & PoK & proof \\
      \hline
      stat & perf & comp             & attack\r{theo:know.break.sigma}& stat\c{watrous-qzk}  & attack\r{theo:know.break.fs}    & ?   & attack\r{theo:know.break.fischlin} & ? \\
      stat & comp & comp             & attack\r{theo:break.sigma.comp}& attack\r{theo:break.sigma.comp}& attack\r{theo:break.fs.comp}& attack\r{theo:break.fs.comp}& attack\r{theo:break.fischlin.comp} & attack\r{theo:break.fischlin.comp} \\
      stat & perf & perf             & stat\c{qpok}      & stat\c{watrous-qzk}  & ?         & ?     & ?         & ? \\
      \hline
    \end{tabular}
  \end{center}
  \fullonly{\vspace*{-5mm}}
  \caption{Taxonomy of proofs of knowledge. For different combinations
    of security properties of the underlying sigma-protocol
    (statistical (stat)/perfect (perf)/computational (comp)), is there an attack in the
    quantum setting (relative to an oracle)? Or do we get a
    statistically/computationally secure proof/proof of knowledge (PoK)? The
    superscripts refer to theorem numbers in this paper or to
    literature references. Note that in all cases, classically we have
    at least computational security.}
  \label{fig:overview}
\end{figure*}
\autoref{fig:overview} gives an overview of the results relating to proofs of knowledge.
Our main result are the separations listed in the bullet points
above. Towards that goal, we additionally develop two tools that
may be of independent interest in quantum cryptographic proofs:
\begin{compactitem}
\item  \autoref{sec:pickone}: We develop the ``pick-one'' trick, a technique for providing the
  adversary with the ability to compute a value with a certain
  property, but not two of them. (See ``our technique'' below.) This
  technique and the matching lower bound on the adversary's query
  complexity may be useful for developing further oracle separations
  between quantum and classical security. (At least it gives rise to
  all the separations listed above.)
\item \fullonly{\autoref{sec:opsi}:} We show \shortonly{(in the full version) }how to create an oracle that allows us to create
  arbitrarily many copies of a given state $\ket\Psi$, but that is not
  more powerful than having many copies of $\ket\Psi$, even if queried
  in superposition.  Again, this might be useful for other oracle
  separations, too.  (The construction of $\Opsi$ in
  \autoref{sec:pickone} is an example for this.)
\end{compactitem}

\paragraph{Related work.} 
Van der Graaf~\cite{Graaf:1998:Towards} first noticed that security definitions
based on rewinding might be problematic in the quantum
setting. Watrous~\cite{watrous-qzk} showed how the problems with quantum
rewinding can be solved for a large class of zero-knowledge
proofs. Unruh~\cite{qpok} gave similar results for proofs of knowledge;
however he introduced the additional condition ``strict soundness''
and they did not cover the computational case (arguments and arguments
of knowledge). Our work (the results on sigma-protocols,
\autoref{sec:attack.sigma}) shows that these restrictions are not
accidental: both strict soundness and statistical security are
required for the result from \cite{qpok} to hold.
Protocols that are
secure classically but insecure in the quantum setting were constructed before:
\cite{zhandry12random} presented classically secure pseudorandom functions that become insecure 
when the adversary is not only quantum, but can also \emph{query the
  pseudorandom function in superposition}. Similarly for secret sharing schemes
\cite{damgaard13superposition} and one-time MACs \cite{boneh13quantummac}.
But, in all of these cases, the negative results are shown for the case when the adversary
is allowed to interact with the honest parties in superposition. Thus, the cryptographic protocol is different in the classical case and the quantum case.
In contrast, we keep the protocols the same, with only classical
communication and only change adversary's internal power (by allowing
it to be a polynomial-time quantum computer which may access quantum oracles).
We believe that this is the first such separation.
Boneh, Dagdelen, Fischlin, Lehmann, Schaffner, and Zhandry
\cite{boneh11quantumro} first showed how to correctly define the
random oracle in the quantum setting (namely, the adversary has to
have superposition access to it).  For the Fiat-Shamir construction
(using random oracles as modeled by \cite{boneh11quantumro}),
an impossibility result was given by Dagdelen, Fischlin, and 
Gagliardoni~\cite{dagdelen13fiat}. However,
their impossibility only shows that security of Fiat-Shamir cannot be
shown using extractors that do not perform quantum rewinding;\footnote{\label{revfootnote:rewind}They do
  allow extractors that restart the adversary with the same classical
  randomness from the very beginning. But due to the randomness
  inherent in quantum measurements, the adversary will then not
  necessarily reach the same state again. They also do not allow the
  extractor to use a purified (i.e., unitary) adversary to avoid
  measurements that introduce randomness.} but such quantum rewinding
is possible and used in the existing positive results from
\cite{watrous-qzk,qpok} which would also not work in a model without
quantum rewinding.
A\label{revpage:oblivious} \emph{variant of} Fiat-Shamir has been shown
  to be a quantum secure signature scheme \cite{dagdelen13fiat}.
  Probably their scheme can also be shown to be a quantum
  zero-knowledge proof of knowledge.\footnote{The unforgeability proof
    from \cite{dagdelen13fiat} is already almost a proof of the proof
    of knowledge property. And the techniques from \cite{qro-nizk} can
    probably be applied to show that the protocol form \cite{dagdelen13fiat} is
    zero-knowledge.}  However, their construction assumes
  sigma-protocols with ``oblivious commitments''.  These are a much
  stronger assumption that usual sigma-protocols: as shown in
  \cite[Appendix A]{qro-nizk}, sigma-protocols with oblivious
  commitments are by themselves already non-interactive zero-knowledge
  proofs in the CRS model (albeit single-theorem, non-adaptive ones). 
  \cite{qro-nizk} presents a non-interactive quantum zero-knowledge
  proof of knowledge in the random oracle model, based on arbitrary sigma-protocols
  (it does not even need strict soundness). That protocol uses
  ideas different from both Fiat-Shamir and Fischlin's scheme to avoid
  rewinding.

  It was known for a long time that it is difficult to use classical
  definitions for computational binding in the quantum setting
  ({\shortonly{\expandafter\def\csname cite@adjust\endcsname{\hspace*{.5pt}}}\cite{dumais00qbc}} is the first reference we are aware of), but
  none showed so far that the computational definition was truly
  insufficient.

\paragraph{Our technique.} The schemes we analyze are all based on
sigma-protocols which have the \emph{special soundness} property:
In a proof of a statement~$s$, 
given two accepting conversations $(\com,\ch,\resp)$ and
$(\com,\ch',\resp')$, one can efficiently extract a witness for~$s$. (The
\emph{commitment}\index{commitment} $\com$ and the
\emph{response}\index{response} $\resp$ are sent by the prover, and
the \emph{challenge}\index{challenge} $\ch$ by the verifier.)  In the
classical case, we can ensure that the prover cannot produce one
accepting conversation without having enough information to produce
two. This is typically proven by rewinding the prover to get two
conversations. So in order to break the schemes in the quantum case,
we need to give the prover some information that allows him to succeed
in one interaction, but not in two.

To do so, we use the following trick (we call it the \emph{pick-one
  trick}):%
\index{pick-one trick} Let $S$ be a set of values (e.g., accepting
conversations). Give the quantum state
$\ket{\Psi}:=\frac1{\sqrt{\abs{S}}}\sum_{x\in S}\ket x$ to the
adversary. Now the adversary can get a random $x\in S$ by measuring
$\ket{\Psi}$. However, on its own that is not more useful than just
providing a random $x\in S$. So in addition, we provide an oracle that
applies the unitary $\OF$ with $\OF\ket{\Psi}=-\ket{\Psi}$ and
$\OF\ket{\Psi^\bot}=\ket{\Psi^\bot}$ for all $\ket{\Psi^\bot}$
orthogonal to $\ket{\Psi}$. Now the adversary can use (a variant of)
Grover's search starting with state $\ket{\Psi}$ to find some $x\in S$
that satisfies a predicate $P(x)$ of his choosing, as long as
$\abs{S}/\abs{\{x\in S:P(x)\}}$ is polynomially bounded. Note however: once the
adversary did this, $\ket{\Psi}$ is gone, he cannot get a second
$x\in S$.

How do we use that to break proofs of knowledge? The simplest case is
attacking the sigma-protocol itself.  Assume the challenge space is
polynomial. (I.e., $\abs\ch$ is logarithmic.) Fix a commitment $\com$,
and let $S$ be the set of all $(\ch,\resp)$ that form an accepting
conversation with $\com$. Give $\com$ and $\ket{\Psi}$ to the malicious
prover.  (Actually, in the full proof we provide an oracle $\Opsi$
that allows us to get $\ket{\Psi}$ for a random $\com$.)  He sends
$\com$ and receives a challenge $\ch'$. And using the pick-one trick, he
gets $(\ch,\resp)\in S$ such that $\ch=\ch'$. Thus sending $\resp$
will make the verifier accept.

This in itself does not constitute a break of the protocol. A
malicious prover is allowed to make the verifier accept, as long as he
knows a witness. Thus we need to show that even given $\ket{\Psi}$ and
$\OF$, it is hard to compute a witness. Given two accepting
conversations $(\com,\ch,\resp)$ and $(\com,\ch',\resp')$ we can
compute a witness. So we need that given $\ket{\Psi}$ and $\OF$, it is
hard to find two different $x,x'\in S$.  We show this below
(under certain assumptions on the size of $S$, see
\autoref{theo:pickone.sound}, \autoref{coro:pick.one.sound-oall}). Thus the
sigma-protocol is indeed broken: the malicious prover can make the
verifier accept using information that does not allow
him to compute a witness. (The full counterexample will need
additional oracles, e.g., for membership test in $S$ etc.)
Counterexamples for the other constructions (Fiat-Shamir, Fischlin,
etc.) are constructed similarly. We stress that this does not contradict the security of
sigma-protocols with strict soundness \cite{qpok}. Strict soundness
implies that there is only one response per challenge. Then $\abs S$
is polynomial and it becomes possible to extract two accepting
conversations from $\ket{\Psi}$ and $\OF$.

The main technical challenge is to prove that given $\ket{\Psi}$ and
$\OF$, it is hard to find two different $x,x'\in
S$.
This is done using the representation-theoretic form of ``quantum adversary" lower bound method for quantum algorithms \cite{Ambainis10,AMRR}. The method is based on viewing a quantum algorithm as a sequence of transformations on a bipartite quantum system that consists of two registers: one register $\calH_A$ that contains the algorithm's quantum state and another register $\calH_I$ that contains the information which triples $(com, ch, resp)$ belong to $S$. The algorithm's purpose is to obtain two elements $x_1, x_2\in S$ using only a limited type of interactions betweeen $\calH_A$ and $\calH_I$.
(From a practical perspective, a quantum register $\calH_I$ holding the membership
information about $S$ would be huge.
However, we do not propose to implement such a register. Rather, we use it as a tool
to prove a lower bound which then implies a corresponding lower bound in the usual model where $S$ is accessed via oracles.)

We then partition the state-space of $\calH_I$ into subspaces corresponding to group representations of the symmetry group of $\calH_I$ (the set of
all permutations of triples $(com, ch, resp)$ that satisfy some natural requirements).
Informally, these subspaces correspond to possible states of algorithm's knowledge
about the input data: having no information about any $s\in S$, knowing one value $x\in S$, knowing two values $x_1, x_2\in S$ and so on.

The initial state in which the algorithm has $\ket{\Psi}$ corresponds to $\calH_I$ being in the state ``the algorithm knows one $x\in S$". (This is very natural because measuring $\ket{\Psi}$ gives one value $x\in S$ and there is no way to obtain two values $x\in S$ from this state with a non-negligible probability.) We then show that each application of the available oracles (such as $O_F$ and the membership test for $S$) can only move a tiny part of the state in $\calH_I$ from the ``the algorithm knows one $x\in S$" subspace of $\calH_I$ to the ``the algorithm knows two $x\in S$" subspace. Therefore, to obtain two values $x_1, x_2\in S$, we need to apply the available oracles
a large number of times.

While the main idea is quite simple, implementing it requires a sophisticated analysis of the representations of the symmetry group of $\calH_I$ and how they evolves when
the oracles are applied.

Actually, below we prove an even stronger result: We do not wish to
give the state $\ket{\Psi}$ as input to the adversary. (Because that
would mean that the attack only works with an input that is not
efficiently computable, even in our relativized model.) Thus, instead,
we provide an oracle $\Opsi$ for efficiently constructing this
state. But then, since the oracle can be invoked arbitrarily many
times, the adversary could create two copies of $\ket{\Psi}$, thus
easily obtaining two $x,x'\in S$! Instead, we provide an oracle
$\Opsi$ that provides a state $\ketpsi$ which is a superposition of
many $\ket\Psi=\ketpsiy y$ for independently chosen sets
$S_y$. Now the adversary can produce $\ketpsi$ and using a measurement
of $y$, get many states $\ketpsiy y$ for random $y$'s, but no two
states  $\ketpsiy y$ for the same $y$. Taking these additional
capabilities into account complicates the  proof further, as does the
presence of additional oracles that are needed, e.g., to construct the
prover (who does need to be able to get several $x\in S$).

\paragraph{On the meaning of oracle separations.}\pagelabel{revpage:oracles}
  At this point, we should say a few words about
  what it implies that our impossibility results are relative to a
  certain oracle.  Certainly, our results do not necessarily imply
  that the investigated schemes are insecure or unprovable in the
  ``real world'', i.e., without oracles. However, our results give a
  number of valuable insights.  Foremost, they tell us which proof
  techniques cannot be used for showing security of those schemes:
  only non-relativizing proofs can work. This cuts down the search
  space for proofs considerable. Also, it shows that security proofs
  would need new techniques; the proof techniques from
  \cite{watrous-qzk,qpok} at least are relativizing. And even
  non-relativizing proof techniques such as (in the classical setting)
  \cite{barak01bb} tend to use specially designed (and more
  complicated) protocols than their relativizing counterparts, so our
  results might give some evidence that the specific protocols we
  investigate here have no proofs at all, whether relativizing or
  non-relativizing.  Furthermore, oracle-based impossibilities can
  give ideas for non-oracle-based impossibilities. If we can find
  computational problems that exhibit similar properties as our
  oracles, we might get analogous impossibilities without resorting to
  oracles (using computational assumptions instead).\footnote{For
    example, \cite{aaronson12money} 
    presents a construction that might
    allow to implement an analogue to the oracle $\OF$. Essentially,
    if the set $S$ (called $A$ in \cite{aaronson12money}) is a linear
    code, then they give a candidate for how to obfuscate $\OF$
    (called $V_A$ in \cite{aaronson12money}) such that one can apply
    $\OF$ but does not learn $A$.  Of course, this does not give us a
    candidate for how to construct the other oracles needed in this
    work, but it shows that the idea of actually replacing our custom
    made oracles by computational assumptions may not be far fetched.
  } However, we should stress that even if we get rid of the oracles,
  our results do not state that \emph{all} sigma-protocols lead to
  insecure schemes. It would not be excluded that, e.g., the
  graph-isomorphism sigma-protocol \cite{Goldreich:1991:Proofs} is
  still a proof of knowledge. What our approach aims to show is the impossibility of
  \emph{general} constructions that are secure for \emph{all}
  sigma-protocols.

  Finally, we mention one point that is important in general when
  designing oracle separations in the quantum world: even relative to
  an oracle, the structural properties of quantum circuits should not
  change. For example, any quantum algorithm (even one that involves intermediate 
  measurements or other non-unitary operations)
  can be replaced by a
  unitary quantum circuit, and that unitary circuit can be reversed. If we choose
  oracles that are not reversible, then
  we lose this property. (E.g.,
  oracles that perform measurements 
  or that perform random choices are non-reversible.)
  So an impossibility result based on such
  oracles would only apply in a world where quantum circuits are not
  reversible. Thus for meaningful oracle separations, we need to ensure
  that: (a) all oracles are unitary, and (b) all oracles have
  inverses. This makes some of the definitions of oracles in our work
  (\autoref{def:ora.dist}) more involved than would be necessary if we
  had used non-unitary oracles.

\fullonly{}

\begin{shortversion}
  \paragraph{Organization.} 
  \autoref{sec:sec.def} introduces security definitions.
  \autoref{sec:pickone} develops the pick-one trick. \autoref{sec:com}
  shows the insecurity of computationally binding commitments,
  \autoref{sec:attack.sigma} that of sigma-protocols,
  \autoref{sec:fiat} that of the Fiat-Shamir construction, and
  \autoref{sec:fischlin} that of Fischlin's construction.  Additional
  details and full proofs are given in the full version~\cite{qpok-imposs-fv}.
\end{shortversion}

\begin{fullversion}

\section{Preliminaries}
\label{sec:preliminaries}

\paragraph{Security parameter.} As usual in cryptography, we assume
that all algorithms are parametric in a \emph{security parameter}
$\secpar$\symbolindexmark{\secpar}%
\index{security parameter}. Furthermore, parameters of said algorithms
can also implicitly depend on the security parameter. E.g., if we say
``Let $\ell$ be a superlogarithmic integer. Then $A(\ell)$ runs in
polynomial time.'', then this formally means ``Let $\ell$ be a
superlogarithmic function. Then the running time of
$A(\secpar,\ell(\secpar))$ is a polynomially-bounded function of~$\secpar$.''

\paragraph{Misc.} $x\otR M$\symbolindexmark{\otR} means that $x$ is
uniformly randomly chosen from the set $M$.
$x\ot A(y)$\symbolindexmark{\ot} means that $x$ is assigned the
classical output of the (usually probabilistic or quantum) algorithm
$A$ on input $y$.

{%
\paragraph{Quantum mechanics.} For space reasons, we cannot give an
introduction to the mathematics of quantum mechanics used here. We
refer the reader to, e.g., \cite{nielsenchuang-10year}. A
\emph{quantum state}%
\index{quantum state}\index{state!quantum} is a vector of norm~$1$ in
a Hilbert space, written $\ket\Psi$.\symbolindexmark{\ket} Then
$\bra\Psi$\symbolindexmark{\bra} is its
dual. $\TD(\rho,\rho')$\symbolindexmark{\TD} denotes the \emph{trace
  distance}%
\index{trace distance}%
\index{distance!trace} between mixed states $\rho,\rho'$. We write
short $\TD(\ket{\Psi},\ket{\Psi'})$ for
$\TD(\selfbutter\Psi,\selfbutter{\Psi'})$.
$\SD(X;Y)$\symbolindexmark{\SD} in contrast is the \emph{statistical
  distance}%
\index{statistical distance}%
\index{distance!statistical} between random variables $X$ and $Y$.
}

\paragraph{Oracles.} We make heavy use of oracles in this
paper. Formally, an oracle $\calO$\symbolindexmark{\calO} is a unitary
transformation on some Hilbert space $\calH$. An oracle algorithm $A$ with
access to $\calO$ (written $A^{\calO}$) is then a quantum algorithm
which has a special gate for applying the unitary $\calO$.  $\calO$
may depend on the security parameter.  $\calO$ may be probabilistic in
the sense that at the beginning of the execution, the unitary $\calO$
is chosen according to some distribution
(like the random oracle in cryptography).
However, $\calO$ may not be
probabilistic in the sense that $\calO$, when queried on the same
value twice, gives two different random answers (like an encryption
oracle for a probabilistic encryption scheme would). Such a behavior
would be difficult to define formally when allowing queries to $\calO$
in superposition. When defining $\calO$, we use the shorthand
$\calO(x):=f(x)$ to denote
$\calO\ket{x,y}:=\calO\ket{x,y\oplus f(x)}$. We call an oracle of this
form classical. Our classical algorithms will only access oracles of
this form. We stress that even for a classical oracle $\calO$, a
quantum algorithm can query $\calO(x)$ in superposition of different $x$.
We often give access to several oracles $(\calO_1,\calO_2,\dots)$ to
an algorithm. This can be seen as a specific case of access to a
single oracle by setting
$\calO\ket i\ket\Psi:=\ket i\otimes\calO_i\ket\Psi$.

In our setting, oracles are used to denote a relativised world in
which those oracles happen to be efficiently computable. If a unitary
$U$ is implemented by an efficient quantum circuit, $U^\dagger$ can
also be implemented by an efficient quantum circuit. We would expect
this also to hold in a relativised setting. Thus for any oracle
$\calO$, algorithms should have access to their inverses, too. In our
work this is ensured because all oracles defined here are self-inverse
($\calO=\calO^\dagger$).
\end{fullversion}

\fullshort\subsection\section{Security definitions}
\label{sec:sec.def}

A \emph{sigma-protocol}%
\index{sigma-protocol} for a relation $R$ is a three message proof
system. It is described by the lengths $\ellcom,\ellch,\ellresp$%
\symbolindexmark{\ellcom}\symbolindexmark{\ellch}\symbolindexmark{\ellresp}
of the messages, a polynomial-time prover $(P_1,P_2)$ and a
polynomial-time verifier $V$. The first
message from the prover is $\com\ot P_1(s,w)$\symbolindexmark{\com} with $(s,w)\in R$ and is called
\emph{commitment}%
\index{commitment!(in sigma-protocol)}, the uniformly random reply
from the verifier is $\ch\otR\bits\ellch$\symbolindexmark{\ch} (called \emph{challenge}%
\index{challenge!(in sigma-protocol)}), and the prover answers with
$\resp\ot P_2(\ch)$\symbolindexmark{\resp} (the \emph{response}%
\index{response!(in sigma-protocol)}). We assume $P_1,P_2$ to share state. Finally
$V(s,\com,\ch,\resp)$ outputs whether the verifier accepts.

We will make use of the following standard properties of
sigma-protocols. Note that we have chosen to make the definition
stronger by requiring honest entities (simulator, extractor) to be
classical while we allow the adversary to be quantum.
\begin{definition}[Properties of sigma-protocols]\label{def:sigma.props} Let
  $(\ellcom,\ellch,\ellresp,P_1,P_2,V,R)$ be a sigma-protocol. We define:
  \begin{compactitem}
  \item \textbf{Completeness:}\index{completeness!(of sigma-protocol)}
    For all $(s,w)\in R$,
    $\Pr[\ok=0:\com\ot P_1(s,w),
    \ch\otR\bits\ellch, \resp\ot P_2(\ch), \ok\ot V(s,\com,\ch,\resp)]$
    is negligible.

    \fullonly{(Intuitively: an honestly generated proof succeeds for
      overwhelming probability.)}
  \item \textbf{Perfect special soundness:}%
    \index{special soundness!(of sigma protocol)}%
    \index{soundness!special (of sigma protocol)}
    There is a polynomial-time classical algorithm
    $\ESigma$\symbolindexmark{\ESigma} (the extractor) such that for any
    $(s,\com,\ch,\resp,\ch',\resp')$ with $\ch\neq\ch'$,
    we have that $\Pr[(s,w)\notin R
    \land\ok=\ok'=1:
    \ok\ot V(s,\com,\ch,\resp),\ok'\ot
    V(s,\com,\ch',\resp'),
    w\ot \ESigma(s,\com,\ch,\resp,\ch',\resp')]=0$.

    \fullonly{(Intuitively: given two valid interactions with the same
      commitment, one can efficiently extract a witness.)}
  \item \textbf{Computational special soundness:}%
    \index{computational special soundness!(of sigma-protocol)}%
    \index{special soundness!computational (of sigma-protocol)}%
    \index{soundness!computational special (of sigma-protocol)}
    There is a polynomial-time classical algorithm
    $\ESigma$ (the extractor)  such that for any
    polynomial-time quantum algorithm~$A$ (the adversary), we have that $\Pr[(s,w)\notin R\land\ch\neq\ch'\land\ok=\ok'=1:\penalty0
    (s,\com,\ch,\resp,\ch',\resp')\ot A,\penalty0 \ok\ot V(s,\com,\ch,\resp),\penalty0 \ok'\ot
    V(s,\com,\ch',\resp'),\penalty0 w\ot \ESigma(s,\com,\ch,\resp,\ch',\resp')]$ is negligible.

    \fullonly{(Intuitively: given two valid interactions with the same
      commitment chosen by a polynomial-time adversary, one can
      efficiently extract a witness with overwhelming probability.)}
  \item \textbf{Statistical honest-verifier zero-knowledge (HVZK):}%
    \index{honest-verifier zero-knowledge!(of sigma-protocol)}%
    \index{zero-knowledge!honest-verifier (of sigma-protocol)}%
    \index{HVZK|see{honest-verifier zero-knowledge}}\footnote{In the
      context of this
      paper, HVZK is equivalent to zero-knowledge because our
      protocols have logarithmic challenge length $\ellch$ \cite{watrous-qzk}.}
    There
    is a polynomial-time classical algorithm $\SSigma$\symbolindexmark{\SSigma} (the simulator) such that for
    any (possibly unlimited) quantum algorithm $A$ and all $(s,w)\in
    R$, the following is negligible:
    \begin{alignat*}2
      \babs{
        &\Pr[b=1: \com\ot P_1(s,w), \ch\otR\bits\ellch, 
        \shortonly{\\&\qquad\quad}
        \resp\ot P_2(\ch),
        b\ot A(\com,\ch,\resp)]\\
        -
        &\Pr[b=1:  (\com,\ch,\resp)\ot S(s), 
        \shortonly{\\&\qquad\quad}
        b\ot A(\com,\ch,\resp)]
      }
    \end{alignat*}

    \fullonly{(Intuitively: An interaction between honest verifier and honest
      prover can be simulated in polynomial-time without knowing the
      witness.)}
  \item \textbf{Strict soundness:}%
    \index{strict soundness}%
    \index{soundness!strict}
    For any $(s,\com,\ch)$ and any $\resp\neq\resp'$ we have
    $\Pr[\ok=\ok'=1:\ok\ot V(s,\com,\ch,\resp),\ok'\ot V(s,\com,\ch,\resp')]=0$.

    \fullonly{(Intuitively: Given the commitment and the challenge, there is at most one possible accepted response.)}
  \item \textbf{Computational strict soundness:}%
    \index{computational strict soundness}%
    \index{strict soundness!computational}%
    \index{soundness!computational strict}\footnote{Also known as
      \emph{unique responses}%
      \index{unique responses}%
      \index{responses!unique} in \cite{fischlin05online}.}
    For any polynomial-time quantum algorithm $A$ (the adversary), we have that
    $\Pr[\ok=\ok'=1\land\resp\neq\resp':
    (s,\com,\ch,\resp,\resp')\ot A,
    \ok\ot V(s,\com,\ch,\resp),\ok'\ot V(s,\com,\ch,\resp')]$ is negligible.

    \fullonly{(Intuitively: Given the commitment and the challenge, it is
      computationally hard to find more than one accepting response.)}
  \item \textbf{Commitment entropy:}%
    \index{commitment entropy}%
    \index{entropy!commitment}
    For all $(s,w)\in R$ and $\com\ot P_1(s,w)$, the min-entropy of
    $\com$ is superlogarithmic.

        \fullonly{(Intuitively: the commitment produced by the prover cannot be guessed with more than negligible probability.)}
  \end{compactitem}
  In a relativized setting, all quantum algorithms additionally get
  access to all oracles, and all classical algorithms additionally get
  access to all classical oracles.
  \tqed
\end{definition}

In this paper, we will mainly be concerned with proving that certain
schemes are \emph{not} proofs of knowledge. Therefore, we
will not need to have precise definitions of these concepts; we only
need to know what it means to break them.

\begin{definition}[Total breaks]\label{def:total.break}
  Consider an interactive or non-interactive proof system $(P,V)$ for
  a relation $R$. 
  Let $\LR:=\{s:\exists w.(s,w)\in R\}$\symbolindexmark{\LR} be the language defined by $R$.
  A \emph{total break}%
  \index{total break}%
  \index{break!total}
  is a 
  polynomial-time quantum algorithm $A$ such that the following
  probability is overwhelming:
  \[
  \Pr[\ok=1\ \land\ s\notin L_R:
  s\ot A,
  \ok\ot\mpair A{V(s)}
  ]
  \]
  Here $\mpair A{V(s)}$\symbolindexmark{\mpair} denotes the output of $V$ in an interaction
  between $A$ and $V(s)$.
  (Intuitively, the adversary performs a total break if the adversary manages with overwhelming probability
    to convince the verifier $V$ of a statement $s$ that is not in the language $L_R$.)

  A \emph{total knowledge break}%
  \index{total knowledge break}%
  \index{total break!knowledge}
  is a polynomial-time quantum algorithm $A$ such that for all
  polynomial-time quantum algorithms $E$ we have that:
  \begin{compactitem}
  \item Adversary success:%
    \index{adversary success} $\Pr[\ok=1:s\ot A,\ok\ot\mpair A{V(s)}]$ is overwhelming.
  \item Extractor failure:%
    \index{extractor failure} $\Pr[(s,w)\in R:s\ot A,w\ot E(s)]$ is negligible.
  \end{compactitem}
  Here $E$ has access to the final state of $A$.
  (Intuitively, the adversary performs a total knowledge break if the adversary manages with overwhelming probability
    to convince the verifier $V$ of a statement $s$, but the extractor $E$ cannot extract a witness $w$ for that statement.)

  When applied to a proof system relative to an oracle $\calO$,
    both $A$ and $E$ get access to~$\calO$. In settings where $R$ and
    $\calO$ are probabilistic, the probabilities are averaged over all
    values of $R$ and $\calO$. \tqed
\end{definition}

Note that these definitions of attacks are quite strong. In
particular, $A$ does not get any auxiliary state.  And $A$ needs to
succeed with overwhelming probability and make the extraction fail
with overwhelming probability. (Usually, proofs / proofs of knowledge
are considered broken already when the adversary has non-negligible
success probability.) Furthermore, we require $A$ to be
polynomial-time.

In particular, a total break implies that a proof system is neither
a proof nor an argument. And total knowledge break implies that it is
neither a proof of knowledge nor an argument of knowledge, with
respect to all definitions the authors are aware
of.\footnote{Definitions that would not be covered would be such
  where the extractor gets additional auxiliary input not available to
  the adversary. We are, however, not aware of such in the literature.}

\begin{fullversion}
  
\section{State creation oracles}
\label{sec:opsi}

We first show a result that shows that having access to an oracle
$\Opsi$ for creating copies of an unknown state $\ket\Psi$ is not more
powerful than having access to a \emph{reservoir state}%
\index{reservoir state} $\ket R$ of polynomially-many copies of~$\ket\Psi$ (some
of them in superposition with a fixed state $\ket\bot$). 
(Such an oracle is, in our setting, implemented as
  $\Opsi\ket\Psi=\ket\bot$,
  $\Opsi\ket\bot=\ket\Psi$, and is the identity on states orthogonal to $\ket\bot,\ket\Psi$.)
We\pagelabel{revpage:reservoir} will need this later, because it allows us to assume in our proofs 
that the adversary has access to such a reservoir state instead 
of access to the oracle $\Opsi$. It turns out to be much easier to 
show that those reservoir states do not help the adversary in solving the 
Two Values problem than it is to deal directly with $\Opsi$ in the proof.

Note that
the fact that $\Opsi$ is no more powerful than $\ket R$ is not immediate: 
 $\Opsi$ can be queried in
superposition, and its inverse applied; this might give more power
than copies of the state $\ket\Psi$. In fact, we know of no way to generate,
e.g., $\fsq\ket\Psi+\fsq\ket\bot$ for a given (known) state $\ket\bot$ and 
unknown~$\ket\Psi$, even given many copies of $\ket\Psi$ (unless we have
enough copies of $\ket\Psi$ to determine a complete description of
$\ket\Psi$ by measuring). Yet $\fsq\ket\Psi+\fsq\ket\bot$ can be
generated with a single query to $\Opsi$.\footnote{For example, one can initialize 
a register with $\fsq\ket\bot+\fsq\ket0$ where $\ket0$ is any fixed state guaranteed 
to be (almost) orthogonal to $\ket\bot$ and $\ket\Psi$.
Applying $\Opsi$ yields $\fsq\ket\Psi+\fsq\ket0$.
Finally, by applying the fixed (and thus known) 
unitary $U:\ket0\mapsto\ket\bot$, 
we get $\fsq\ket\Psi+\fsq\ket\bot$.}
 This is why our reservoir
$\ket R$ has to contain such superpositions in addition to pure states $\ket\Psi$.
\begin{theorem}[Emulating state creation oracles]\label{theo:emulate.opsi}
  Let $\ket\Psi$ be a state, chosen according to some
  distribution. Let $\ketbot$ be a fixed state orthogonal to
  $\ket\Psi$. (Such a state can always be found by extending the
  dimension of the Hilbert space containing $\ket\Psi$ and using the
  new basis state as $\ketbot$.)
  Let $\Opsi$ be an oracle with $\Opsi\ket\Psi=\ket\bot$,
  $\Opsi\ket\bot=\ket\Psi$, and $\Opsi\ket{\Psi^\bot}=\ket{\Psi^\bot}$
  for any $\ket{\Psi^\bot}$ orthogonal to both $\ket\Psi$ and
  $\ket\bot$. 
  Let $\calO$ be an
  oracle, not necessarily independent of $\ket\Psi$. Let $\ket\Phi$ be
  a quantum state, not necessarily independent of $\ket\Psi$. Let
  $n,m\geq 0$ be integers. Let
  $\ket R:=\ket\Psi^{\otimes m}\otimes\ket{\alpha_1}\otimes\dots\otimes \ket{\alpha_n}$ where
  $\ket{\alpha_j}:=(\cos\frac{j\pi}{2n})\ket\Psi+(\sin\frac{j\pi}{2n})\ketbot$.

  Let $A$ be an oracle algorithm that makes $q_\Psi$ queries to
  $\Opsi$.
  Then there is an oracle algorithm $B$ that
  makes the same number of queries to $\calO$ as $A$ such that:
  \begin{multline*}
    \TD\bigl(
      B^{\calO}(\ket R,\ket\Phi)
      ,
      A^{\Opsi,\calO}(\ket\Phi)
      \bigr) \\
    \leq   \quad
     \frac{\pi q_\Psi}{2\sqrt n} + q_\Psi\,o(\tfrac1{\sqrt n})
    +
    \frac{2q_\Psi}{\sqrt{m+1}}
    \quad\leq\quad
    O\Bigl(\frac{q_\Psi}{\sqrt n}+\frac{q_\Psi}{\sqrt m}\Bigr).
  \end{multline*}
\end{theorem}

The idea behind this lemma is the following: To implement $\Opsi$, we
need a way to convert $\ketbot$ into $\ket\Psi$ and vice versa. At the
first glance this seems easy: If we have a reservoir $R$ containing
$\ket\Psi^{\otimes n}$ for sufficiently large $n$, we can just take a
new $\ket\Psi$ from $R$. And when we need to destroy $\ket\bot$, we
just move it into $R$. This, however, does not work because the
reservoir $R$ ``remembers'' whether we added or removed $\ket\Psi$
(because the number of $\ket\Psi$'s in $R$ changes). So if we apply
$\Opsi$ to, e.g., $\fsq\ket\Psi+\fsq\ket0$, the reservoir $R$
essentially acts like a measurement whether we applied $\Opsi$ to
$\ket\Psi$ or $\ket0$. 

To avoid this, we need a reservoir $R$ in a state that does not change
when we add $\ket\Psi$ or $\ket\bot$ to the reservoir. Such a state
would be
$\ket{R^\infty}:=\ket\Psi^{\otimes\infty}\otimes\ketbot^{\otimes\infty}$. If
we add or remove $\ket\Psi$ to an infinite state
$\ket\Psi^{\otimes\infty}$, that state will not change. Similarly for
$\ketbot$. (The reader may be worried here whether an infinite tensor
product is mathematically well-defined or physically meaningful. We do
not know, but the state $\ket{R^\infty}$ is only used for motivational
purposes, our final proof only uses finite tensor products.)

Thus we have a unitary operation $S$ such that
$S\ketbot\ket{R^\infty}=\ket\Psi\ket{R^\infty}$. Can we use this
operation to realize $\Opsi$? Indeed, an elementary calculation
reveals that the following circuit implements $\Opsi$ on $X$ when
$R,Z$ are initialized with $\ket{R^\infty},\ket0$.

\begin{align}
  \begin{gathered}\Qcircuit @R=.4em @C=.5em {
      \lstick{X}&\qw    &\gate{U_\bot}&\qw    &\multigate{1}{S}&\qw    &\gate{\ORef}&\qw    &\multigate{1}{S^\dagger}&\qw    &\gate{U_\bot}& \qw   &\qw\\
      \lstick{R}&\qw    &\qw          &\qw    &\ghost{S}       &\qw    &\qw         &\qw    &\ghost{S^\dagger}       &\qw    &\qw          & \qw   &\qw\\
      \lstick{Z}&\gate H&\ctrl{-2}    &\gate H&\ctrl{-1}       &\gate H&\ctrl{-2}   &\gate H&\ctrl{-1}               &\gate H&\ctrl{-2}    &\gate H&\qw
    }
  \end{gathered}
  \label{eq:opsi.circ}
  & 
  \\[5pt] 
  &\hskip-.7in
  \text{with }U_\bot:=1-2\selfbutter\bot\notag \\
  &\hskip-.7in
  \text{and }\ORef:=1-2\selfbutter\Psi
\end{align}

Note that we have introduced a new oracle $\ORef$ here. We will deal
with that oracle later.

Unfortunately, we cannot use $\ket{R^\infty}$. Even if such a state
should be mathematically well-defined, the algorithm $B$ cannot perform the
infinite shift needed to fit in one more $\ket\Psi$ into
$\ket{R^\infty}$. The question is, can $\ket{R^\infty}$ be
approximated with a finite state? I.e., is there a state $\ket R$ such
that $S\ketbot\ket R\approx \ket\Psi\ket R$ for a suitable $S$?
Indeed, such a state exists, namely the state $\ket R$ from
\autoref{lemma:emulate.opsi}. For sufficiently large $n$, the
beginning of $\ket R$ is approximately
$\ket\Psi\otimes\ket\Psi\otimes\ket\Psi\otimes\dots$, while the tail
of $\ket R$ is approximately
$\dots\otimes\ketbot\otimes\ketbot\otimes\ketbot$. In between, there
is a smooth transition. If $S$ adds $\ket\bot$ to the end and removes
$\ket\Psi$ from the beginning of $\ket R$, the state still has
approximately the same form (this needs to be made quantitative, of
course). That is, $S$ is a cyclic left-shift on $\ketbot\ket R$.

Hence $\ket R$ is a good approximate drop-in replacement for
$\ket{R^\infty}$, and the circuit \eqref{eq:opsi.circ} approximately realizes $\Opsi$ when
$R,Z$ are initialized with $\ket R,\ket 0$. 

However, we now have introduced the oracle $\ORef$.
We need to show how to emulate that oracle:
$\ORef$ essentially
implements a measurement whether a given state $\ket\Phi$ is
$\ket\Psi$ or orthogonal to $\ket\Psi$. Thus to implement $\ORef$, we
need a way to test whether a given state is $\ket\Psi$ or not. The
well-known swap test \cite{buhrman01finger} is not sufficient, because for $\ket\Phi$ orthogonal to
$\ket\Psi$, it gives an incorrect answer with probability $\frac12$
and destroys the state. Instead, we use the following test that has an
error probability $O(1/m)$ given $m$ copies of $\ket\Psi$ as
reference: Let $\ket T:=\ket\Psi^{\otimes m}$. Let $V$ be the space of
all $(m+1)$-partite states that are invariant under
permutations. $\ket\Psi\ket T$ is such a state, while for $\ket{\Phi}$
orthogonal to $\ket\Psi$, $\ket{\Phi}\ket T$ is almost orthogonal to
$V$ for large $m$ (up to an error of $O(1/m)$). So by measuring
whether $\ket\Phi\ket T$ is in $V$, we can test whether $\ket\Phi$ is
$\ket\Psi$ or not (with an error $O(1/m)$), and when doing so the
state $\ket T$ is only disturbed by $O(1/m)$. We can thus simulate any
algorithm that uses $\ORef$ up to any inversely polynomial precision
using a sufficiently large state $\ket T$. 

We then get \autoref{theo:emulate.opsi} by extending the state $\ket
R$ to also contain $\ket T$.

\medskip

Formally, the theorem is an immediate consequence of
\multiautoref{lemma:emulate.opsi,lemma:emulate.oref} 
\fullshort{in Appendix~\ref{app:lemmas:opsi}}{below}.

\delaytext{lemmas opsi}{
\begin{lemma}\label{lemma:emulate.opsi}
  Let $\ket\Psi$ be a state, chosen according to some
  distribution. Let $\ketbot$ be a fixed state orthogonal to
  $\ket\Psi$. (Such a state can always be found by extending the
  dimension of the Hilbert space containing $\ket\Psi$ and using the
  new basis state as $\ketbot$.)

  Let $\Opsi$ be an oracle with $\Opsi\ket\Psi=\ket\bot$,
  $\Opsi\ket\bot=\ket\Psi$, and $\Opsi\ket{\Psi^\bot}=\ket{\Psi^\bot}$
  for any~$\ket{\Psi^\bot}$ orthogonal to both $\ket\Psi$ and
  $\ket\bot$. 
  Let $\ORef:=I-2\selfbutter\Psi$.\symbolindexmark{\ORef}

  Let $\calO$ be an
  oracle, not necessarily independent of $\ket\Psi$. Let $\ket\Phi$ be
  a quantum state, not necessarily independent of $\ket\Psi$. 

  Let $n\geq0$ be an integer.  Let
  $\ket R:=\ket{\alpha_1}\otimes\dots \otimes\ket{\alpha_n}$ where
  $\ket{\alpha_j}:=(\cos\frac{j\pi}{2n})\ket\Psi+(\sin\frac{j\pi}{2n})\ketbot$.

  Then there is an oracle algorithm $B$ that makes $q_\Psi$
  queries to $\ORef$ and makes the same number of queries to $\calO$
  as $A$ such that:
  \[
  \TD\bigl(
  B^{\ORef,\calO}(\ket R,\ket\Phi) 
  ,
  A^{\Opsi,\calO}(\ket\Phi)
  \bigr)
  \leq
   \tfrac{\pi q_\Psi}{2\sqrt n} + q_\Psi\,o(\tfrac1{\sqrt n}).
  \]
\end{lemma}

\begin{proof} In this proof, we use
  the following shorthand notation: $\ket{\Phi}=\ket{\Phi'}\pm\varepsilon$ means
  that $\bnorm{\ket{\Phi}-\ket{\Phi'}}\leq\varepsilon$.

 We first show that
  \begin{equation}
    S\ketbot\ket R=\ket\Psi\ket R\pm \varepsilon_n
\quad \text{and} \quad
    S^\dagger\ket\Psi\ket R=\ketbot\ket R\pm \varepsilon_n
    \quad\text{with}\quad
    \varepsilon_n:=\tfrac{\pi}{2\sqrt n}+o(\tfrac1{\sqrt n})
    \label{eq:SR.quality}
  \end{equation}
  where
  $S\ket{\Phi_0}\ket{\Phi_1}\dots\ket{\Phi_n}:=\ket{\Phi_1}\dots\ket{\Phi_n}\ket{\Phi_0}$ (cyclic shift)
  and $\ket R$ is as in the statement of the lemma (the reservoir state).

  We have
  \begin{align*}
    (S\ketbot\ket R)^\dagger(\ket\Psi\ket R)
    &= \Bigl(\ket{\alpha_1}\ket{\alpha_2}\dots\ket{\alpha_{n}}\ketbot\Bigr)^\dagger\Bigl(\ket\Psi\ket{\alpha_1}\dots\ket{\alpha_{n-1}}\ket{\alpha_n}\Bigr) \\
    &= \inner{\alpha_1}\Psi\cdot\prod_{j=1}^{n-1} \inner{\alpha_{j+1}}{\alpha_j}\cdot\inner\bot{\alpha_n} \\
    &\starrel= \cos\tfrac\pi{2n} \cdot \prod_{j=1}^{n-1} \cos\bigl(\tfrac{(j+1)\pi}{2n}-\tfrac{j\pi}{2n}\bigr)\cdot \sin\tfrac{n\pi}{2n} 
    = (\cos\tfrac{\pi}{2n})^n.
  \end{align*}
  Here $(*)$ uses that $\ket\Psi$ and $\ket\bot$ are orthogonal (and
  the definition of $\ket{\alpha_j}$ from the statement of the lemma).
  For any quantum states $\ket\Phi,\ket{\Phi'}$ we have
  $\bnorm{\ket\Phi-\ket{\Phi'}}^2=(\ket\Phi-\ket{\Phi'})^\dagger
  (\ket\Phi-\ket{\Phi'})
  =1-\inner\Phi{\Phi'}-\inner{\Phi'}\Phi+1=2(1-\Re\inner\Phi{\Phi'})$
  where $\Re$ denote the real part. Thus
  $\norm{S\ketbot\ket R - \ket\Psi\ket R} \leq
  \sqrt{2(1-(\cos\tfrac{\pi}{2n})^n)}\in \frac{\pi}{2\sqrt
    n}+o(\tfrac1{\sqrt n})=\varepsilon_n$.
  (The asymptotic bound uses \autoref{lemma:cos}.)
  This shows
  the lhs of \eqref{eq:SR.quality}. The rhs follows from the rhs by
  applying the unitary $S^\dagger$ on both sides.

  \medskip

  Let $U_\Psi$ denote the unitary computed by circuit
  \eqref{eq:opsi.circ} on \autopageref{eq:opsi.circ}. We will show that for any~$\ket\Phi$, 
  \begin{equation}
    U_\Psi\ket\Phi\ket R\ket 0 = (\Opsi\ket\Phi)\ket R\ket 0 \pm \varepsilon_n.
    \label{eq:upsi.emul}
  \end{equation}
  By linearity of $U_\Psi,\Opsi$ and the triangle inequality, it is
  sufficient to verify this for $\ket\Phi=\ket\Psi$,
  $\ket\Phi=\ketbot$, and $\ket\Phi$ orthogonal to both
  $\ket\Psi,\ketbot$. In an execution of circuit \eqref{eq:opsi.circ} on state $\ket\Phi\ket R\ket 0$, we denote the
  state before $S$ with $\ket{\Phi_1}$, the state after $S$ with
  $\ket{\Phi_2}$, the state before $S^\dagger$ with $\ket{\Phi_3}$,
  and the state after $S^\dagger$ with $\ket{\Phi_4}$. We denote the
  final state with $\ket{\Phi'}=U_\Psi\ket\Phi\ket R\ket 0$.

  For $\ket\Phi=\ket\Psi$, we have
  \begin{alignat*}2
    \ket{\Phi_1} &= \ket\Psi\ket R\ket 0, \qquad & \ket{\Phi_2} &= \ket\Psi\ket R\ket 0, \\
    \ket{\Phi_3} &= \ket\Psi\ket R\ket 1, & \ket{\Phi_4} &\eqrefrel{eq:SR.quality}= \ketbot\ket R\ket 1\pm\varepsilon_n,\\
    \ket{\Phi'} &= \ketbot\ket R\ket 0\pm\varepsilon_n = (\Opsi\ket\Phi)\ket R\ket 0\pm\varepsilon_n.
    \hskip-2in
  \end{alignat*}
  For $\ket\Phi=\ketbot$, we have
  \begin{alignat*}2
    \ket{\Phi_1} &= \ketbot\ket R\ket 1, \qquad & \ket{\Phi_2} &\eqrefrel{eq:SR.quality}= \ket\Psi\ket R\ket 1\pm\varepsilon_n, \\
    \ket{\Phi_3} &= \ket\Psi\ket R\ket 0\pm\varepsilon_n, \qquad & \ket{\Phi_4} &= \ket\Psi\ket R\ket 0\pm\varepsilon_n,\\
    \ket{\Phi'} &= \ket\Psi\ket R\ket 0\pm\varepsilon_n = (\Opsi\ket\Phi)\ket R\ket 0\pm\varepsilon_n.
    \hskip-2in
  \end{alignat*}
  And for $\ket\Phi$ orthogonal to $\ket\Psi$ and $\ketbot$, we have
  \begin{alignat*}2
    \ket{\Phi_1} &= \ket\Phi\ket R\ket 0, \qquad & \ket{\Phi_2} &= \ket\Phi\ket R\ket 0, \\
    \ket{\Phi_3} &= \ket\Phi\ket R\ket 0, & \ket{\Phi_4} &= \ket\Phi\ket R\ket 0,\\
    \ket{\Phi'} &= \ket\Phi\ket R\ket 0 = (\Opsi\ket\Phi)\ket R\ket 0.
    \hskip-2in
  \end{alignat*}
  Thus \eqref{eq:upsi.emul} holds. 

  Without loss of generality, we assume that the algorithm $A$ is
  unitary and only (optionally) performs a final measurement at the
  end. Let $B$ be like $A$, except that
  $B$ has additional register $R,Z$ initialized with $\ket R$,
  $\ket 0$, and that $B$ computes circuit \eqref{eq:opsi.circ} on
  $X,R,Z$ whenever $A$ invokes $\Opsi$ on $X$. (And when $A$ performs
  a controlled invocation of $\Opsi$, then $B$ executes the circuit
  with all operations accordingly controlled.)  Let $\ket{\Phi_0}$ be
  the initial state of $A$ and $B$, and let
  $\ket{\Phi_A},\ket{\Phi_B}$ be the final state of $A,B$ (right
  before the final measurement), respectively. Then by induction, from
  \eqref{eq:upsi.emul} we get
  $\bnorm{\ket{\Phi_A}-\ket{\Phi_B}}\leq q_\Psi\varepsilon_n$. By
  \autoref{lemma:dist.similar3},
  $\TD(\ket{\Phi_A}-\ket{\Phi_B})\leqq_\Psi\varepsilon_n$.  Thus
  \begin{multline*}
    \TD\bigl(
    B^{\ORef,\calO}(\ket R,\ket\Phi)
      ,
      A^{\Opsi,\calO}(\ket\Phi)
    \bigr) \leq  q_\Psi\varepsilon_n
    \leq  \tfrac{\pi q_\Psi}{2\sqrt n} + q_\Psi\,o(\tfrac1{\sqrt n}).
    \mathqed
  \end{multline*}
\end{proof}

\begin{lemma}\label{lemma:emulate.oref}
  Let $\ket\Psi$ be a state, chosen according to some
  distribution. Let $\ORef:=I-2\selfbutter\Psi$. Let $\calO$ be an
  oracle, not necessarily independent of $\ket\Psi$. Let $\ket\Phi$ be
  a quantum state, not necessarily independent of $\ket\Psi$. Let $A$
  be an oracle algorithm that makes $q_\mathrm{Ref}$ queries to $\ORef$. Let
  $m\geq 0$ be an integer. Then there is an oracle algorithm~$B$ that
  makes the same number of queries to $\calO$ as $A$ such that:
  \[
  \TD\bigl(
  B^{\calO}(\ket\Psi^{\otimes m},\ket\Phi)
  ,
  A^{\ORef,\calO}(\ket\Phi)
  \bigr) \leq   \frac{2q_\mathit{Ref}}{\sqrt{m+1}}.
  \]
\end{lemma}

\begin{proof}
  Let $\calH$ be the space in which $\ket\Psi$ lives (i.e.,
  $\ket\Psi\in\calH$).  Let $S$ denote a cyclic shift on
  $(m+1)$-partite states. That is,
  $S\ket{\Phi_0}\ket{\Phi_1}\dots\ket{\Phi_m}:=\ket{\Phi_1}\dots\ket{\Phi_m}\ket{\Phi_0}$
  for all $\ket{\Phi_i}\in\calH$.
  (extended linearly to all of $\calH^{\otimes m+1}$).
  $S$ is unitary.

  Let $V\subseteq\calH^{\otimes m+1}$ be the space of states invariant
  under $S$. I.e., $\ket\Phi\in V$ iff $S\ket\Phi=\ket\Phi$.

  Let $U_V$ be the unitary with $U_V\ket\Phi=-\ket\Phi$ for
  $\ket\Phi\in V$, and $U_V\ket\Phi=\ket\Phi$ for $\ket\Phi$
  orthogonal to $V$. (That is, $U_V=I-2P_V$ where $P_V$ is the
  orthogonal projector onto $V$.)

  In this proof, we use
  the following shorthand notation: $\ket{\Phi}=\ket{\Phi'}\pm\varepsilon$ means
  that $\bnorm{\ket{\Phi}-\ket{\Phi'}}\leq\varepsilon$.

  Let $\ket T:=\ket\Psi^{\otimes m}$.

  We show that for any $\ket\Phi\in\calH$,
  \begin{equation}\label{eq:uv.quality}
    U_V\ket\Phi\ket T 
    = (\ORef\ket\Phi)\ket T \pm \tfrac2{\sqrt{m+1}}.
  \end{equation}
  We first show this for $\ket\Phi$ orthogonal to $\ket\Psi$. We
  decompose $\ket\Phi\ket T=\alpha\ket{\chi}+\beta\ket{\kappa}$ for
  quantum states $\ket\chi\in V$, and $\ket\kappa$ orthogonal to
  $V$. 
  Since $\ket\chi\in V$, we have $\bra\chi=\bra\chi S^j$ for any $j$. Thus
  \begin{align*}
    \abs\alpha &= \abs{\bra\chi(\ket\Phi\ket T)}
    = \Babs{\tfrac1{m+1}\sum_{j=0}^{m}\bra\chi S^j(\ket\Phi\ket T)} \\[-3pt]
    &= \tfrac1{m+1}\Babs{\bra\chi \Bigl(\sum_{j=0}^{m} S^j\ket\Phi\ket T\Bigr)}
    \starrel\leq \tfrac1{m+1}\babs{\sqrt{m+1}}
    = \tfrac1{\sqrt{m+1}}.
  \end{align*}
  Here $(*)$ follows from the fact that $\ket\Psi$ and $\ket\Phi$ are
  orthogonal, and hence all $S^j\ket\Phi\ket T$ $(j=0,\dots,m)$ are
  orthogonal, and thus $\bnorm{\sum_j S^j\ket\Phi\ket T}=\sqrt{m+1}$.
  Thus 
  \[
  \bnorm{U_V\ket\Phi\ket T-(\ORef\ket\Phi)\ket T}
  = \bnorm{\ket\Phi\ket T-2\alpha\ket\chi\ -\ \ket\Phi\ket T}
  = \abs{2\alpha} \leq \tfrac2{\sqrt{m+1}}.
  \]
  Thus shows \eqref{eq:uv.quality} for the case that $\ket\Phi$ is
  orthogonal to $\ket\Psi$. If $\ket\Phi=\ket\Psi$,
  \eqref{eq:uv.quality} follows since
  $\ket\Phi\ket T=\ket\Psi^{\otimes m}\in V$ and thus
  $U_V\ket\Phi\ket T=-\ket\Phi\ket T=\ORef\ket\Phi\ket T$. By
  linearity and the triangle inequality, \eqref{eq:uv.quality} then
  holds for all $\ket\Phi\in \calH$.

  \medskip

  Without loss of generality, we assume that the algorithm $A$ is
  unitary and only (optionally) performs a final measurement at the
  end.
  Let $B$ be like $A$, except that $B$
  has additional register $T$ initialized with $\ket T$ (which is given as input), and that $B$
  applies $U_V$ to $X,T$ whenever $A$
  invokes $\ORef$ on $X$. (And when $A$ performs a controlled
  invocation of $\ORef$, then $B$ executes the circuit with all
  operations accordingly controlled.)  Let $\ket{\Phi_0}$ be the
  initial state of $A$ and $B$, and let $\ket{\Phi_A},\ket{\Phi_B}$ be
  the final state of $A,B$ (right before measuring the output),
  respectively. Then by induction, from \eqref{eq:uv.quality} we get
  $\bnorm{\ket{\Phi_A}-\ket{\Phi_B}}\leq\tfrac{2q_\mathrm{Ref}}{\sqrt{m+1}}$. By
  \autoref{lemma:dist.similar3},
  $\TD(\ket{\Phi_A}-\ket{\Phi_B})\leq\frac{2q_\mathit{Ref}}{\sqrt{m+1}}$.
  Thus
  \begin{align*}
    \TD\bigl(
      B^{\calO}(\ket T,\ket\Phi)
      ,
      A^{\ORef,\calO}(\ket\Phi)
    \bigr) \leq  \frac{2q_\mathit{Ref}}{\sqrt{m+1}}.
    \mathqed
  \end{align*}
\end{proof}
}

\shortonly{\usedelayedtext{lemmas opsi}}

\end{fullversion}

\section{The pick-one trick}
\label{sec:pickone}%
\index{pick-one trick}

In this section, we  first show a basic case of the pick-one trick
which focusses on the core query complexity aspects.  In
\autoref{sec:add.oras}, we extend this by a number of additional
oracles that will be needed in the rest of the paper.

{

\newcommand{\I}{{\mathbb{I}}}
\newcommand{\C}{{\mathbb{C}}}
\newcommand{\J}{{\mathbb{J}}}
\newcommand{\R}{{\mathbb{R}}}
\newcommand{\Z}{{\mathbb{Z}}}
\renewcommand{\S}{{\mathbb{S}}}
\newcommand{\W}{{\mathbb{W}}}

\newcommand{\cA}{{\mathcal{A}}}
\newcommand{\cI}{{\mathcal{I}}}
\newcommand{\cO}{{\mathcal{O}}}
\newcommand{\cQ}{{\mathcal{Q}}}
\newcommand{\cH}{{\mathcal{H}}}
\newcommand{\cG}{{\mathcal{G}}}
\newcommand{\cM}{{\mathcal{M}}}
\newcommand{\cS}{{\mathcal{S}}}
\newcommand{\cR}{{\mathcal{R}}}
\newcommand{\cX}{{\mathcal{X}}}
\newcommand{\cY}{{\mathcal{Y}}}

\newcommand{\spn}{\mathop{\mathrm{span}}}
\newcommand{\+}{\oplus}
\newcommand{\wrr}{\mathop{\mathrm{wr}}}
\newcommand{\Ind}{\mathop{\mathrm{Ind}}}
\newcommand{\Tr}{\mathop{\mathrm{Tr}}}

\renewcommand{\>}{\rangle}
\newcommand{\<}{\langle}

\newcommand{\ES}{S}
\newcommand{\SPsi}{\Sigma\Psi}
\newcommand{\SPhi}{\Sigma\Phi}

\begin{definition}[Two values problem]\label{defn:twoValues}%
  \index{two values problem}%
  \index{values problem!two}%
  \index{problem!two values}
  Let $X,Y$ be finite sets and let $k\leq|X|$ be a positive integer. 
  For each $y\in Y$, let $\Sy y$\symbolindexmark{\Sy} be a uniformly random subset of $X$ of
  cardinality $k$, let $\ketpsiy y:=\sum_{x\in\ES_y}|x\>/\sqrt{k}$.\symbolindexmark{\ketpsiy}
  Let $\ketpsi=\sum_{y\in Y}|y\>|\Psi(y)\>/\sqrt{|Y|}$\symbolindexmark{\ketpsi} and $\ket\SPhi=\sum_{y\in Y,x\in X}\ket y\ket x/\sqrt{\abs Y\cdot\abs X}$.
The {\em Two Values} problem is to find $y\in Y$ and $x_1,x_2\in \ES_y$ such that $x_1\neq x_2$ given the following resources:
\begin{compactitem}
\item one instance of the state
$\bigotimes_{\ell=1}^h(\alpha_{\ell,0}|\SPsi\>+\alpha_{\ell,1}|\SPhi\>)$,
where $h$ and the coefficients $\alpha$ are independent of the $S_y$'s and are such that this state has unit norm;
\item an oracle $\OV$\symbolindexmark{\OV} such that for all  $y\in Y$, $x\in X$, $\OV(y,x)=0$ if $x\notin\Sy y$ and $\OV(y,x)=1$ if $x\in\Sy y$.
\item on oracle $\OF$\symbolindexmark{\OF} that, for all $y\in Y$, maps
$|y,\Psi(y)\>$ to $-|y,\Psi(y)\>$ and, for any $|\Psi^\bot\>$ orthogonal to $|\Psi(y)\>$, maps $|y,\Psi^\bot\>$ to itself.\tqed
\end{compactitem}
\end{definition}

The two values problem is at the core of the \emph{pick-one trick}: if
we give an adversary access to the resources described in
\autoref{defn:twoValues}, he will be able to search for one
$x\in\Sy y$ satisfying a predicate $P$ (shown in
\autoref{theo:pick1.complete} below). But  he will not be able to
find two different $x,x'\in\Sy y$ (\autoref{theo:pickone.sound}
below); we will use this to foil any attempts at extracting by
rewinding.

\begin{theorem}[Hardness of the two values problem]\label{theo:pickone.sound}
Let $\cA$ be an algorithm for the Two Values problem that makes $q_V$ and $q_F$ queries to oracles $\cO_V$ and $\cO_F$, respectively.
The success probability for $\cA$ to find $y\in Y$ and $x_1,x_2\in \ES_y$ such that $x_1\neq x_2$ is at most
\begin{equation*}
O\left(
\frac{h}{\abs Y^{1/2}} + \frac{(q_V+q_F)^{1/2}\kk^{1/4}}{\abs X^{1/4}}
+ \frac{(q_V+q_F)^{1/2}}{\kk^{1/4}}
\right).
\mtqed
\end{equation*}
\end{theorem}

\begin{fullversion}
  That is, in order to get a constant success probability in
    finding $x_1,x_2$, one would need at least $h\in\Omega(\sqrt{\abs{Y}})$
    copies of the state $\ket\Psi$, or make
    $\Omega(\min\{\sqrt\kk,\sqrt{\abs X/\kk}\})$ queries.  Or to put it
    differently, if $\sqrt\kk$ and $\sqrt{\abs X/\kk}$ are both
    superpolynomial, a polynomial-time adversary (who necessarily has
    polynomially-bounded $h,q_V,q_F$) finds $x_1,x_2$ only with
    negligible probability.
  
\end{fullversion}

The proof uses the adversary-method from \cite{Ambainis10,AMRR} as
described in the introduction\fullonly{ and is given in
Appendices~\ref{app:proof:theo:pickone.sound}
and~\ref{sec:hardness-last-sec}}. In \autoref{sec:add.oras} we extend
this hardness result to cover additional oracles.

}

\begin{theorem}[Searching one value]\label{theo:pick1.complete}
  Let $\Sy y\subseteq X$ and $\OF,\OV$ be as
  in \autoref{defn:twoValues}.

  There is a polynomial-time oracle algorithm $E_1$ that on input $\ketpsi$ returns a
  uniformly random $y\in Y$ and $\ketpsiy y$.  There is a polynomial-time
  oracle algorithm $E_2$ such that:
  For any $\delta_{\min}>0$,   for any $y\in Y$, 
  for any predicate $P$ on $X$ with $\abs{\{x\in\Sy y:P(x)=1\}}/\abs{\Sy y}\geq\delta_{\min}$,
  and for any $n\geq0$ we have
  \begin{align*}
  \Pr[x\in\Sy y\ \land\ P(x)=1:
  x\ot E_2^{\OV,\OF,P}(n,\delta_{\min},y,\ketpsiy y)]
  \shortonly{\\}
  \geq
  1-2^{-n}.
  \end{align*}
  (The running time of $E_2$ is polynomial-time in $n$, $1/\delta_{\min}$, $\abs y$.)
  \tqed
\end{theorem}

This theorem is proven with a variant of Grover's algorithm
\cite{grover}: Using Grover's algorithm, we search for an $x$ with
$P(x)=1$. However, we do not search over all $x\in\bits\ell$ for some
$\ell$, but instead over all $x\in\Sy y$. When searching over $\Sy y$,
the initial state of Grover's algorithm needs to be
$\sum_x\frac1{\sqrt{\abs{\Sy y}}}\ket x=\ketpsiy y$ instead of
$\sum_x 2^{-\ell/2}\ket x=:\ket{\Phi}$. And the diffusion operator%
\index{diffusion operator!(Grover's algorithm)} $I-2\selfbutter\Phi$
needs to be replaced by $I-2\selfbutterpsiy y$. Fortunately, we
have access both to $\ketpsiy y$ (given as input), and to  $I-2\selfbutterpsiy y$
(through the oracle $\OF$). To get an overwhelming success probability,
Grover's algorithm is usually repeated until it succeeds.  (In
particular, when the number of solutions is not precisely known
\cite{bbht}.) We cannot do that: we have only one copy of the initial
state. Fortunately, by being more careful in how we measure the final
result, we can make sure that the final state in case of failure is
also a suitable initial state for Grover's algorithm. \fullonly{The full proof is given in \autoref{sec:proof:theo:pick1.complete}.}

\subsection{Additional oracles}
\label{sec:add.oras}

In this section, we extend the hardness of the two values problem to
cover additional oracles that we will need in various parts of the
paper. 

\begin{definition}[Oracle distribution]\label{def:ora.dist}
  Fix integers $\ellcom,\ellch,\ellresp$ (that may depend on the security parameter) 
  such that $\ellcom,\ellresp$ are superlogarithmic and $\ellch$ is logarithmic.
  Let $\ellrand:=\ellcom+\ellresp$.\symbolindexmark{\ellrand}

  Let $\Oall=(\OE,\OP,\OR,\OS,\OF,\Opsi,\OV)$\symbolindexmark{\Oall} be chosen according to the
  following distribution:
  \begin{compactitem}
    \item Let $s_0$ be arbitrary but fixed (e.g., $s_0:=0$). Pick $w_0\otR\bits\ellrand$.
    \item Choose $\Sy y$, $\OV$, $\OF$ as in
      \autoref{defn:twoValues} with $Y:=\bits\ellcom$ and
      $X:=\bits\ellch\times\bits\ellresp$ and $\kk:=2^{\ellch+\floor{\ellresp/3}}$.\symbolindexmark{\kk}
    \item For each $z\in\bits{\ellrand}$, pick $y\otR Y$ and $x\otR\Sy y$, and
      set $\OS(z):=(y,x)$.\symbolindexmark{\OS}
    \item Let $\ketbot$\symbolindexmark{\ketbot} be a quantum state orthogonal to all $\ket{\com,\ch,\resp}$
      (i.e., we extend the dimension of the space in which $\ketpsi$ lives by one).
      $\Opsi\ketbot:=\ketpsi$,\symbolindexmark{\Opsi}
      $\Opsi\ketpsi:=\ketbot$, and
      $\Opsi\ket{\Phi}:=\ket\Phi$ for $\ket\Phi$ orthogonal to
      $\ketpsi$ and $\ketbot$. 
    \item Let $\OE(\com,\ch,\resp,\ch',\resp'):=w_0$\symbolindexmark{\OE} iff
      $(\ch,\resp),(\ch',\resp')\in\Sy\com\land(\ch,\resp)\neq(\ch',\resp')$ and $\OE:=0$ everywhere
      else.
    \item Let $\OR(s_0,w_0):=1$\symbolindexmark{\OR} and $\OR:=0$ everywhere else.
    \item For each $\com\in\bits\ellcom,\ch\in\bits\ellch,z\in\bits\ellrand$, let $\OP(w_0,\com,\ch,z)$\symbolindexmark{\OP} be assigned a
      uniformly random $\resp$ with $(\ch,\resp)\in \Sy\com$. (Or
      $\bot$ if no such $\resp$ exists.) Let
      $\OP(w,\cdot,\cdot,\cdot):=0$ for $w\neq w_0$. \tqed
    \end{compactitem}
\end{definition}

\begin{fullversion}
  From these oracles, $\OP$ is later used to implement
    the prover in our sigma-protocols, $\OE$ for the extractor, $\OR$
    to test membership in the relation $R$, and $\OS$ to implement the
    simulator. Notice that $\OS$ and $\OP$ get an additional input $z$
    that seems useless. However, $z$ is needed to get several
    independent answers from the oracle given otherwise equal inputs
    (i.e., it emulates probabilistic behavior).  

    Note\pagelabel{revpage:r.o.random} that both the relation $R$ and the oracles are chosen
    randomly (but not independently of each other). We will assume
    this implicitly in all further theorems. We could also get a
    result relative to a fixed (i.e., non-probabilistic) relation and
    oracle by using the probabilistic method. We omit the details from
    this work.
  
\end{fullversion}

The following corollary is a strengthening of
\autoref{theo:pickone.sound} to the oracle distribution from
\autoref{def:ora.dist}. For later convenience, we express the
soundness additionally in terms of guessing $w_0$. \fullonly{Since the formula
would become unwieldy, we do not give a concrete asymptotic bound here. But
such a bound can be easily derived from the inequalities
(\ref{eq:pap1}--\ref{eq:p6})
in the proof.}
\begin{corollary}[Hardness of two values 2]\label{coro:pick.one.sound-oall}
  Let $\Oall=(\OE,\OP,\OR,\OS,\OF,\Opsi,\OV),w_0$ be as in
  \autoref{def:ora.dist}.  Let $A$ be an oracle algorithm making at
  most $q_E,q_P,q_R,q_S,q_F,q_\Psi,q_V$ queries to
  $\OE,\OP,\OR,\OS,\OF,\Opsi,\OV$, respectively.  Assume that
  $q_E,q_P,q_R,q_S,q_F,q_V$ are polynomially-bounded (and  $\ellcom,\ellresp$ are superlogarithmic
  by \autoref{def:ora.dist}).  Then:
  \fullshort{
  \begin{compactenum}[(i)]
  \item\label{item:pick.one.w} $\Pr[w=w_0: w\ot A^{\Oall}]$ is negligible.
  \item\label{item:pick.one.xx} $\Pr[(\ch,\resp)\neq (\ch',\resp')\land (\ch,\resp),(\ch',\resp')\in\Sy \com: (\com,\ch,\resp,\ch',\resp')\ot A^{\Oall}]$ is
    negligible.
  \end{compactenum}}{
  \begin{compactenum}[(i)]
  \item\label{item:pick.one.xx} $\Pr[(\ch,\resp)\neq (\ch',\resp')\land (\ch,\resp),(\ch',\resp')\in\Sy \com: (\com,\ch,\resp,\ch',\resp')\ot A^{\Oall}]$ is
    negligible.
  \item\label{item:pick.one.w} $\Pr[w=w_0: w\ot A^{\Oall}]$ is negligible.\tqed
  \end{compactenum}
}
\end{corollary}

This corollary is shown by reduction to
\autonameref{theo:pickone.sound}.  Given an adversary that violates
\eqref{item:pick.one.w}, we remove step by step the oracles that are
not covered by \autoref{theo:pickone.sound}.  First, we remove the
oracles $\OP,\OR$. Those do not help the adversary (much) to find
$w_0$ because $\OP$ and $\OR$ only give non-zero output if their input
already contains $w_0$. Next we change $A$ to output a collision
$(\ch,\resp)\neq (\ch',\resp')\land (\ch,\resp),(\ch',\resp')\in\Sy
\com$
instead of the witness $w_0$; since $w_0$ can only be found by
querying $\OE$ with such a collision, this adversary succeeds with
non-negligible probability, too. Furthermore, $A$ then does not need
access to $\OE$ any more since $\OE$ only helps in finding $w_0$. Next
we get rid of $\Opsi$: \fullonly{as shown in \autonameref{theo:emulate.opsi},}
$\Opsi$ can be emulated (up to an inversely polynomial error) using
(suitable superpositions on) copies of the state $\ketpsi$.  Finally
we remove $\OS$: Using the ``small range distribution'' theorem from \cite{zhandry12random},
$\OS$ can be
replaced by an oracle that provides only a polynomial number of
triples $(\com,\ch,\resp)$. Those triples the adversary can produce
himself by measuring polynomially-many copies of $\ketpsi$ in the
computational basis. Thus we have shown that without loss of
generality, we can assume an adversary that only uses the oracles
$\OF,\OV$ and (suitable superpositions of) polynomially-many copies of
$\ketpsi$, and that tries to find a collision. But that such an
adversary cannot find a collision was shown in
\autoref{theo:pickone.sound}.

And \eqref{item:pick.one.xx} is shown by observing that an adversary
violating \eqref{item:pick.one.w} leads to one violating
\eqref{item:pick.one.xx} using one extra $\OE$-query.

\fullonly{The full proof is given in \autoref{app:proof:coro:pick.one.sound-oall}.}

\section{Attacking commitments}
\label{sec:com}

In the classical setting, a non-interactive commitment scheme is
usually called computationally binding if it is hard to output a
commitment and two different openings (\autoref{def:binding}
below). We now show that in the quantum setting, this definition is
extremely weak. Namely, it may still be possible to commit to a value
and then to open the commitment to an arbitrary value (just not to two
values \emph{at the same time}).

\paragraph{Security definitions.} To state this more formally, we
define the security of commitments: A \emph{non-interactive commitment}%
\index{commitment scheme!non-interactive}%
\index{non-interactive!commitment scheme} scheme consists of
algorithms $\COM,\COMverify$, such that $(c,u)\ot\COM(m)$ returns a
commitment $c$ on the message $m$, and an opening information $u$. 
The\pagelabel{revpage:commit.syntax}
  sender then sends $c$ to the recipient, who is not supposed to learn anything about $m$.
Only when the sender later sends $m,u$, the recipients learns $m$. But, intuitively speaking, 
the sender should not be able to ``change his mind'' about $m$ after sending $c$ (binding property).
We
require \emph{perfect completeness}%
\index{completeness!perfect (of commitment scheme)}%
\index{perfect completeness!(of commitment scheme)}%
\index{commitment scheme!perfect completeness}, i.e., for any $m$ and
$(c,u)\ot\COM(m)$, $\COMverify(c,m,u)=1$ with probability~$1$.
In our setting, $c,m,u$ are all classical.

\begin{definition}[Computationally binding]\label{def:binding}
  A commitment scheme $\COM,\COMverify$ is \emph{computationally
    binding}%
  \index{binding!computationally}%
  \index{computationally binding}%
  \index{commitment scheme!computationally binding} iff for any
  quantum polynomial-time algorithm $A$ the following probability is negligible:
  \begin{align*}
    \Pr[&\ok=\ok'=1\ \land\ m\neq m':(c,m,u,m',u')\ot A,\\
    &\ok\ot\COMverify(c,m,u),\
    \ok\ot\COMverify(c,m',u')]
    \mtqed
  \end{align*}
\end{definition}
We will show below that this definition is \emph{not} the right one in the
quantum setting.

\cite{qpok} also introduces a stronger variant of the binding
property, called strict binding, which requires that also the opening
information $u$ is unique (not only the message). 
The results from \cite{qpok} show that strict binding commitments
can behave better under rewinding, so perhaps strict binding commitments 
can avoid the problems that merely binding commitments have?
We define a
computational variant of this property here:
\begin{definition}[Computationally strict binding]\label{def:strict.binding}
A commitment scheme $\COM,\penalty0 \COMverify$ is \emph{computationally
    strict binding}%
  \index{strict binding!computationally}%
  \index{binding!computationally strict}%
  \index{computationally strict binding}%
  \index{commitment scheme!computationally strict binding} iff for any
  quantum polynomial-time algorithm $A$ the following probability is negligible:
  \begin{align*}
    \Pr[&\ok=\ok'=1\fullonly\ \land\fullonly\ (m,u)\neq (m',u'):(c,m,u,m',u')\ot A,\\
    &\ok\ot\COMverify(c,m,u),\
    \ok\ot\COMverify(c,m',u')]
    \mtqed
  \end{align*}
\end{definition}
We will show below that this stronger definition is also not sufficient.

\begin{definition}[Statistically hiding]\label{def:hiding}
  A commitment scheme $\COM,\COMverify$ is \emph{statistically
    hiding}%
  \index{hiding!statistically}%
  \index{statistically hiding}%
  \index{commitment scheme!statistically hiding} iff for all $m_1,m_2$
  with $\abs{m_1}=\abs{m_2}$ and $c_i\ot\COM(m_i)$ for $i=1,2$, $c_1$
  and $c_2$ are statistically indistinguishable.
\end{definition}

\paragraph{The attack.} We now state the insecurity of computationally
binding commitments. The remainder of this section will prove the following
theorem.
\begin{theorem}[Insecurity of binding commitments]\label{theo:com.attack}
  There is an oracle $\calO$ and a non-interactive commitment scheme
  $\COM,\COMverify$ such that:
  \begin{compactitem}
  \item The scheme is perfectly complete, computationally binding, computationally strict
    binding, and statistically hiding.
  \item There is a quantum polynomial-time adversary $B_1,B_2$ such that for
    all $m$, 
    \begin{align*}
    \Pr[\ok=1:
    \shortonly{{}&}
    c\ot B_1(\abs m), u\ot B_2(m),
    \shortonly{\\&}
    \ok\ot\COMverify(c,m,u)]
    \end{align*}
    is overwhelming. (In other words, the adversary can open to a
    value $m$ that he did not know while committing.)
\tqed
  \end{compactitem}
\end{theorem}

In the rest of this section, when referring to the sets $\Sy\com$ from
\autoref{def:ora.dist}, we will call them $\Sy y$ and
we refer to their members as $x\in\Sy y$. (Not
$(\ch,\resp)\in\Sy\com$.) In particular, oracles such as $\OS$ will
returns pairs $(y,x)$, not triples $(\com,\ch,\resp)$, etc.

We construct a commitment scheme relative to the
oracle $\Oall$ from \autoref{def:ora.dist}. (Note: that oracle distribution
contains more oracles than we need for \autoref{theo:com.attack}.
However, we will need in later sections that our commitment scheme is
defined relative to the same oracles as the proof systems there.)
\begin{definition}[Bad commitment scheme]\label{def:com}
  Let $\ibit i(x)$ denote the $i$-th bit of $x$.
  We define $\COM,\COMverify$%
  \symbolindexmark{\COM}\symbolindexmark{\COMverify} as follows:
  \begin{compactitem}
    \item $\COM(m)$\textnormal: For $i=1,\dots,\abs m$, pick
      $z_i\otR\bits\ellrand$ and let  $(y_i,x_i):=\OS(z_i)$. Let
      $p_i\otR\{1,\dots,\ellch+\ellresp\}$. 
      Let $b_i:=m_i\oplus\ibit{p_i}(x_i)$. Let
      $c:=(p_1,\dots,p_{\abs m},y_1,\dots,y_\abs m,b_1,\dots,b_{\abs m})$ and
      $u:=(x_1,\dots,x_{\abs m})$. Output $(c,u)$.
    \item $\COMverify(c,m,u)$ with
      $c=(p_1,\dots,p_n,y_1,\dots,y_n,b_1,\dots,b_n)$ and
      $u=(x_1,\dots,x_n)$\textnormal: Check whether $\abs m=n$. Check whether
      $\OV(y_i,x_i)=1$ for $i=1,\dots,n$. Check whether
      $b_i=m_i\oplus\ibit{p_i}(x_i)$ for $i=1,\dots,n$. Return $1$ if all
      checks succeed.\tqed
  \end{compactitem}
\end{definition}
For the results of the current section, there is actually no need for
the values $p_i$ which select which bit of $x_i$ is used for masking
the committed bit $m_i$. (E.g., we could always use the least
significant bit of $x_i$.)
But in \autoref{sec:fischlin} (attack on Fischlin's scheme) we will
need commitments of this particular form to enable a specific attack
where we need to open commitments to certain
values while \emph{simultaneously} searching for these values in the first place.

\begin{lemma}[Properties of $\COM$]\label{lemma:com.props}
  The scheme from \autoref{def:com} is perfectly complete, computationally binding,
  computationally strict binding, and statistically hiding. (Relative
  to $\Oall$.)\tqed
\end{lemma}

The computational binding and computational strict binding property
are a consequence of \autonameref{coro:pick.one.sound-oall}: to open a
commitment to two different values, the adversary would need to find
one $y_i$ (part of the commitment) and two
$x_i\in\Sy{y_i}$ (part of the two openings). \autoref{coro:pick.one.sound-oall} states that this
only happens with negligible probability. Statistical hiding follows
from the fact that for each $y_i$, there are superpolynomially many
$x_i\in\Sy{y_i}$, hence $\ibit{p_i}(x_i)$ is almost independent of~$y_i$.

\begin{fullversion}
  The proof is given in \autoref{sec:proof:lemma:com.props}.
\end{fullversion}

\begin{lemma}[Attack on $\COM$]\label{lemma:com.attack}
  There is a quantum polynomial-time adversary $B_1,B_2$ such that for all
  $m$,
  \fullshort{
    \begin{align*}
    \fullonly{\varepsilon_{\COM}:=}\Pr[\ok=1:
    c\ot B_1(\abs m), u\ot B_2(m),
    \ok\ot\COMverify(c,m,u)]
  \end{align*}
}{
  $
    \fullonly{\varepsilon_{\COM}:=}\Pr[\ok=1:
    c\ot B_1(\abs m), u\ot B_2(m),
    \ok\ot\COMverify(c,m,u)]
    $}  is overwhelming.\tqed
\end{lemma}

Basically, the adversary $B_1,B_2$ commits to a random commitment. And to unveil to a
message $m$, he needs to find values $x_i\in\Sy{y_i}$ with
$\ibit{p_i}(x_i)=m_i\oplus b_i$. Since half of all $x_i$ have this
property, such $x_i$ can be found using \autonameref{theo:pick1.complete}.

\begin{fullversion}
  The full proof is given in
  Appendix~\ref{app:proof:lemma:com.attack}.
\end{fullversion}

\medskip
\noindent
\autoref{theo:com.attack} then follows immediately from
\multiautoref{lemma:com.props,lemma:com.attack}.

\section{Attacking sigma-protocols}
\label{sec:attack.sigma}

We will now show that in general, sigma-protocols with special
soundness are not necessarily proofs of knowledge. \cite{qpok} showed
that if a sigma-protocol additionally has strict soundness, it is a
proof of knowledge.  It was left as an open problem whether that
additional condition is necessary.  The following theorem resolves
that open question by showing that the results from \cite{qpok} do not
hold without strict soundness (not even with computational strict
soundness), relative to an oracle.

\begin{theorem}[Insecurity of sigma-protocols]\label{theo:know.break.sigma}
  There is an oracle $\Oall$ and a relation $R$ and a sigma-protocol
  relative to~$\Oall$ with logarithmic $\ellch$ (challenge length),
  completeness, perfect special soundness, computational strict
  soundness, and statistical
  honest-verifier zero-knowledge for which there exists a total
  knowledge break.

  In contrast, a sigma-protocol relative to $\Oall$ with completeness,
  perfect special soundness, and statistical honest-verifier
  zero-knowledge is a classical proof of knowledge.
\end{theorem}
Note that a corresponding theorem with polynomially bounded $\ellch$
follows immediately by parallel repetition of the sigma-protocol.

The remainder of this section will prove
\autoref{theo:know.break.sigma}.  As a first step, we construct the
sigma-protocol.
\begin{definition}[Sigma-protocol]\label{def:sigma}
  Let $\COM,\COMverify$ be the commitment scheme from
  \autoref{def:com}.\footnote{The commitment described there has
    the property that it is computationally binding, but still it is
    possible for the adversary to open the commitment to any value,
    only not to several values at the same time. The commitment is
    defined relative to the same oracle distribution as the
    sigma-protocol here, which is why we can use it.}

  Relative to the oracle distribution from \autoref{def:ora.dist}, we
  define the following sigma-protocol
  $(\ellcom,\ellch,\ellresp,P_1,P_2,V,R)$ for the relation
  $R:=\{(s_0,w_0)\}$:
  \begin{compactitem}
    \item $P_1(s,w)$ picks $\com\otR\bits\ellcom$.
      For each $\ch\in\bits\ellch$, he picks $z_\ch\otR\bits\ellrand$
      and computes
      $\resp_\ch:=\OP(w,\com,\ch,z_\ch)$ 
      and $(c_\ch,u_\ch)\ot\COM(\resp_\ch)$. Then $P_1$ outputs
      $\com^*:=(\com,(c_\ch)_{\ch\in\bits\ellch})$.
    \item $P_2(\ch)$ outputs $\resp^*:=(\resp_\ch,u_\ch)$.
    \item For $\com^*=(\com,(c_\ch)_{\ch\in\bits\ellch})$ and
      $\resp^*=(\resp,u)$, let $V(s,\com^*,\ch,\resp^*):=1$ iff
      $\OV(\com,\ch,\resp)=1$ and $s=s_0$ and
      $\COMverify(c_\ch,\resp,u)=1$.\tqed
  \end{compactitem}
\end{definition}

The commitments $c_\ch$ are only needed to get computational strict
soundness. A slightly weaker \autoref{theo:know.break.sigma} without
computational strict soundness can be achieved using the sigma-protocol
from \autoref{def:sigma} without the commitments $c_\ch$; the proofs
stay the same, except that the steps relating to the commitments are
omitted.

\begin{lemma}[Security of the sigma-protocol]\label{lemma:sigma.sec}
  The sigma-protocol from \autoref{def:sigma} has: completeness,
  perfect special soundness, computational strict soundness,
  statistical honest-verifier zero-knowledge, commitment entropy.\tqed
\end{lemma}

Perfect special soundness follows from the existence of the oracle
$\OE$. That oracle provides the witness $w_0$ given two accepting
conversations, as required by perfect special soundness. Computational
strict soundness stems from the fact that the message $\com^*$
contains commitments $c_\ch$ to all possible answers. Thus to break
computational strict soundness (i.e., to find two different accepting
$\resp^*$), the adversary would need to open one of the commitments
$c_\ch$ in two ways. This happens with negligible probability since
$\COM$ is computationally strict binding. Statistical honest-verifier
zero-knowledge follows from the existence of the oracle
$\OS$ which provides simulations. (And the commitment $c_\ch$ that are
not opened can be filled with arbitrary values due to the statistical
hiding property of $\COM$.)

\begin{fullversion}
  The full proof is given in Appendix~\ref{app:proof:lemma:sigma.sec}.
\end{fullversion}

\begin{lemma}[Attack on the sigma-protocol]\label{lemma:sigma-break}
  Assume that $\ellch$ is logarithmically bounded.
  Then there exists a total knowledge break (\autoref{def:total.break}) against
  the sigma-protocol from \autoref{def:sigma}.\tqed
\end{lemma}

To attack the sigma protocol, the malicious prover uses
\autonameref{theo:pick1.complete} to get a $\com$ and a corresponding
state $\ketpsiy{\com}$. Then, when receiving $\ch$, he needs to find
$(\ch',\resp)\in\Sy{\com}$ with $\ch'=\ch$. Since an inversely
polynomial fraction of $(\ch',\resp)$ satisfy $\ch'=\ch$ ($\ellch$ is
logarithmic), this can be done with
\autoref{theo:pick1.complete}. This allows the prover to succeed in
the proof with overwhelming probability. (He additionally needs to
open the commitments $c_\ch$ to suitably. 
This can be done using \autonameref{lemma:com.attack}.)
However, an extractor that has
the same information as the prover (namely, access to the oracle
$\Oall$) will fail to find $w_0$ by
\autonameref{coro:pick.one.sound-oall}.

\begin{fullversion}
  The full proof is given in
  Appendix~\ref{app:proof:lemma:sigma-break}.
\end{fullversion}

\medskip
\noindent
Now \autoref{theo:know.break.sigma} follows from
\multiautoref{lemma:sigma.sec,lemma:sigma-break}. (The fact that the
sigma-protocol is a classical proof of knowledge is shown in
\cite{damgaard10sigma}.)

Note that we cannot expect to get a total break (as opposed to a total
knowledge break): Since the sigma-protocol is a classical proof of
knowledge, it is also a classical proof. But a classical proof is also
a quantum proof, because an unlimited classical adversary can simulate
a quantum adversary. However, this argument does not apply
when we consider computationally limited provers, see
\autoref{sec:comp} below.

\subsection{The computational case}
\label{sec:comp}

We now consider the variant of the impossibility result from the
previous section. Namely, we consider sigma-protocols that have only
computational security (more precisely, for which the special
soundness property holds only computationally) and show that these are
not even arguments in general (the results from the previous section
only say that they are not arguments of knowledge).

\begin{theorem}[Insecurity of sigma-protocols, computational]\label{theo:break.sigma.comp}
  There is an oracle $\Oall$ and a relation $R'$ and a sigma-protocol
  relative to~$\Oall$ with logarithmic $\ellch$ (challenge length),
  completeness, \emph{computational} special soundness, and statistical
  honest-verifier zero-knowledge for which there exists a \emph{total
  break}.

  In contrast, a sigma-protocol relative to $\Oall$ with completeness,
  computational special soundness, and statistical honest-verifier
  zero-knowledge is a classical argument.
  \tqed
\end{theorem}

Note that a corresponding theorem with polynomially bounded $\ellch$
follows immediately by parallel repetition of the sigma-protocol.  The
remainder of this section is dedicated to proving
\autoref{theo:break.sigma.comp}.

\begin{definition}[Sigma-protocol, computational]\label{def:sigma.comp} 
  We define a sigma-protocol $(\ellcom,\ellch,\ellresp,P_1,P_2,V,R')$
  as in
  \autoref{def:sigma}, except that the relation is $R':=\varnothing$.\tqed
\end{definition}

\begin{lemma}[Security of the sigma-protocol, computational]\label{lemma:sigma.sec.comp}
  The sigma-protocol from \autoref{def:sigma.comp} has: completeness.
  \emph{computational} special soundness.  computational strict
  soundness.  statistical honest-verifier zero-knowledge.  commitment
  entropy.\tqed
\end{lemma}

Most properties are either immediate or shown as in
\autonameref{lemma:sigma.sec}. However, perfect special soundness does
not hold for the sigma-protocol from \autoref{def:sigma.comp}: There
exist pairs of accepting conversations
$(ch,\resp),(\ch',\resp')\in\Sy{\com}$. But these do not allow us to
extract a valid witness for $s_0$ (because $R'=\varnothing$, so no
witnesses exist). However, we have computational special soundness: by
\autonameref{coro:pick.one.sound-oall}, it is computationally infeasible
to find those pairs of conversations.

\begin{fullversion}
  The full proof is given in
  Appendix~\ref{app:proof:lemma:sigma.sec.comp}.
\end{fullversion}

\begin{lemma}[Attack on the sigma-protocol, computational]\label{lemma:sigma-break.comp}
  Assume that $\ellch$ is logarithmically bounded.
  Then there exists a total break (\autoref{def:total.break}) against
  the sigma-protocol from \autoref{def:sigma.comp}.\tqed
\end{lemma}

In this lemma, we use the same malicious prover as in
\autonameref{lemma:sigma-break}. 
That adversary proves the statement $s_0$. Since $R'=\varnothing$,
that statement is not in the language, thus this prover performs a
total break.

\begin{fullversion}
  The full proof is given in
  Appendix~\ref{app:proof:lemma:sigma-break.comp}.
\end{fullversion}

\medskip\noindent
Now \autoref{theo:break.sigma.comp} follows from
\multiautoref{lemma:sigma.sec.comp,lemma:sigma-break.comp}. (And
sigma-protocols with computational special soundness are arguments of
knowledge and thus arguments; we are not aware of an explicit write-up
in the literature, but the proof from \cite{damgaard10sigma} for
sigma-protocols with special soundness applies to this case, too.)

\section{Attacking Fiat-Shamir}
\label{sec:fiat}

\begin{definition}[Fiat-Shamir]\index{Fiat-Shamir!(proof system)}
  Fix a sigma-protocol $(\ellcom,\ellch,\ellresp,P_1,P_2,V,R)$ and an
  integer $r>0$.  Let $H:\bits*\to\bits{r\cdot\ellch}$ be a random
  oracle.  The \emph{Fiat-Shamir construction}%
  \index{Fiat-Shamir}
  $(\PFS,\VFS)$\symbolindexmark{\PFS}\symbolindexmark{\VFS} is the
  following non-interactive proof system:
  \begin{compactitem}
  \item Prover $\PFS(s,w)$: For $(s,w)\in R$, invoke
    $\com_i\ot P_1(s,w)$ for $i=1,\dots,r$. Let
    $\ch_1\Vert\dots\Vert\ch_r:=H(s,\com_1,\dots,\com_r)$.  Invoke
    $\resp_i\ot P_2(\ch_i)$. Return
    $\pi:=(\com_1,\dots,\com_r,\resp_1,\dots,\resp_r)$.
  \item Verifier
    $\VFS(s,(\com_1,\dots,\com_r,\resp_1,\dots,\resp_r))$:
    Let
    $\ch_1\Vert\dots\Vert\ch_r:=H(s,\com_1,\dots,\com_r)$.
    Check whether $V(s,\com_i,\ch_i,\resp_i)=1$ for all
    $i=1,\dots,r$. If so, return $1$.\tqed
  \end{compactitem}
\end{definition}

\begin{theorem}[Insecurity of Fiat-Shamir]\label{theo:know.break.fs}
  There is an oracle $\Oall$ and a relation $R$ and a sigma-protocol
  relative to~$\Oall$ with logarithmic $\ellch$ (challenge length), completeness,
  perfect special soundness, computational strict soundness,
  statistical honest-verifier zero-knowledge, and
  commitment entropy, such that there is total knowledge
  break on the Fiat-Shamir construction.
  In contrast, the Fiat-Shamir construction based on a sigma-protocol
  with the same properties is a classical argument of knowledge
  (assuming that $r\ellch$ is superlogarithmic).\tqed
\end{theorem}

As the underlying sigma-protocol, we use the one from \autoref{def:sigma}.
The attack on Fiat-Shamir is analogous to that on the sigma-protocol
itself. The only difference is that the challenge $\ch$ now comes from
$H$ and not from the verifier; this does not change the attack
strategy.

\fullonly{The full proof is given in Appendix~\ref{app:proof:theo:know.break.fs}.}

\begin{fullversion}
  \subsection{The computational case}
\end{fullversion}

Again, we get even stronger attacks if the special soundness holds
only computationally. 

\begin{theorem}[Insecurity of Fiat-Shamir, computational]\label{theo:break.fs.comp}
  There is an oracle $\Oall$ and a relation $R$ and a sigma-protocol
  relative to~$\Oall$ with logarithmic $\ellch$ (challenge length), completeness,
  \emph{computational} special soundness, computational strict soundness,
  statistical honest-verifier zero-knowledge, and
  commitment entropy, 
  such that there is a \emph{total
  break} on the Fiat-Shamir construction.
In contrast, the Fiat-Shamir construction based on a sigma-protocol
with the same properties is a classical argument of knowledge
(assuming that $r\ellch$ is superlogarithmic).\tqed
\end{theorem}

The proof is along the lines of those of \autoref{theo:know.break.fs}
and \autoref{lemma:sigma-break.comp}\fullonly{ and given in
Appendix~\ref{app:proof:theo:break.fs.comp}}.

\section{Attacking Fischlin's scheme}
\label{sec:fischlin}

In the preceding sections we have used the pick-one trick to give
negative results for the (knowledge) soundness of sigma protocols and
of the Fiat-Shamir construction. Classically, both protocols are shown
sound using rewinding. This leads to the conjecture that the pick-one
trick is mainly useful for getting impossibilities for protocols with
rewinding-based security proofs. Yet, in this section we show that
this is not the case; we use the pick-one trick to give an
impossibility result for Fischlin's proof system with
online-extractors \cite{fischlin05online}. The crucial point of that
construction is that in the classical security proof, no rewinding is
necessary. Instead, a witness is extracted by passively inspecting the
list of queries performed by the adversary.

\begin{definition}[Fischlin's scheme]\index{Fischlin!(proof system)}
  Fix a sigma-protocol $(\ellcom,\ellch,\ellresp,P_1,P_2,V,R)$.  Fix
  integers $\fisb,\fisr,\fisS,\fist$%
  \symbolindexmark\fisb
  \symbolindexmark\fisr
  \symbolindexmark\fisS
  \symbolindexmark{\fist}
  such that $\fisb\fisr$ and $2^{\fist-\fisb}$ are
  superlogarithmic, $\fisb,\fisr,\fist$ are logarithmic, $\fisS\in O(\fisr)$ ($\fisS=0$ is
  permitted), and $\fisb\leq \fist\leq \ellch$.

  Let
  $H:\bits*\to\bits \fisb$ be a random oracle. 
  \emph{Fischlin's construction}
  $(\PFis,\VFis)$\symbolindexmark{\PFis}\symbolindexmark{\VFis} is the
  non-interactive proof system is defined as follows:
  \begin{compactitem}
  \item $\PFis(s,w)$: See \cite{fischlin05online}. (Omitted here since
    we only need to analyze $\VFis$ for our results.)
  \item $\VFis(s,\pi)$ with
    $\pi=(\com_i,\ch_i,\resp_i)_{i=1,\dots,\fisr}$: Check if
    $V(\com_i,\ch_i,\resp_i)=0$ for all $i=1,\dots,\fisr$.  Check if
    $\sum_{i=1}^\fisr H(x,(\com_i)_i,i,\ch_i,\resp_i)\leq \fisS$ (where
    $H(\dots)$ is interpreted as a binary unsigned integer). If all
    checks succeed, return $1$. \tqed
  \end{compactitem}
\end{definition}

The idea (in the classical case) is that, in order to produce triples
$(\com_i,\ch_i,\resp_i)$ that make $ H(x,(\com_i)_i,i,\ch_i,\resp_i)$
sufficiently small, the prover needs try out several accepting
$\ch_i,\resp_i$ for each $\com_i$. So with overwhelming probability,
the queries made to $H$ will contain at least two $\ch_i,\resp_i$ for
the same $\com_i$. This then allows extraction by just inspecting the
queries. 

In the quantum setting, this approach towards extraction does not
work: the ``list of random oracle queries'' is not a well-defined
notion, because the argument of $H$ is not measured when a query is
performed. In fact, we show that Fischlin's scheme is in fact not an
argument of knowledge in the quantum setting (relative to an oracle):

\begin{theorem}[Insecurity of Fischlin's construction]\label{theo:know.break.fischlin}
  There is an oracle $\Oall$ and a relation $R$ and a sigma-protocol
  relative to~$\Oall$ with logarithmic $\ellch$ (challenge length), completeness,
  perfect special soundness, computational strict soundness,
  statistical honest-verifier zero-knowledge, and
  commitment entropy,  
  such that
  there is a total knowledge break of Fischlin's construction.
  \fullshort{In contrast}{Yet}, Fischlin's construction based on a sigma-protocol with
  the same properties is a classical argument of knowledge.\tqed
\end{theorem}

As the underlying sigma-protocol, we use the one from
\autoref{def:sigma}.  The basic idea is that the malicious prover
finds conversations $(\com_i^*,\ch_i,\resp_i^*)$ by first fixing the
values $\com^*_i$, and then using \autoref{theo:pick1.complete} to
find $\ch,\resp^*$ where $\resp^*_i$ contains $\resp_i$ such that
$(\ch_i,\resp_i)\in\Sy{\com_i}$ and
$H(x,(\com^*_i)_i,i,\ch_i,\resp^*_i)=0$. If $\resp^*_i$ would not
additionally contain commitments $c_\ch$ (see \autoref{def:sigma}),
this would already suffice to break Fischlin's scheme. To additionally
make sure we can open the commitments to the right value, we use a
specific fixpoint property of $\COM$. \fullonly{See the full proof
(Appendix~\ref{app:proof:theo:know.break.fischlin}) for details.}

\begin{fullversion}
  \subsection{The computational case}
\end{fullversion}

\begin{theorem}[Insecurity of Fischlin\fullonly{'s construction}, computational]\label{theo:break.fischlin.comp}
  There is an oracle $\Oall$ and a relation $R$ and a sigma-protocol
  relative to~$\Oall$ with logarithmic $\ellch$ (challenge length),
  completeness, computational special soundness, computational strict
  soundness, statistical honest-verifier zero-knowledge, and
  commitment entropy, such that there is a \emph{total break} on
  Fischlin's construction.
  \fullshort{In contrast}{Yet}, Fischlin's construction based on a sigma-protocol
  with the same properties is a classical argument of knowledge.\tqed
\end{theorem}

\begin{fullversion}
  The proof is given in
  Appendix~\ref{app:proof:theo:break.fischlin.comp}.
\end{fullversion}

\begin{fullversion}
\paragraph{Fischlin's scheme with strict soundness.} 
We conjecture that
\multiautoref{theo:know.break.fischlin,theo:break.fischlin.comp} even
hold with strict soundness instead of computational strict
soundness. We sketch our reasoning: Consider a variant of the oracle
distribution from \autoref{def:ora.dist}, in which $\ellch$ is
superlogarithmic (not logarithmic) and in which the sets $\Sy\com$ are
chosen uniformly at random from all sets $S$ which satisfy
$\forall\ch\exists_1\resp.(\ch,\resp)\in S$. Note that the results
from Sections~\ref{sec:com}--\ref{sec:fiat} do not hold in this
setting, because $\ch$ must be polynomially-bounded to show the existence of successful
adversaries.  (Namely, when \autonameref{theo:pick1.complete} is invoked,
the predicate $P$ is true on a $2^{-\ellch}$ fraction of the all
values.) But the proofs of
\autonameref{lemma:fischlin-break} and \autonameref{lemma:fischlin-break.comp} do not
require this. We conjecture that \autoref{coro:pick.one.sound-oall}
still holds in this modified setting (the cardinality of the $\Sy\com$
satisfies the conditions of \autoref{coro:pick.one.sound-oall}, but
the $\Sy\com$ have additional structure). Then the sigma-protocols
from \multiautoref{def:sigma,def:sigma.comp} (without the commitments
$c_\ch$) will still have the properties shown in
\multiautoref{lemma:sigma.sec,lemma:sigma.sec.comp}, but additionally
they will have strict soundness because for any $\com,\ch$, there exists
only one $\resp$ such that $(\ch,\resp)\in\Sy\com$.

We leave the proof that \autoref{coro:pick.one.sound-oall} holds even
for sets $\Sy\com$ with
$\forall\ch\exists_1\resp.(\ch,\resp)\in \Sy\com$ as an open problem.
\end{fullversion}

{

\shortonly{
\medskip
\small
}

\paragraph{Acknowledgments.}
We thank Marc Fischlin and Tommaso Gagliardoni for valuable 
discussions and the initial motivation for this work.  Andris
Ambainis was supported by FP7 FET project QALGO and ERC Advanced Grant MQC (at the University of Latvia) and by National Science Foundation under agreement No. DMS-1128155 (at IAS, Princeton). Any opinions, findings and conclusions or recommendations expressed in this material are those of the author(s) and do not necessarily reflect the views of the National Science Foundation.  Ansis Rosmanis was supported by the
Mike and Ophelia Lazaridis Fellowship, the David R. Cheriton Graduate
Scholarship, and the US ARO.  Dominique Unruh was supported by the
Estonian ICT program 2011-2015 (3.2.1201.13-0022), the European Union
through the European Regional Development Fund through the sub-measure
``Supporting the development of R\&D of info and communication
technology'', by the European Social Fund's Doctoral Studies and
Internationalisation Programme DoRa, by the Estonian Centre of
Excellence in Computer Science, EXCS.

}

\fullshort{
\bibliography{qpok-imposs}
\bibliographystyle{plain}
}{
\pagebreak
\bibliography{short}
\bibliographystyle{IEEEtran}
}

\begin{fullversion}
  \section*{Symbol index}
  \addcontentsline{toc}{section}{Symbol index}
  \renewenvironment{thesymbolindex}{\begin{longtable}{lll}}{\end{longtable}}
  \printsymbolindex

\section*{Keyword index}
\addcontentsline{toc}{section}{Keyword index}

\makeatletter \renewenvironment{theindex}{%
  \begin{multicols}2 \thispagestyle{plain}\parindent\z@
    \parskip\z@ \@plus .3\p@\relax
    \columnseprule \z@ \columnsep 35\p@ \let\item\@idxitem }{%
  \end{multicols}%
}

\printindex
\end{fullversion}

\begin{fullversion}
\appendix

\section{Auxiliary lemmas}

\begin{lemma}\label{lemma:cos}
  $  \sqrt{2\bigl(1-(\cos\tfrac{\pi}{2n})^n\bigr)}
  \in
  \frac{\pi}{2\sqrt n}+o(\tfrac1{\sqrt n})$.
\end{lemma}

\begin{proof}
  By Taylor's theorem, for $x\to0$,
  \begin{gather}
    \cos x\in 1-\tfrac{x^2}2+O(x^4) ,
    \label{eq:taylor.cos}
    \\
    \ln(1- x)\in -x+O(x^2) ,
    \label{eq:taylor.ln}
    \\
    e^x \in 1+x+O(x^2).
    \label{eq:taylor.e}
  \end{gather}
  Hence for $n\to\infty$,
  \[
  \ln\cos\tfrac\pi{2n}
  \eqrefrel{eq:taylor.cos}\in
  \ln\bigl(1-\tfrac{\pi^2}{8n^2}+O(n^{-4})\bigr)
  \eqrefrel{eq:taylor.ln}\subseteq
  -\tfrac{\pi^2}{8n^2}+O(n^{-4}).
  \]
  Hence
  \[
  2n\Bigl(1-\bigl(\cos\tfrac\pi{2n}\bigr)^n\Bigr)
  \in
  2n\Bigl(1-e^{n\bigl(-\tfrac{\pi^2}{8n^2}+O(n^{-4})\bigr)}\Bigr)
  \eqrefrel{eq:taylor.e}\subseteq
  2n\Bigl(\tfrac{\pi^2}{8n}+O(n^{-2})\Bigr)
  \subseteq
  \tfrac{\pi^2}4+o(1)
  .
  \]
  Thus
  \[
  \sqrt n\cdot \sqrt{2\bigl(1-(\cos\tfrac{\pi}{2n})^n\bigr)}
  \in
  \tfrac\pi2+o(1)
  \]
  and
  \begin{align*}
    \sqrt{2\bigl(1-(\cos\tfrac{\pi}{2n})^n\bigr)} \in 
    \tfrac{\pi}{2\sqrt n}+o(\tfrac1{\sqrt n}).
    \mathqed
  \end{align*}
\end{proof}

\begin{lemma}\label{lemma:p.fraction}
  Let $X$ be a set. Let $P\subseteq X$ be a set. Let $S\subseteq X$ be
  uniformly random with $\abs S=k$. 
  Let $\varphi:=\abs P/\abs X$.
  Let $\delta_{\min}\in[0,\varphi]$. 
  Then
  \[
  \Pr\Bigl[\frac{\abs{P\cap S}}{\abs{S}} < \delta_{\min}\Bigr] \leq  e^{-2k(\varphi-\delta_{\min})^2}.
  \]
\end{lemma}

\begin{proof}
  Let $N:=\abs X$. Let $\delta:=\abs{P\cap S}/\abs{S}$. 
  We can describe the choice of $S$ as sampling $k$ elements
  $x_i\in X$ without replacement. Let $X_i:=0$ if $x_i\in P$ and
  $X_i:=1$ else. Then $1-\delta=\sum_{i=1}^k X_i/k$. And the $X_i$
  result from sampling $k$ elements without replacement from a
  population $C$ consisting of $(1-\varphi)N$ ones and $\varphi N$
  zeros. Note that $\mu:=1-\varphi$ is the expected value of each
  $X_i$.
  Thus we
  get
  \begin{align*}
    \Pr[\delta < \delta_{\min}]
    &\leq\Pr[1-\delta \geq 1-\delta_{\min}]
    =\Pr\Bigl[\sum\tfrac{X_i}k \geq 1-\delta_{\min}\Bigr] \\
    &=\Pr\Bigl[\sum \tfrac{X_i}{k} - \mu \geq \varphi-\delta_{\min}\Bigr]
    \starrel\leq e^{-2k(\varphi-\delta_{\min})^2}.
  \end{align*}
  Here $(*)$ uses Hoeffding's inequality \cite{JASA1963:Hoeffding} (and the fact that
  $0\leq t\leq 1-\mu$ for $t:=\varphi-\delta_{\min}$). Note that Hoeffding's
  inequality also holds in the case of sampling \emph{without}
  replacement, see \cite[Section 6]{JASA1963:Hoeffding}.
\end{proof}

\begin{lemma}\label{lemma:empirical}
  Let $X$ be a finite and $Y$ a countable set.  Let $\calD$ be a
  distribution over $Y$. Let $\renyi{\frac12}(\calD)$\symbolindexmark{\renyi} denote the R\'enyi
  entropy of order $1/2$ of $\calD$. For each $x\in X$, let $\calO(x)$ be an
  independently chosen $y\ot\calD$.  Let $y_1\ot\calD$, and
  $y_2:=\calO(x)$ for $x\otR X$.
  Then 
  \[
  \SD\bigl((\calO,y_1);(\calO,y_2)\bigr)
  \leq \tfrac1{2\sqrt{\abs X}}2^{\frac12\renyi{\frac12}(\calD)}
  \leq \tfrac12\sqrt{\abs{Y}/\abs{X}}.
  \]
  (I.e., we bound the statistical distance between an element $y_1$
  chosen according to $\calD$, and an element $y_2$ chosen by
  evaluating $\calO$ on a random input, when the function $\calO$ is
  known.)
\end{lemma}

\begin{proof}
  Let $n:=\abs X$.  For a function $f:X\to Y$, let $\calD_f$ denote
  the empirical distribution of $f$, i.e.,
  $\calD_f(y)=\frac1n\abs{\{x:f(x)=y\}}$.  Let
  $j(f):=2\,\SD(\calD,\calD_f)$. And let $J_n:=j(\calO)$, i.e., $J_n$ is
  a real-valued random variable. Then \cite[Lemma
  8]{berend12empirical} proves that
  $\E[J_n]\leq\frac1{\sqrt n}\sum_{y\in Y}\sqrt{\calD(y)}=:\gamma$.
  Since
  $\renyi{\frac12}(\calD)=\frac1{1-\frac12}\log\Bigl(\sum_{y\in
    Y}\calD(y)^{\frac12}\Bigr)$
  by definition, we have
  $\gamma=\frac1{\sqrt n}2^{\frac12\renyi{\frac12}(\calD)}$. Since
  $\renyi{1/2}(\calD)\leq \log \abs{Y}$ for any distribution $\calD$
  on $Y$, we furthermore have
  $\gamma \leq \frac1{\sqrt n} 2^{\frac12\log\abs{Y}}=\sqrt{\abs
    Y/\abs X}$.
  Let $\SD(y_1,y_2|E)$ denote the statistical distance between $y_1$
  and $y_2$ conditioned on an event $E$.
  We can finally compute:
  \begin{multline*}
    \SD\bigl((\calO,y_1);(\calO,y_2)\bigr)
    =
    \sum_{f:X\to Y} \Pr[\calO=f]\cdot\SD(y_1,y_2|\calO=f) \\
    =
    \sum_{f:X\to Y} \Pr[\calO=f]\cdot\SD(\calD,\calD_f) 
    =
    \sum_{f:X\to Y} \Pr[\calO=f]\cdot \tfrac12j(f) 
    = \tfrac12 \E[J_n] \leq \tfrac12\gamma.
    \mathqed
  \end{multline*}
\end{proof}

\begin{lemma}\label{lemma:s.lsb}
  Let $\ibit p(x)$\symbolindexmark{\ibit} denote the $p$-th bit of $x$.
  Let $X=\bits\ell$ for some $\ell$, and $k\geq 1$,
  $p\in\{1,\dots,\ell\}$ be integers. Let $S\subseteq X$ be
  uniformly random with $\abs  S=k$. Let $x\otR S$.  Let
  $b^*\otR\bit$. Then $\SD\bigl((S,\ibit p(x));(S,b^*)\bigr)\leq 1/2\sqrt k$.
\end{lemma}

\begin{proof}
  Let $P:=\{x\in S:\ibit p(x)=1\}$. Let $\SD(X;Y|S)$ denote the
  statistical distance between $X$ and $Y$ conditioned on a specific
  choice of $S$. And $\Pr[S]$ denote the probability of a specific
  choice of $S$. Then
  \begin{align*}
    &\SD\bigl((S,\ibit p(x));(S,b^*)\bigr)\\
    &=\sum_{S}\Pr[S]\SD(\ibit p(x);b^*|S) \\
    &=\sum_{S}\Pr[S]\cdot\Babs{\Pr[x\in P:x\otR S]-\Pr[b^*=1:b^*\otR\bit]}\\
    &=\sum_{S}\Pr[S]\cdot\Babs{\tfrac{\abs P}{\abs S}-\tfrac12}
    \starrel\leq
    \sqrt{\smash{\sum_{S}}\Pr[S]\Bigl(\tfrac{\abs P}{\abs S}-\tfrac12\Bigr)^2} 
    =\sqrt{\E\Bigl[\bigl(\tfrac{\abs P}{\abs S}-\tfrac12\bigr)^2\Bigr]} \\
    &=\sqrt{\E\Bigl[\Bigl(\tfrac{\abs P}{\abs S}-E\Bigl[\tfrac{\abs P}{\abs S}\Bigr]\Bigr)^2\Bigr]} 
    =\sqrt{\Var\bigl[\abs P/\abs S\bigr]} \starstarrel= \tfrac1k\sqrt{\Var\bigl[\abs P\bigr]}.
  \end{align*}
  Here $(*)$ uses Jensen's inequality. And $(**)$ that $\abs S=k$.

  $\abs{P}$ is the number of successes when sampling $k$ times without
  replacement from a population of size $2^{\ell}$ containing
  $2^{\ell-1}$ successes (the elements $x\in\bits\ell$ with
  $\ibit p(x)=1$). That is, $\abs P$ has hypergeometric distribution
  with parameters $m=n=2^{\ell-1}$ and $N:=k$ (in the notation of
  \cite{mathworld.hypergeo}). Thus (see \cite{mathworld.hypergeo}):
  \[
  \Var\bigl[\abs P\bigr]=
  \frac{mnN(m+n-N)}{(m+n)^2(m+n-1)}
  =\tfrac14 k\frac{2^{\ell}-k}{2^{\ell}-1} \leq \tfrac k4.
  \]
  Summarizing,
  \begin{align*}
    \SD\bigl((S,\ibit p(x));(S,b^*)\bigr)
    \leq
    \tfrac1k\sqrt{\Var\bigl[\abs P\bigr]}
    \leq
    \tfrac1k\sqrt{k/4} = \frac1{2\sqrt k}.
    \mathqed
  \end{align*}
\end{proof}

\begin{lemma}\label{lemma:pick.dep} Let $C$ and $R$ be finite sets, let $k\geq1$ be an
  integer. Let $S$ be a uniformly chosen subset of $C\times R$ with
  $\abs S=k$. Let $c'\otR C$, and $r\otR S_{|c'}:=\{r:(c',r)\in S\}$
  (with $r:=\bot\notin R$ iff $S_{|c'}=\varnothing$). Let
  $(c'',r'')\otR S$.
  
  Then $\sigma:=\SD\bigl((S,c',r');(S,c'',r'')\bigr)\leq
  \frac{2k^2}{\abs{C\times R}} 
  +
  \frac{\sqrt{\abs C}}{2\sqrt k}
  $.
\end{lemma}

\begin{proof}
  In the following calculation, $G\stackrel\varepsilon\approx H$ means that the
  distribution of $(S,c)$ when picked according to $G$ has statistical
  distance $\leq\varepsilon$ from the distribution of $(S,c)$ when
  picked according to $H$. 
  And $G\equiv H$ means equality of these distributions
  ($G\stackrel0\approx H$).
  And $[C\times R]_k$ denotes the set of all
  $S\subseteq C\times R$ with $\abs S=k$.
  And $x_1,\dots,x_k\txtrel{$\neq$}\ot M$
  means that the $x_i$ are chosen uniformly but distinctly from $M$
  (drawn without replacing).
  \begin{align*}
    & S\otR[C\times R]_k,\ (c,r)\otR S \\
    \equiv{} & F(1),\dots,F(k)\txtrel{$\neq$}\ot C\times R,\ S:=\im F,\
    j\otR\{1,\dots,k\},\ (c,r):=F(j) \\
    \stackrel{\varepsilon_1}\approx{} &
    F(1),\dots,F(k) \otR C\times R,\
    S:=\im F,\
    j\otR\{1,\dots,k\},\ (c,r):=F(j) \\
    \equiv{} &
    F_1(1),\dots,F_1(k) \otR C,\
    F_2(1),\dots,F_2(k) \otR R,\
    S:=\im((F_1,F_2)),\\
    &\qquad
    j\otR\{1,\dots,k\},\ c:=F_1(j),\ r:=F_2(j)
    \\
    \equiv{} &
    F_1(1),\dots,F_1(k) \otR C,\
    j\otR\{1,\dots,k\},\ c:=F_1(j),\\
    &\qquad
    F_2(1),\dots,F_2(k) \otR R,\
    S:=\im((F_1,F_2))
    \\
    \stackrel{\varepsilon_2}\approx {}&
    F_1(1),\dots,F_1(k) \otR C,\
    c\otR C,\
    F_2(1),\dots,F_2(k) \otR R,\
    S:=\im(F_1,F_2)
    \\
    \equiv{}&
    F(1),\dots,F(k)\otR C\times R,\ S:=\im F,\ c\otR C
    \\
    \stackrel{\varepsilon_1}\approx{} &
    F(1),\dots,F(k)\txtrel{$\neq$}\ot C\times R,\ S:=\im F,\ c\otR C
    \\
    \equiv {}&
    S\otR[C\times R]_k,\ c\otR C
  \end{align*}
  Here $\varepsilon_1$ is the probability that at least two
  independently chosen $F(i)\otR C\times R$ are equal, and
  $\varepsilon_2=\SD\bigl((F_1,c);(F_1,u)\bigr)$ for $u\otR C$.

  Thus
  $\SD\bigl((S,c');(S,c'')\bigr)\leq2\varepsilon_1+\varepsilon_2$. Since
  $r'$ given $S,c'$ has the same distribution as $r''$ given $S,c''$,
  it follows
  \begin{equation}\label{eq:Scr}
    \SD\bigl((S,c',r');(S,c'',r'')\bigr)\leq2\varepsilon_1+\varepsilon_2.
  \end{equation}

  We have
  $\varepsilon_1\leq\sum_{i\neq j}\Pr[F(i)=F(j)]=\sum_{i\neq j}
  1/\abs{C\times R}\leq k^2/\abs{C\times R}$.

  For a function $f:\{1,\dots,k\}\to C$, let $\calD_f$ denote
  the empirical distribution of $f$, i.e.,
  $\calD_f(c)=\frac1k\babs{\{i:f(i)=c\}}$. Let $\mathcal U$ denote the
  uniform distribution on $C$.
  Let
  $j(f):=2\,\SD(\mathcal U,\calD_f)$. And let $J_k:=j(F_1)$ for
  $F_1(1),\dots,F_1(k)\otR C$, i.e., $J_k$ is
  a real-valued random variable. Then \cite[Lemma
  8]{berend12empirical} proves that
  $\E[J_k]\leq\frac1{\sqrt k}\sum_{c\in C}\sqrt{\mathcal U(c)}=\sqrt{\abs
    C/k}$. Then
  \begin{multline*}
    \varepsilon_2=
    \SD\bigl((F_1,c);(F_1,u)\bigr)
    =
    \sum_{f} \Pr[F_1=f]\cdot\SD(\calD_f,\mathcal U) \\
    =
    \sum_{f} \Pr[F_1=f]\cdot \tfrac12j(f) 
    = \tfrac12 \E[J_k] \leq \tfrac12\sqrt{\abs
    C/k}.
  \end{multline*}
  With \eqref{eq:Scr}, the lemma follows.
\end{proof}

\medskip\noindent
We restate an auxiliary lemma from \cite[full version, Lemma 7]{qtc}:
\begin{lemma}\label{lemma:dist.similar}
  Let $\ket{\Psi_1},\ket{\Psi_2}$ be quantum states that can be
  written as $\ket{\Psi_i}=\ket{\Psi^*_i}+\ket{\Phi^*}$ where both
  $\ket{\Psi^*_i}$ are orthogonal to $\ket{\Phi^*}$.
  Then $\TD(\ket{\Psi_1},\ket{\Psi_2})\leq2\norm{\ket{\Psi_2^*}}$.
\end{lemma}

\begin{lemma}\label{lemma:dist.similar3}
  Let $\ket{\Psi_1},\ket{\Psi_2}$ be quantum states.
  Then $\TD(\ket{\Psi_1},\ket{\Psi_2})\leq\bnorm{\ket{\Psi_1}-\ket{\Psi_2}}$.
\end{lemma}

\begin{proof} 
  Fix a basis such that $\ket{\Psi_1}=\ket0$ and
  $\ket{\Psi_2}=\alpha\ket0+\beta\ket1$. Then
  $\abs\alpha^2+\abs\beta^2=1$ and
  \begin{align*}
    \TD(\ket{\Psi_1},\ket{\Psi_2})^2
    \starrel\leq
    1-\abs{\inner{\Psi_1}{\Psi_2}}^2
    = 1-\abs\alpha^2
    = \abs\beta^2
    \leq \abs{1-\alpha}^2+\abs{\beta^2}
    =\bnorm{\ket{\Psi_1}-\ket{\Psi_2}}^2.
  \end{align*}
  Here $(*)$ uses that the trace distance is bounded in terms of the fidelity
    (e.g., \cite[(9.101)]{nielsenchuang-10year}).
  Thus
  $\TD(\ket{\Psi_1},\ket{\Psi_2})\leq\bnorm{\ket{\Psi_1}-\ket{\Psi_2}}$.
\end{proof}

\begin{lemma}[Preimage search in a random function]\label{lemma:grover.dist}
  Let $\gamma\in[0,1]$. Let $Z$ be a finite set. Let $q\geq0$ be an
  integer.  Let $F:Z\to\bit$ be the following function: For each $z$,
  $F(z):=1$ with probability $\gamma$, and $F(z):=0$ else. Let
  $N$ be the function with $\forall z:N(z)=0$.

  If an oracle algorithm $A$ makes at most $q$ queries, then
  \[
  \Babs{\Pr[b=1:b\ot A^F]
    -
    \Pr[b=1:b\ot A^N]} \leq 2q\sqrt\gamma.
  \]
\end{lemma}

\begin{proof}
  We can assume that $A$ uses three quantum registers $A,K,V$ for its
  state, oracle inputs, and oracle outputs. For a function $f$, let
  $O_f\ket{a,k,v}:=\ket{a,k,v\oplus f(k)}$. Then the final state of
  $A^f()$ is $(UO_f)^q\ket{\Psi_0}$ for some unitary $U$ and some
  initial state $\ket{\Psi_0}$. The output $b$ of $A^f$ is then obtained
  by obtained by performing a projective measurement
  $P_\mathit{final}$ on that final state.

  Let $\ket{\Psi_f^i}:=(UO_f)^i\ket{\Psi_0}$ and
  $\ket{\Psi^i}:=(UO_N)^i\ket{\Psi_0}=U^i\ket{\Psi_0}$. (Recall: $N$
  is the constant-zero function.) 

  We compute:
  \begin{alignat*}2
    D_i^f&:=\TD(\ket{\Psi_f^i},\ket{\Psi^i}) = \TD(O_f\ket{\Psi_f^{i-1}},\ket{\Psi^{i-1}}) \\
    &\leq \TD(O_f\ket{\Psi_f^{i-1}},O_f\ket{\Psi^{i-1}})+\TD(O_f\ket{\Psi^{i-1}},\ket{\Psi^{i-1}}) \\
    &= D_{i-1}^f+\TD(O_f\ket{\Psi^{i-1}},\ket{\Psi^{i-1}}).
  \end{alignat*}
  Furthermore $D^f_0=\TD(\ket{\Psi_0},\ket{\Psi_0})=0$, thus
  $D^f_q\leq \sum_{i=0}^{q-1}\TD(O_f\ket{\Psi^{i}},\ket{\Psi^{i}})$.

  Let 
  $Q_z$ be the projector projecting $K$ onto $\ket z$ (i.e., $Q_z=I\otimes\selfbutter z\otimes I$).
 $Q_f$ is the projector projecting $K$ onto all $\ket z$ with $f(z)=1$ (i.e., 
  $Q_f=\sum_{z:f(z)=1} Q_z$). Let $\alpha_f:=\Pr[F=f]$.

  We then have
  \begin{align}
    \sum_f \alpha_f\bnorm{Q_f\ket{\Psi_i}}^2
    &\starrel= \sum_f \alpha_f \sum_{z:f(z)=1} \bnorm{Q_z\ket{\Psi_i}}^2 
    =\sum_{z\in Z} \sum_{f:f(x)=1} \alpha_f \bnorm{Q_z\ket{\Psi_i}}^2 \notag\\
    &\starstarrel= \lambda \sum_z \bnorm{Q_z\ket{\Psi_i}}^2 = \lambda\bnorm{\ket{\Psi_i}}^2 = \lambda.
    \label{eq:avg.norm}
  \end{align}
  Here $(*)$ uses that $Q_f=\sum_{z:f(z)=1} Q_z$ and all $Q_z\ket{\Psi_i}$ are orthogonal.
  And $(**)$ uses that $\sum_{f:f(x)=1} \alpha_f=\Pr[F(x)=1]=\lambda$.
  
  Then
  \begin{align}
    &\sum_{f}\alpha_f\TD(\ket{\Psi_f^q},\ket{\Psi^q}) 
    = \sum_f \alpha_f D_q^f
    \leq \sum_{f,i}\alpha_f\TD(O_f\ket{\Psi^i},\ket{\Psi^i}) \notag\\
    &= \sum_{f,i}\alpha_f\TD\bigl(O_fQ_f\ket{\Psi^i}+(1-Q_f)\ket{\Psi^i},\
                              Q_f\ket{\Psi^i}+(1-Q_f)\ket{\Psi^i}\bigr) \notag\\
    &\starrel\leq \sum_{f,i}\alpha_f2\norm{Q_f\ket{\Psi^i}} 
    \starstarrel\leq 2\sum_i\sqrt{\sum_{f}\alpha_f\norm{Q_f\ket{\Psi^i}}^2} \notag\\
    &\eqrefrel{eq:avg.norm}=
    2\sum_i\sqrt{\lambda}
    = 2q\sqrt{\lambda}.
    \label{eq:avg.td}
  \end{align}
  Here $(*)$ uses \autoref{lemma:dist.similar}. And
  $(**)$ uses Jensen's inequality. 
  Finally,
  \begin{align*}
    &\Babs{\Pr[b=1:b\ot A^F]
      -
      \Pr[b=1:b\ot A^N]} \\
    &\leq \sum_f \alpha_f 
    \Babs{\Pr[b=1:b\ot A^f]
      -
      \Pr[b=1:b\ot A^N]} \\
    &\leq
    \sum_f\alpha_f \TD(\ket{\Psi_f^q},\ket{\Psi^q})
    \eqrefrel{eq:avg.td}\leq 2q\sqrt{\lambda}.
    \mathqed
  \end{align*}
\end{proof}

\noindent The following lemma formalizes that an oracle $\calO_1$ does not help
(much) in finding a value $w$ if $\calO_1$ only gives answers when $w$
is already contained in its input.

\begin{lemma}[Removing redundant oracles 1]\label{lemma:remove.wora}
  Let $w$, $\calO_1$, $\calO_2$ be chosen according to some joint
  distribution.  Here $w$ is a bitstring, and $\calO_1,\calO_2$ are
  oracles, and $\calO_1$ is classical (i.e.,
  $\forall x,y.\exists y'.\calO_1\ket x\ket y=\ket x\ket{y'}$). Fix a
  function $f$. Assume that for all $x$ with $f(x)\neq w$,
  $\calO_1(x)=0$. (In other words,
  $\calO_1\ket{x}\ket{y}=\ket{x}\ket{y}$ for $f(x)\neq w$.)

  Let $A$ be an oracle machine that makes at most $q$ queries to
  $\calO_1$ and $q'$ queries to $\calO_2$. Then
  there is another oracle machine $\Hat A$ that makes at most $q'$ queries to
  $\calO_2$ such that:
  \[
  \Pr[w=w':w'\ot A^{\calO_1,\calO_2}]
  \leq
  2(q+1)\sqrt{\Pr[w'=w:w'\ot\Hat A^{\calO_2}]}
  \]
\end{lemma}

\begin{proof}
  We can assume that $A$ is unitary until the final measurement of its
  output. Then the final state of $A$ before that measurement is
  $\ket{\Psi^*}:=(U_2\calO_1)^qU_2\ket\Psi$ for some unitary $U_2$ depending only on
  $\calO_2$, and $\calO_1$ operating on quantum registers $K,V$ for
  oracle input and output, and $\ket\Psi$ being some initial state
  independent of $\calO_1,\calO_2,w$.  Let
  $\ket{\Psi_i}:=(U_2\calO_1)^{q-i}U_2^{i+1}\ket\Psi$.
  Note that $\ket{\Psi_0}=\ket{\Psi^*}$.
  Let
  $P_X:=\sum_{x:f(x)=w}\selfbutter x\otimes I$ and $\Bar P_X:=1-P_X$.
  Note that since $\calO_1\ket x\ket y=\ket x\ket y$ for $f(x)\neq w$,
  we have $\calO_1=\calO_1P_X+\Bar P_X$.  
  We have for $i=1,\dots,q$:
  \begin{align*}
    \TD(\ket{\Psi_{i-1}},\ket{\Psi_i})
    &=\TD\bigl((U_2\calO_1)^{q-i}(U_2\calO_1)U_2^{i}\ket{\Psi},
    (U_2\calO_1)^{q-i}U_2U_2^{i}\ket{\Psi}\bigr) \\
    &=\TD(\calO_1U_2^i\ket{\Psi},U_2^{i}\ket{\Psi}) \\
    &=\TD\bigl(\calO_1P_XU_2^i\ket{\Psi}+\Bar P_XU_2^i\ket{\Psi},
                      P_XU_2^i\ket{\Psi}+\Bar P_XU_2^i\ket{\Psi}\bigr)
                      \\
    &\starrel\leq 2 \norm{P_XU_2^i\ket\Psi}.
  \end{align*}
  Here $(*)$ uses \autoref{lemma:dist.similar} (using
  that $\ket{\Psi^*_1}:=\calO_1P_XU_2^i\ket{\Psi}$ and 
  $\ket{\Psi^*_2}:=P_XU_2^i\ket{\Psi}$ are both orthogonal to
  $\ket{\Phi^*}:=\Bar P_XU_2^i\ket{\Psi}$ because $\calO_1$ is
  classical and therefore does not leave the image of $P_X$).

  Thus
  $\TD(\ket{\Psi^*},\ket{\Psi_q})\leq\sum_{i=1}^q2\norm{P_XU_2^i\ket\Psi}$.
  For $i=1,\dots,q$, let $A_i^{\calO_2}$ be the oracle algorithm that computes $U_2^i\ket\Psi$ and
  measures register $K$ in the computational basis, giving
  outcome~$x$,
  and then outputs $f(x)$. (Note that $A_i$ does not need access
  to $\calO_1$ because $U_2$ does not depend on $\calO_1$.) Then
  $\Pr[w=w':w'\ot A_i]=\norm{P_XU_2^i\ket\Psi}^2$.
  Let $A_0$ be the oracle machine that performs the same operations as
  $A$, except that it omits all calls to $\calO_1$. That is, its state
  before measuring the output is $\ket{\Psi_q}$. Thus 
  \begin{multline*}
    \babs{\Pr[w=w':w'\ot A^{\calO_1,\calO_2}]-\Pr[w=w':w'\ot A_0^{\calO_2}]} \\
    \leq
    \TD(\ket{\Psi^*},\ket{\Psi_q})
    \leq
    \sum_{i=1}^q 2\sqrt{\Pr[w=w':w'\otR A_i^{\calO_2}]}
  \end{multline*}
  Let $\Hat A^{\calO_2}$ be the algorithm that picks
  $i\otR\{0,\dots,q\}$ and runs $A_i$.  Then
  \begin{align*}
    \Pr[w=w':w'\ot A^{\calO_1,\calO_2}] &\leq \sum_{i=1}^{q}2
    \sqrt{\Pr[w=w':w'\ot A_i^{\calO_2}]} + 
    \Pr[w=w':w'\ot A_0^{\calO_2}] \\
    &\leq2 (q+1)\sum_{i=0}^{q}\tfrac1{q+1}
    \sqrt{\Pr[w=w':w'\ot A_i^{\calO_2}]} \\
    &\starrel\leq2 (q+1)
    \sqrt{\sum_{i=0}^{q}\tfrac1{q+1} \Pr[w=w':w'\ot A_i^{\calO_2}]} \\
    &=2(q+1)\sqrt{\Pr[w=w':w'\ot\Hat A^{\calO_2}]}.
  \end{align*}
  Here $(*)$ uses Jensen's inequality. 
\end{proof}

\medskip\noindent
The following lemma formalizes that if $w$ is a random bitstring that
can be accessed only by querying an oracle $\calO_1$ on some input
$x\in X$, then the probability of finding $w$ using $\calO_1$ is
bounded in terms of the probability of finding some $x\in X$ without
using $\calO_1$.

\begin{lemma}[Removing redundant oracles 2]\label{lemma:remove.xora}
  Let $w$, $X$, $\calO_1$, $\calO_2$ be chosen according to some joint
  distribution such that $w$ and $\calO_2$ are stochastically
  independent. 
  Here $X$ is a set of bitstrings, and $\calO_1,\calO_2$
  are oracles, and $\calO_1$ is classical (i.e.,
  $\forall x,y.\exists y'.\calO_1\ket x\ket y=\ket x\ket{y'}$). 
  And $w$ is uniformly distributed on $\bits\ell$.
  Assume that for all $x\notin X$, $\calO_1(x)=0$. (In
  other words, $\calO_1\ket{x}\ket{y}=\ket{x}\ket{y}$ for
  $x\notin X$.)

  Let $A$ be an oracle machine that makes at most $q$ queries to
  $\calO_1$ and $q'$ queries to $\calO_2$. Then
  there is another oracle machine $\Hat A$ that makes at most $q$ queries to
  $\calO_2$ such that:
  \[
  \Pr[w=w':w'\ot A^{\calO_1,\calO_2}]
  \leq
  2q\sqrt{\Pr[x\in X:x\ot \Hat A^{\calO_2}]}+2^{-\ell}
  \]
\end{lemma}

\begin{proof}
  Let $P_X:=\sum_{x\in X}\selfbutter x\otimes I$ and
  $\Bar P_X:=1-P_X$.  Note that since
  $\calO_1\ket x\ket y=\ket x\ket y$ for $x\notin X$, we have
  $\calO_1=\calO_1P_X+\Bar P_X$.  

  Let $\ket\Psi$, $\ket{\Psi^*}$, $\ket{\Psi_q}$ and $U_2$ be defined
  as in the proof of \autoref{lemma:remove.wora}. (Remember that all
  of these only depend on $\calO_2$, not $\calO_1$.) Exactly as in
  \autoref{lemma:remove.wora}, we get
  $\TD(\ket{\Psi^*},\ket{\Psi_q})\leq\sum_{i=1}^q2\norm{P_XU_2^i\ket\Psi}$.
  For $i=1,\dots,q$, let $A_i^{\calO_2}$ be the oracle algorithm that computes
  $U_2^i\ket\Psi$ and measures register $K$ in the computational basis
  and outputs the outcome.  Then
  $\Pr[x\in X:x\ot A_i^{\calO_2}]=\norm{P_XU_2^i\ket\Psi}^2$.  

  Like in the proof of \autoref{lemma:remove.wora}, let $A_0$ be the
  oracle machine that performs the same operations as $A$, except that
  it omits all calls to $\calO_1$. That is, its state before measuring
  the output is $\ket{\Psi_q}$. Let $\Hat A^{\calO_2}$ pick a random
  $i\otR\{1,\dots,q\}$ (not $i\otR\{0,\dots,q\}$ as in
  \autoref{lemma:remove.wora}!) and run $A_i^{\calO_2}$.
  Then
  \begin{align}
    \hskip1cm&\hskip-1cm
    {\Pr[w=w':w'\ot A^{\calO_1,\calO_2}]-\Pr[w=w':w'\ot A_0^{\calO_2}]}
    \leq
    \TD(\ket{\Psi^*},\ket{\Psi_q})\notag\\
    &\leq
    2q \sum_{i=1}^q\tfrac1q \sqrt{\Pr[x\in X:x\otR A_i^{\calO_2}]}
    \starrel\leq
    2q \sqrt{\sum\nolimits_{i=1}^q\tfrac1q\Pr[x\in X:x\otR
      A_i^{\calO_2}]} \notag\\
    &=2q \sqrt{\Pr[x\in X:x\ot\Hat A^{\calO_2}]}.
    \label{eq:a.a0.diff}
  \end{align}
  Here $(*)$ uses Jensen's inequality.

  Since $w$ and $\calO_2$ are independent and $w$ is uniform on
  $\bits\ell$,
  $\Pr[w=w':w'\ot A_0^{\calO_2}]\leq 2^{-\ell}$. 
  With \eqref{eq:a.a0.diff}, we get 
  $2q\sqrt{\Pr[x\in X:x\ot \Hat A^{\calO_2}]} \geq \Pr[w=w':w'\ot
  A^{\calO_1,\calO_2}]-2^{-\ell}$.
\end{proof}

\begin{theorem}[Small range distributions \cite{zhandry12random}]\label{theo:small.range}
  Fix sets $Z,Y$ and a distribution $\calD_Y$ on $Y$, and integers $s,q$.

  Let $H:Z\to Y$
  be chosen as: for each $z\in Z$, $H(z)\ot\calD_Y$. 

  Let $G:Z\to Y$ be
  chosen as: Pick $y_1,\dots,y_s\ot\calD_Y$, then for each $z\in Z$,
  pick $i_z\otR\{1,\dots,s\}$, and set $G(z):=y_{i_z}$. 

  Let $A$ be an oracle algorithm
  making at most $q$ queries.
  Then
  \[
  \Babs{\Pr[b=1:b\ot A^H]-\Pr[b=1:b\ot A^G]}\leq 
  \pi^2(2q)^3/6s < 14q^3/s.
  \]
\end{theorem}

\begin{proof}
  This is merely a reformulation of \cite[Corollary
  VII.5]{zhandry12random}.  (Note that the distance between
  distributions in \cite{zhandry12random} is defined to be twice the
  statistical, this is why in our formulation of the theorem the bound
  is only half as large.)
\end{proof}

\section{Proofs for \autoref{sec:opsi}}
\label{app:lemmas:opsi}

\usedelayedtext{lemmas opsi}

\newcommand{\I}{{\mathbb{I}}}
\newcommand{\C}{{\mathbb{C}}}
\newcommand{\J}{{\mathbb{J}}}
\newcommand{\R}{{\mathbb{R}}}
\newcommand{\Z}{{\mathbb{Z}}}
\renewcommand{\S}{{\mathbb{S}}}
\newcommand{\W}{{\mathbb{W}}}

\newcommand{\cA}{{\mathcal{A}}}
\newcommand{\cI}{{\mathcal{I}}}
\newcommand{\cO}{{\mathcal{O}}}
\newcommand{\cQ}{{\mathcal{Q}}}
\newcommand{\cH}{{\mathcal{H}}}
\newcommand{\cG}{{\mathcal{G}}}
\newcommand{\cM}{{\mathcal{M}}}
\newcommand{\cS}{{\mathcal{S}}}
\newcommand{\cR}{{\mathcal{R}}}
\newcommand{\cT}{{\mathcal{T}}}
\newcommand{\cX}{{\mathcal{X}}}
\newcommand{\cY}{{\mathcal{Y}}}

\newcommand{\spn}{\mathop{\mathrm{span}}}
\newcommand{\+}{\oplus}
\newcommand{\wrr}{\mathop{\mathrm{wr}}}
\newcommand{\Ind}{\mathop{\mathrm{Ind}}}
\newcommand{\Tr}{\mathrm{Tr}}

\renewcommand{\>}{\rangle}
\newcommand{\<}{\langle}

\newcommand{\Mapsto}{\mathop{\longmapsto}}

\newcommand{\ES}{S}
\newcommand{\SPsi}{\Sigma\Psi}
\newcommand{\SPhi}{\Sigma\Phi}
\newcommand{\pVec}{{\boldsymbol\sigma}} %
\newcommand{\bI}{{\boldsymbol\cI}} %
\newcommand{\z}{{\boldsymbol z}} %
\newcommand{\Id}{\I} %

\newcommand{\Lzero}{a} %
\newcommand{\Lone}{b} %

\def\rket{\rangle}
\def\lbra{\langle}
\def\H{{\cal H}}
\def\A{{\cal A}}

\newcommand{\pro}[2]{{\Pi_{#1}^{#2}}}

\section{Proof of \autoref{theo:pickone.sound}}
\label{app:proof:theo:pickone.sound}
\subsection{Preliminaries}

Let $M=|Y|$ and $N=|X|$ and, without loss of generality, let $Y=\{1,\ldots,M\}$ and $X=\{1,\ldots,N\}$.
 Let $D\subset\{0,1\}^N$ be the set of all $\binom{N}{k}$ $N$-bit strings of Hamming weight $k$.
 For every $y$, we associate $S_y$ with a string $z_{y}\in D$ whose $x$-th entry $z_{y,x}:=(z_y)_x$ is $1$ if and only if $x\in S_y$. 
This association is one-to-one.
The black-box oracles essentially hide an input $z=(z_1,\ldots,z_M)\in D^M$.
Let us write $|\Psi(z_y)\>$ and $|\SPsi(z)\>$ instead of $|\Psi(y)\>$ and $|\SPsi\>$, respectively, to emphasize that these states depend on $z$.

Let $\S_L$ denote the symmetric group of a finite set $L$, that is, the group with the permutations of $L$ as elements and the composition as a group operation. 
For a positive integer $n$, let $\S_n$ denote the isomorphism class of the symmetric groups $\S_L$ with $|L|=n$.
A permutation $\sigma\in \S_{X}$ acts on $z_y\in D$ in a natural way: we define
\begin{equation}
\label{eqn:SXaction}
\sigma(z_y):=(z_{y,\sigma^{-1}(1)},\ldots,z_{y,\sigma^{-1}(N)}),
\end{equation}
so that
$(\sigma(z_y))_{\sigma(x)}=z_{y,x}$ holds. A permutation $\pi\in \S_{Y}$ acts on $z\in D^M$ in the same way: we define
\(
\pi(z):=(z_{\pi^{-1}(1)},\ldots,z_{\pi^{-1}(M)}).
\)

Consider a pair $(\pVec,\pi)$, where $\pVec=(\sigma_1,\ldots,\sigma_M)\in \S_X^M$ and $\pi\in\S_Y$. Let this pair act on $z\in D^M$ by first permuting the entries of $z$ with respect to $\pi$ and then permuting entries within each $(\pi(z))_y$ with respect to $\sigma_y$. Namely, let
\begin{equation}
\label{eqn:rep1}
(\pVec,\pi):(z_1,\ldots,z_M) \mapsto (\sigma_1(z_{\pi^{-1}(1)}),\ldots,\sigma_M(z_{\pi^{-1}(M)})).
\end{equation}
This action defines a (linear) representation of the wreath product $\W:=\S_{X}\wr\S_{Y}$.
\begin{definition}[{\cite[Chapter 4]{james:symmetric}}]  %
The {\em wreath product} $G\wr\S_{M}$ of groups $G$ and $\S_{M}$ is the group whose elements are $(\pVec,\pi)\in G^M\times\S_{M}$ and whose group operation is 
\[
\big((\sigma'_1,\ldots,\sigma'_M),\pi'\big)
 \big((\sigma_1,\ldots,\sigma_M),\pi\big) :=
\big(  (\sigma'_1\sigma_{(\pi')^{-1}(1)},\ldots,\sigma'_M\sigma_{(\pi')^{-1}(M)})
,\pi'\pi \big).
\]
\end{definition}

Let $X_2$ be the set of all $\binom{N}{2}$ size-two subsets of $X$. In addition to (\ref{eqn:rep1}), we are also interested in the following two representations of $\W$ defined by its action on the sets $Y\times X$ and $Y\times X_2$, respectively:
\begin{alignat}{2}
\label{eqn:rep2}
&(\pVec,\pi):(y,x) &&\mapsto (\pi(y),\sigma_{\pi(y)}(x)), \\
\label{eqn:rep3}
&(\pVec,\pi):(y,\{x_1,x_2\}) &&\mapsto (\pi(y),\{\sigma_{\pi(y)}(x_1),\sigma_{\pi(y)}(x_2)\}).
\end{alignat}
The former representation concerns oracle queries and the latter---the output of the algorithm.
 
For $w=(y,x)\in Y\times X$, let $z_w=z_{y,x}$. Note that the representations (\ref{eqn:rep1})
and (\ref{eqn:rep2}) are such that, for $\tau\in\W$, we have $(\tau(z))_{\tau(w)}=z_w$.
\subsection{Registers and symmetrization of the algorithm}

Let $\cH_A$ be the workspace on which $\cA$ operates. We express
\begin{equation}\label{eqn:registers}
\cH_A = \cH_{Q}\otimes\cH_{B}
\otimes\cH_O\otimes\cH_R\otimes\cH_W,
\end{equation}
 where the tensor factors are defined as follows.
\begin{itemize}
\item $\cH_Q:=\cH_{Q_Y}\otimes \cH_{Q_X}$ and $\cH_B$ are the ``query'' registers that the oracles $\cO_V$ and $\cO_F$ use, where $\cH_{Q_Y}$, $\cH_{Q_X}$, and $\cH_B$ correspond to the sets $Y$, $X$, and $\{0,1\}$, respectively. For all $(y,x,b)\in Y\times X\times \{0,1\}$, we have
\begin{equation} \label{eqn:quantumOV}
\cO_V|y,x,b\> := |y,x,b\oplus z_{y,x}\>
\end{equation}
and $\cO_F$ maps $|y,\Psi(y),b\>$ to $-|y,\Psi(y),b\>$ and, for every $|\Psi^\bot\>$ orthogonal to $|\Psi(y)\>$, maps  $|y,\Psi^\bot,b\>$ to itself. 
\item $\cH_O:=\cH_{O_Y}\otimes \cH_{O_{X2}}$ is the ``output'' register, where $\cH_{O_Y}$ and $\cH_{O_{X2}}$ correspond to the sets $Y$ and $X_2$, respectively.
\item $\cH_R:=\bigotimes_{\ell=1}^h\cH_{R(\ell)}$ is the (initial) ``resource'' register, where $\cH_{R(\ell)}=\cH_{R_Y(\ell)}\otimes\cH_{R_X(\ell)}$, in which $\cH_{R_Y(\ell)}$ and $\cH_{R_X(\ell)}$ correspond to the sets $Y$ and $X$, respectively.
At the beginning of the algorithm, the register $\cH_R$ is initialized to the resource state
\begin{equation}
\label{eqn:xiResource}
|\xi'(z)\>:=\bigotimes\nolimits_{\ell=1}^h(\alpha_{\ell,0}|\SPsi(z)\>+\alpha_{\ell,1}|\SPhi\>).
\end{equation}
Also, let $\cH_{R_Y}:=\bigotimes_{\ell=1}^h\cH_{R_Y(\ell)}$ and $\cH_{R_X}:=\bigotimes_{\ell=1}^h\cH_{R_X(\ell)}$. 
\item $\cH_W$ is the rest of the workspace.
\end{itemize}
Let us also define $\cH_{A-Q}$, $\cH_{A-O}$, and $\cH_{A-R}$ to be the space corresponding to all the registers of the algorithm except $\cH_Q$, $\cH_O$, and $\cH_R$, respectively. Let $\Id$ be the identity operator.
We frequently write subscripts below states and unitary transformations to clarify, respectively, which registers they belong to or act on. 
For example, we may write $|\xi'(z)\>_R$ instead of $|\xi'(z)\>$.
We do this especially when the order of registers is not that of (\ref{eqn:registers}).
We may also concatenate subscripts when we use multiple registers at once. For example, we may write $\Id_{QB}$ instead of $\Id_Q\otimes\Id_B$.

Let $|\xi_\emptyset(z)\>_A:=|\xi'(z)\>_R\otimes |\xi''\>_{A-R}$ be the initial state of the algorithm, where $|\xi''\>_{A-R}$ is independent from $z$.
The algorithm makes in total $q_T:=q_V+q_F$ oracle calls. 
For $q\in\{0,1,\ldots,q_T-1\}$, let
\[
 |\xi_q(z)\>_A = \sum_{w\in Y\times X}|w\>_{Q}|\xi_{q,w}(z)\>_{A-Q}
\]
 be the state of the algorithm $\cA$, as a sequence of transformations on $\cH_A$, just before $(q+1)$-th oracle call, $\cO_V$ or $\cO_F$, where $|\xi_{q,w}(z)\>_{A-Q}$ are unnormalized.
Similarly, for $q=q_T$, let 
\[
|\xi_{q_T}(z)\>_A = \sum_{w\in Y\times X_2}|w\>_{O}|\xi_{q_T,w}(z)\>_{A-O}
\]
be the final state of the algorithm.

Let $U_{\bI}$, and $U_{Q}$, and $U_{O}$ be unitary transformations
corresponding to representations (\ref{eqn:rep1}), (\ref{eqn:rep2}),
and (\ref{eqn:rep3}) of $\W$, respectively, where the register
$\cH_\bI$ is yet to be defined. (That is,  $U_{\bI}$,
  $U_{Q}$,  $U_{O}$ are actually families of unitaries, indexed by
  elements $\tau\in\W$.)
We add a subscript $\tau\in\W$ when we want to specify that we are considering the representation of the element $\tau$, for example, we may write $U_{Q,\tau}$. 
Since $\cH_R$ is essentially the $h$-th tensor power of $\cH_Q$, we define $U_{R}:=U_{Q}^{\otimes h}$. 
The tensor product of two (or more) representations of $\W$ is also a representation of $\W$. 
Let $U_{\bI Q}:=U_\bI\otimes U_Q$ and $U_{\bI O}:=U_\bI\otimes U_O$, an we later use analogous notation for other ``concatenations''.

We first ``symmetrize'' $\cA$ by adding an extra register $\cH_S$ holding a ``permutation'' $\tau\in\W$. Initially, $\cH_S$ holds a uniform superposition over all permutations:
\[
|\W\>_S:=\frac{1}{\sqrt{M!(N!)^M}}\sum_{\tau\in\W}|\tau\>_S.
\]
Then, at specific points in the algorithm, we insert unitary transformations controlled by the content $\tau$ of $\cH_S$.
\begin{enumerate}
\item\label{step:urtau}
At the beginning of the algorithm, we insert the controlled transformation $U_{R,\tau}$ on the register $\cH_R$. Recall that, if (and only if) $z_{y,x}=1$, then $(\tau(z))_{\tau(y,x)}=1$. Hence,
\[
\sum_{\tau\in\W}|\tau\>_S|\xi(z)\>_A
\Mapsto^{\tau\text{ on }\cH_R}
\sum_{\tau\in\W}|\tau\>_S|\xi(\tau(z))\>_A.
\]
\item
Before each oracle call, $\cO_V$ or $\cO_F$, we insert the controlled transformation 
$U^{-1}_{Q,\tau}$ on the register $\cH_Q$. 
Note that $(\tau(z))_{y,x}=1$ if and only if $z_{\tau^{-1}(y,x)}=1$, and $\cO_V$ and $\cO_F$ use $z$ as the input. 
After the oracle call, we insert the controlled $U_{Q,\tau}$.
\item
At the end of the algorithm, we insert the controlled transformation $U^{-1}_{O,\tau}$ on the register $\cH_O$
containing the output of $\cA$ because, again, $z_{\tau^{-1}(y,x)}=1$ if and only if $(\tau(z))_{y,x}=1$. 
\end{enumerate}
The effect of the symmetrization is that, on the subspace $|\tau\>_S$, the
algorithm is effectively running on the input $\tau(z)$. If the
original algorithm $\cA$ succeeds on every input $z$ with average success
probability $p$, the symmetrized algorithm succeeds on every input with success
probability $p$.

Next, we recast $\cA$ into a different form, using an ``input'' register $\cH_\bI$ that stores $z\in D^M$. Namely, let 
\(
\cH_\bI :=\bigotimes_{y=1}^M \cH_{I(y)}
\)
 be an $\binom{N}{k}^M$-dimensional Hilbert space whose basis
states correspond to possible inputs $z$, where we define $\cH_{I(y)}$ to be $\binom{N}{k}$-dimensional Hilbert space whose basis states correspond to $z_y\in D$. Since all the spaces $\cH_{I(y)}$ are essentially equivalent, we write $\cH_{I}$ instead of $\cH_{I(y)}$ when we do not care which particular $y\in Y$ we are talking about, and $\cH_\bI=\cH_I^{\otimes M}$. 

 Initially, $\cH_\bI$ is in the uniform superposition of all the basis
 states of $\cH_\bI$.
More precisely, $\cH_\bI\otimes\cH_S\otimes\cH_A$ takes the
  following initial state (before applying the controlled transformation
  $U_{R,\tau}$ in 
  step~\ref{step:urtau} of the symmetrisation above):
\[
 \binom{N}{k}^{-M/2}\sum_{z\in D^M}|z\>_\bI
 \otimes |\W\>_S
 \otimes |\xi_\emptyset(z)\>_A.
\]
We transform the symmetrised version of $\cA$ into a sequence of transformations
on a Hilbert space $\cH = \cH_\bI\otimes\cH_S\otimes\cH_A$ . A black-box transformation $\cO$ (where
$\cO= \cO_V$ or $\cO = \cO_F$) is replaced by a transformation
\(
\cO'=\sum\nolimits_{z\in D^M}|z\>\<z|\otimes \cO(z),
\)
 where $\cO(z)$ is the transformation $\cO$ for the case when the input is equal
to $z$.

At the end, the algorithm measures the input register $\cH_\bI$ and the output register $\cH_O=\cH_{O_Y}\otimes\cH_{O_{X2}}$ in the computational basis, and outputs the result of this measurement: $z\in D^M$, $y\in Y$, and $\{x_1,x_2\}\in X_2$. The algorithm is successful if $z_{y,x_1}=z_{y,x_2}=1$.

For $q\in\{0,\ldots,q_T-1\}$, let $|\phi^-_q\>$ be the state of the algorithm just before the controlled $U^{-1}_{Q,\tau}$ transformation preceding the $(q+1)$-th oracle call, and let $|\phi_q\>$ be the state just after we apply this $U^{-1}_{Q,\tau}$ and still before the oracle call.
Due to the symmetrization, we have
\[
|\phi_q^-\> = \gamma \sum_{z\in D^m}|z\>_\bI\sum_{\tau\in\W}|\tau\>_S\sum_{w\in Y\times X}|w\>_{Q}|\xi_{q,w}(\tau(z))\>_{A-Q},
\]
where $\gamma=1/\sqrt{M!(N!\binom{N}{k})^M}$,
and, after we apply $U^{-1}_{Q,\tau}$, we have
\begin{equation}
\label{eqn:psiPlus}
|\phi_q\> = \gamma \sum_{z\in D^m}|z\>_\bI\sum_{\tau\in\W}|\tau\>_S\sum_{w\in Y\times X}|\tau^{-1}(w)\>_{Q}|\xi_{q,w}(\tau(z))\>_{A-Q}.
\end{equation}
Recall the representations $U_{\bI}$ and $U_{Q}$ of $\W$. 
Let us also consider the right regular representation of $\W$ acting on $\cH_S$: for $\kappa\in\W$, let $U_{S,\kappa}|\tau\>:=|\tau\kappa^{-1}\>$. Let
\(
U_{\bI SQ}:=U_{\bI}\otimes U_{S}\otimes U_{Q},
\)
and, for all $\kappa\in\W$, we have
\begin{multline}\label{eqn:TmuSymmetry}
(U_{\bI SQ,\kappa}\otimes\Id_{A-Q})|\phi_q\>
 = \gamma \sum_{z\in D^m}|\kappa(z)\>_\bI\sum_{\tau\in\W}|\tau\kappa^{-1}\>_S\sum_{w\in Y\times X}|\kappa\tau^{-1}(w)\>_{Q}|\xi_{q,w}(\tau(z))\>_{A-Q} \\
 = \gamma \sum_{z\in D^m}|\kappa(z)\>_\bI\sum_{\tau\in\W}|\tau\kappa^{-1}\>_S\sum_{w\in Y\times X}|(\tau\kappa^{-1})^{-1}(w)\>_{Q}|\xi_{q,w}((\tau\kappa^{-1})(\kappa(z)))\>_{A-Q}=|\phi_q\>.
\end{multline}

For $q\in\{0,1,\ldots,q_T-1\}$, 
let $\rho'_q$ be the density matrix obtained from $|\phi_q\>\<\phi_q|$ by tracing
out the $\cH_S$ and $\cH_{A-Q}$ registers and, in turn, let $\rho_q$ be obtained from
$\rho'_q$ by tracing out the register $\cH_Q$. Due to (\ref{eqn:TmuSymmetry}), we have
\begin{equation} \label{eqn:stepSym}
U_{\bI Q,\tau}
\rho'_q
U_{\bI Q,\tau}^{-1} = \rho'_q
\quad\text{and}\quad
U_{\bI,\tau}
\rho_q
U_{\bI,\tau}^{-1} = \rho_q
\quad\text{for all~}\tau\in\W.
\end{equation}

Similarly, for $q=q_T$, let $|\phi_{q_T}\>$ be the final state of the algorithm (i.e., the state after the controlled $U^{-1}_{O,\tau}$), and it satisfies an analogous symmetry to (\ref{eqn:TmuSymmetry}): for all $\kappa\in\W$, we have 
\(
 (U_{\bI SO,\kappa}\otimes\Id_{A-O})|\phi_{q_T}\> = |\phi_{q_T}\>.
\)
Let $\rho''_{q_T}$ be the density matrix obtained from $|\phi_{q_T}\>\<\phi_{q_T}|$  by tracing out all the registers but $\cH_\bI$ and $\cH_O$, let $\rho_{q_T}$ be obtained from $\rho''_{q_T}$ by tracing out the register $\cH_O$. Again,
we have
\begin{equation} \label{eqn:finalSym}
U_{\bI O,\tau}
\rho''_{q_T}
U_{\bI O,\tau}^{-1} = \rho''_{q_T}
\quad\text{and}\quad
U_{\bI,\tau}
\rho_{q_T}
U_{\bI,\tau}^{-1} = \rho_{q_T}
\quad\text{for all~}\tau\in\W.
\end{equation}

Note that, throughout the algorithm, the density matrix of the $\cH_\bI$ part of the state of the algorithm can be affected only by oracle calls. Therefore, for $q\in\{0,1,\ldots,q_T\}$, this density matrix equals $\rho_q$ just after $q$-th oracle call (at the very beginning of the algorithm, if $q=0$) and remains such till $(q+1)$-th oracle call (till the end of the algorithm, if $q=q_T$).

\subsection{Representation theory of $\S_X$} \label{sec:repTheory}

Consider a positive integer $n$. The representation theory of $\S_n$ is closely related to partitions.
A partition $\lambda$ of $n$ is a non-increasing list $(\lambda_1,\dots,\lambda_k)$ of positive integers satisfying $\lambda_1+\dots+\lambda_k = n$. There is one-to-one correspondence
 between irreducible representations ({\em irreps}, for short) of $\S_n$ and partitions $\lambda\vdash n$, and we will use these terms interchangeably. For example, $(n)$ corresponds to the trivial representation and  $(1^n)=(1,1,\ldots,1)$ to the sign representation. 
(One may refer to~\cite{serre:representation} for more background on the representation theory of finite groups and to~\cite{james:symmetric, sagan:symmetric} for the representation theory of the symmetric group and the wreath product.)

The group action of $\S_X$ on $\cH_{I}$ is given by (\ref{eqn:SXaction}), which defines a representation $U_{I}$ of $\S_X$ (this representation is independent from $y$).
In order to decompose $U_{I}$ into a direct sum of irreps of $\S_N$ (recall that $X=\{1,\ldots,N\}$), first consider the subgroup $\S_{k}\times\S_{N-k}$ of $\S_{N}$, where $\S_k$
permutes $\{1,\ldots,k\}$ and $\S_{N-k}$ permutes $\{k+1,\ldots,N\}$.
Let $V_{I,\sigma}$ be $U_{I,\sigma}$ restricted to $\sigma\in \S_k\times \S_{N-k}$ and the one-dimensional space  $\spn\{|1^k0^{N-k}\>_{I}\}$.
$V_{I}$ is a representation of $\S_{k}\times \S_{N-k}$ and, since it acts trivially on $1^k0^{N-k}$, we have
\(
V_I\cong(k)\times(N-k).
\)
And, since
\[
{|\S_N|}\big/{|\S_k\times \S_{N-k}|}
 = {|D|}\big/{|\{1^k0^{N-k}\}|},
\]
 $U_{I}$ is equal to the induced representation when we induce $V_{I}$ from
$\S_k\times \S_{N-k}$ to $\S_N$. For shortness, we write $U_{I}=V_{I}\uparrow \S_N$.
 The Littlewood-Richardson rule then implies
 \begin{equation}
 \label{eqn:JohnsonIrreps}
  ((k)\times(N-k))\uparrow \S_N = (N) \+ (N-1,1) \+ (N-2,2) \+ \ldots \+ (N-k,k).
 \end{equation}
Hence, we have
\begin{equation*}
\cH_{I} = \bigoplus\nolimits_{i=0}^k \cH_{I}^{(N-i,i)},
\end{equation*}
where $U_{I}$ restricted to $\cH_{I}^{(N-i,i)}$ is an irrep of $\S_N$ corresponding to the partition $(N-i,i)$ of $N$.
It is also known (see \cite{godsil:assoc1,knuth}) that $\cH_{I}^{(N-i,i)}=\cT^i_{I}\cap(\cT^{i-1}_{I})^\bot$, where $\cT^i_{I}$ is the space spanned by all $\binom{N}{i}$ states
\begin{equation}
\label{eqn:TSpanStates}
|\psi_{x_1,\ldots,x_i}\> =
\frac{1}{\sqrt{\binom{N-i}{k-i}}}
\sum_{\substack{z_y\in D\\z_{y,x_1}=\ldots=z_{y,x_i}=1}}
|z_y\>
\end{equation}
(the value of $y$ is irrelevant here).
When $i=0$, let us denote this state by $|\psi_\emptyset\>$.
\subsection{Framework for the proof} \label{sec:framework}

We use the representation-theoretic framework developed in \cite{Ambainis10} (and used in \cite{ambainis:newMethod2} and \cite{AMRR}). 
Let
\begin{alignat*}{2}
&\cH_{I,\Lzero}:=\cT^1_{I}=\cH^{(N)}_{I}\oplus\cH^{(N-1,1)}_{I},
&\qquad& \cH_{I,\Lone}:=\cH_I\cap(\cH_{I,\Lzero})^\bot, \\
& \cH_{\bI,\Lzero}:=\cH_{I,\Lzero}^{\otimes M},
&& \cH_{\bI,\Lone}:=\cH_\bI\cap(\cH_{\bI,\Lzero})^\bot.
\end{alignat*}
And let $\Pi_{I,\Lzero}$, $\Pi_{I,\Lone}$, $\Pi_{\bI,\Lzero}$, and  $\Pi_{\bI,\Lone}$ denote the projections to the spaces $\cH_{I,\Lzero}$, $\cH_{I,\Lone}$, $\cH_{\bI,\Lzero}$, and  $\cH_{\bI,\Lone}$, respectively.

Recall that $\rho_q$ is the density matrix of the $\cH_\bI$ part of the state of the algorithm
anywhere between $q$-th and $(q+1)$-th oracle calls (interpreting $(-1)$-st and $(q_T+1)$-th oracle calls as the beginning and the end of the algorithm, respectively). 
Recall that $\rho_q$ is fixed under the action of $\W$---for all $\tau\in\W$, we have 
$U_{\bI,\tau} \rho_q U_{\bI,\tau}^{-1} = \rho_q$---and so are $\Pi_{\bI,\Lzero}$ and $\Pi_{\bI,\Lone}$.
Let 
\[
p_{\Lzero,q} := %
\Tr(\rho_q\Pi_{\bI,\Lzero})
\qquad\text{and}\qquad
p_{\Lone,q} := 1-p_{\Lzero,q} = \Tr(\rho_q\Pi_{\bI,\Lone}).
\]
\autonameref{theo:pickone.sound} then follows from the following three lemmas.

\begin{lemma} \label{lem:FinalBound}
The success probability of the algorithm is at most $\frac{2(k-1)}{N-1}+\sqrt{2p_{\Lone,q_T}}$.
\end{lemma}

\begin{lemma} \label{lem:InitialBound}
(At the very beginning of the algorithm) we have $p_{\Lone,0} < h^2/(2M)$.
\end{lemma}

\begin{lemma} \label{lem:StepBound}
For all $q\in\{0,\ldots,q_T-1\}$, we have $|p_{\Lone,q}-p_{\Lone,q+1}|= O(\max\{\sqrt{k/N},\sqrt{1/k}\})$.
\end{lemma}

One can see that $M$, the size of the set $Y$, does not appear in the statements of Lemmas \ref{lem:FinalBound} and \ref{lem:StepBound}. The size of $Y$ indeed does not matter for them, as in we will eventually reduce the general case for Lemmas \ref{lem:FinalBound} and \ref{lem:StepBound} to the case when $|Y|=1$.

\subsection{Proof of Lemma \ref{lem:InitialBound}} \label{sec:InitialBound}

Let us rewrite 
(\ref{eqn:xiResource})
as
\begin{align*}
|\xi'(z)\>_R
& =\bigotimes_{\ell=1}^h\Big(\frac{1}{\sqrt{M}}
\sum_{y\in Y}|y\>_{R_Y(\ell)}\big(
\alpha_{\ell,0}|\Psi(z_y)\>+\alpha_{\ell,1}|\Phi\>\big)_{R_X(\ell)}
\Big) \\
& = \frac{1}{\sqrt{M^h}}
\sum_{y_1,\ldots,y_h\in Y}|y_1,\ldots,y_h\>_{R_Y}
|\xi'(y_1,\ldots,y_h)\>_{R_X},
\end{align*}
where $|\Phi\>:=\sum_{x\in X}|x\>/\sqrt{|X|}$ and
\[
|\xi'(y_1,\ldots,y_h)\>_{R_X}:=\bigotimes\nolimits_{\ell=1}^h
(\alpha_{\ell,0}|\Psi(z_{y_\ell})\>_{R_X(\ell)}+\alpha_{\ell,1}|\Phi\>_{R_X(\ell)})
\]
has unit norm for $\<\Psi(z_y)|\Phi\>=\<\SPsi(z)|\SPhi\>=\sqrt{k/N}$.
Let $Y_h$ be the set of all $(y_1,\ldots,y_h)\in Y^h$ such that $y_\ell\neq y_{\ell'}$ whenever $\ell\neq\ell'$.
Let us write $|\xi'(z)\>_R=|\xi'_\Lzero(z)\>_R+|\xi'_\Lone(z)\>_R$, where the unnormalized state $|\xi'_\Lzero(z)\>_R$ corresponds to all $(y_1,\ldots,y_h)\in Y_h$ in the register $\cH_{R_Y}$. 
Then, $\||\xi'_\Lone(z)\>\|^2$ equals the probability that among $h$ numbers chosen independently and uniformly at randomly from $\{1,\ldots,M\}$ at least two numbers are equal. Analysis of the birthday problem tells us that this probability is at most $h(h-1)/(2M)$ \cite{katz:crypto}. For $c\in \{\Lzero,\Lone\}$,
let
\[
|\phi_{c}\> :=
\binom{N}{k}^{-M/2}\sum_{z\in D^m}|z\>_\bI
|\W\>_S
|\xi'_c(z)\>_R|\xi''\>_{A-R},
\]
and note that $\||\phi_{c}\>\|=\||\xi'_c(z)\>\|$. The initial state of the algorithm is $|\phi_{\Lzero}\>+|\phi_{\Lone}\>$.
(Note: in this proof, the subscript of $\phi$ does not denote the number of queries.)

\begin{claim} \label{clm:InitialBound}
We have $(\Pi_{\bI,\Lzero}\otimes \Id_{SA})|\phi_{\Lzero}\>=|\phi_{\Lzero}\>$.
\end{claim}

Claim \ref{clm:InitialBound} implies that 
$(\Pi_{\bI,\Lone}\otimes \Id_{SA})|\phi_{\Lzero}\>=0$,
and therefore
\begin{multline*}
p_{\Lone,0}= \Tr(\rho_0\Pi_{\bI,\Lone})
 = \Tr\Big(\Tr_{SA}\big(
(|\phi_{\Lzero}\>+|\phi_{\Lone}\>)(\<\phi_{\Lzero}|+\<\phi_{\Lone}|)
\big)\Pi_{\bI,\Lone}\Big)
\\
= \<\phi_{\Lone}|(\Pi_{\bI,\Lone}\otimes \Id_{SA})|\phi_{\Lone}\>
\leq \<\phi_{\Lone}|\phi_{\Lone}\> < h^2/(2M).
\end{multline*}

\begin{proof}[of Claim \ref{clm:InitialBound}]
First, let $|\Omega_0(z_{y_\ell})\>:=|\Psi(z_{y_\ell})\>$ and $|\Omega_1(z_{y_\ell})\>:=|\Phi\>$,
so that
\[
|\xi'(y_1,\ldots,y_h)\>_{R_X}
=
\sum_{\beta=(\beta_1,\ldots,\beta_h)\in\{0,1\}^h}
(\alpha_{1,\beta_1}\ldots\alpha_{h,\beta_h})
|\Omega_{\beta_1}(z_{y_1}),\ldots,\Omega_{\beta_h}(z_{y_h})\>_{R_X}.
\]
For all $\beta\in\{0,1\}^h$ and all $(y_1,\ldots,y_h)\in Y_h$, let
\begin{equation}
\label{eqn:phizero_gamma}
|\phi_{\Lzero,\beta}(y_1,\ldots,y_h)\>:=
\gamma
\sum_{z\in D^m}|z\>_\bI
|\W\>_S
|y_1,\ldots,y_h\>_{R_Y}
|\Omega_{\beta_1}(z_{y_1}),\ldots,\Omega_{\beta_h}(z_{y_h})\>_{R_X}
|\xi''\>_{A-R},
\end{equation}
where
$\gamma=\binom{N}{k}^{-M/2} (\alpha_{1,\beta_1}\ldots\alpha_{h,\beta_h})\big/\sqrt{M^h}$.
We have
\[
|\phi_\Lzero\>=\sum_{\beta\in\{0,1\}^h}\sum_{(y_1,\ldots,y_h)\in Y_h}|\phi_{\Lzero,\beta}(y_1,\ldots,y_h)\>,
\]
and it is enough to show that
\[
(\Pi_{\bI,\Lzero}\otimes \Id_{SA})
|\phi_{\Lzero,\beta}(y_1,\ldots,y_h)\>=
|\phi_{\Lzero,\beta}(y_1,\ldots,y_h)\>
\]
for all $\beta\in\{0,1\}^h$ and $(y_1,\ldots,y_h)\in Y_h$. 

Notice that, if $\beta_\ell=1$, then the register $\cH_{R_X(\ell)}$ contains the state $|\Phi\>$ and this register is not entangled with any the other registers. 
Therefore, it suffices to consider the case when $\beta=0^h$.
Without loss of generality, let $(y_1,\ldots,y_h)=(1,\ldots,h)$.

For simplicity, let
$|\hat\phi\>$ be the the state $|\phi_{\Lzero,0^k}(1,\ldots,h)\>/\gamma$ restricted to registers $\cH_\bI$ and $\cH_{R_X}$, for these registers are not entangled with the other registers and we have
\[
\Tr_{SA}(|\phi_{\Lzero,0^h}(1,\ldots,h)\>\<\phi_{\Lzero,0^h}(1,\ldots,h)|) = \gamma^2 
\Tr_{R_X}(|\hat\phi\>\<\hat\phi|).
\]
We have
\[
|\hat\phi\>=
\sum_{z\in D^m}|z\>_\bI
|\Psi(z_1),\ldots,\Psi(z_h)\>_{R_X}
=
\bigotimes_{y=1}^h
\Big(\sum_{z_y\in D}|z_y\>_{I(y)}|\Psi(z_y)\>_{R_X\!(y)}\Big)
\otimes
\bigotimes_{y=h+1}^M
\Big(\sum_{z_y\in D}|z_y\>_{I(y)}\Big).
\]
Recall the states $|\psi_{x_1,\ldots,x_i}\>\in\cH_I$ from (\ref{eqn:TSpanStates}).
We have
\begin{align*}
&
\sum_{z_y\in D}|z_y\>|\Psi(z_y)\>
\propto
\sum_{z_y\in D}|z_y\>\sum_{\substack{x\in X\\z_{y,x}=1}}|x\>
=
\sum_{x\in X} \Big(
\sum_{\substack{z_y\in D\\z_{y,x}=1}}|z_y\>
\Big) |x\>
\propto
\sum_{x\in X} |\psi_x\>|x\>
\in \cT_{I(y)}^1\otimes \cH_{R_X\!(y)};
\\
& \sum_{z_y\in D}|z_y\>\propto|\psi_\emptyset\>\in\cT_{I(y)}^0=\cH_{I(y)}^{(N)}.
\end{align*}
The claim follows from the definition of $\cH_{\bI,\Lzero}$ (Section \ref{sec:framework}).
\end{proof}

\subsection{Proof of Lemma \ref{lem:FinalBound}}
\paragraph{Reduction to the $p_{q_T,\Lone}=0$ case.}
 
Let us first reduce the lemma to its special case when $p_{q_T,\Lone}=0$.
This reduction was used in \cite{Ambainis10} for a very similar problem. %
Recall that the final state of the algorithm $|\phi_{q_T}\>$ satisfies the symmetry
$(U_{\bI SO,\tau}\otimes\Id_{A-O})|\phi_{q_T}\>=|\phi_{q_T}\>$ for all $\tau\in\W$, and note that,
for $c\in\{\Lzero,\Lone\}$, the state 
\[
|\phi_{q_T}^{c}\>:=\frac
{   (\Pi_{\bI,c}\otimes \Id_{SA})|\phi_{q_T}\>}
{\|(\Pi_{\bI,c}\otimes \Id_{SA})|\phi_{q_T}\>\|}
=\frac{1}{\sqrt{p_{c,q_T}}}
{   (\Pi_{\bI,c}\otimes \Id_{SA})|\phi_{q_T}\>}
\]
satisfies the same symmetry. We have
\[
|\phi_{q_T}\>=
\sqrt{1-p_{\Lone,q_T}}|\phi_{q_T}^{\Lzero}\> +
\sqrt{p_{\Lone,q_T}}|\phi_{q_T}^{\Lone}\>.
\]
Since $|\phi_{q_T}^{\Lzero}\>$ and $|\phi_{q_T}^{\Lone}\>$ are orthogonal, we have
\begin{equation}\label{eqn:varDistance}
\||\phi_{q_T}\>-|\phi^{\Lzero}_{q_T}\>\|
=
\sqrt{(1-\sqrt{1-p_{\Lone,q_T}})^2+(\sqrt{p_{\Lone,q_T}})^2}
\leq\sqrt{2p_{\Lone,q_T}}
\end{equation}

 From now on, let us assume that $p_{\Lone,q_T}=0$ and, thus, 
$|\psi_{q_T}\>=|\psi^{\Lzero}_{q_T}\>$. \autoref{lemma:dist.similar3} and (\ref{eqn:varDistance}) states that this changes the success probability by at most $\sqrt{2p_{\Lone,q_T}}$.

\paragraph{Reduction to the $|Y|=1$ case.}

Recall that $\rho''_{q_T}=\Tr_{S,A-O}|\phi_{q_T}\>\<\phi_{q_T}|$, and we have
\begin{equation*}
(\Pi_{\bI,\Lzero}\otimes \Id_O)\rho''_{q_T} = \rho''_{q_T}
\quad\text{and}\quad
\forall\tau\in\W\colon
 U_{\bI O,\tau}
 \rho''_{q_T}
 U_{\bI O,\tau}^{-1}
=
 \rho''_{q_T}.
\end{equation*}
The algorithm makes its final measurement of the $\cH_\bI$ and $\cH_O$
registers, ignoring all the other registers, therefore the success
probability is completely determined by $\rho''_{q_T}$.
Let us assume that the algorithm measures (and then discards) the $\cH_{O_Y}$ register first, before measuring $\cH_\bI$ and $\cH_{O_{X2}}$, and that the outcome of this measurement is $y\in Y$. Due to the symmetry, we get each outcome $y$ with the same probability $1/M$.

Now the algorithm can discard the registers $\cH_{I(y')}$ for all $y'\neq y$, as their content do not affect the success probability. We are left with
\[
\rho''_{q_T,y} = M \Tr_{I(y')\colon y'\neq y}
\big(
(\Id_{\bI O_{X2}}\otimes\<y|_{O_Y})
\rho''_{q_T}
(\Id_{\bI O_{X2}}\otimes|y\>_{O_Y})
\big),
\]
which is a density matrix on the registers $\cH_{I(y)}$ and $\cH_{O_{X2}}$, and it satisfies
\begin{gather*}
(\Pi_{I,\Lzero}\otimes \Id_{O_{X2}})
\rho''_{q_T,y} = \rho''_{q_T,y}
,\\
\forall\sigma\in\S_X\colon
(U_{I,\sigma}\otimes U_{O_{X2},\sigma})
 \rho''_{q_T,y}
(U_{I,\sigma}\otimes U_{O_{X2},\sigma})^{-1}
=
 \rho''_{q_T,y}
\end{gather*}
(we use the subscript $I$ instead of $I(y)$ as $y$ is fixed from now on).
The success probability of the algorithm equals the probability that we measure the state $\rho''_{q_T,y}$ in the computational basis and obtain $z_y\in D$ and $\{x_1,x_2\}\in X_2$ such that $z_{y,x_1}=z_{y,x_2}=1$. Hence, we have reduced the proof to the case when $|Y|=1$.

\paragraph{The $|Y|=1$ case.}

Since $y\in Y$ is fixed, to lighten the notation, in the remainder of the proof of  Lemma \ref{lem:FinalBound}, let us write $z'$ instead of $z_y$ and $z'_x$ instead of $z_{y,x}$.

Let us now assume that the algorithm measures the $\cH_{O_{X2}}$ register, obtaining $\{x_1,x_2\}\in X_2$, and only then measures $\cH_{I}$. Due to the symmetry, the measurement yields each outcome $\{x_1,x_2\}$ with the same probability $1/\binom{N}{2}$, and let
\[
\hat\rho:=\binom{N}{2}
(\Id_{I}\otimes\<\{x_1,x_2\} |_{O_{X2}})
\rho''_{q_T,y}
(\Id_{I}\otimes | \{x_1,x_2\}\>_{O_{X2}})
\]
be the density matrix of the register $\cH_{I}$ after the measurement. Without loss of generality,
let $\{x_1,x_2\}=\{1,2\}$, and let $\hat\S:=\S_{\{1,2\}}\times\S_{\{3,\ldots,N\}}<\S_X$ be the group of all permutations  $\sigma\in\S_X$ that map $\{1,2\}$ to itself. Now we have
\begin{equation} \label{eqn:rhoHatSymmetry}
\Pi_{I,\Lzero}
\hat\rho = \hat\rho
\quad\text{and}\quad
\forall\sigma\in\hat\S\colon
 U_{I,\sigma}
 \hat\rho
U_{I,\sigma}^{-1}
=
 \hat\rho.
\end{equation}
Let $\hat\Pi$ denote
the projection to the subspace of $\cH_{I}$ spanned by all $|z'\>$ such that $z'_{1}=z'_{2}=1$. 
We note that
$U_{I,\sigma}\hat\Pi U_{I,\sigma}^{-1}= \hat\Pi$ for all $\sigma\in \hat\S$. 
One can see that the success probability of the algorithm is $\Tr(\hat\Pi\hat\rho)$,
and it is left to show

\begin{claim}
$\Tr(\hat\Pi\hat\rho)\leq{2(k-1)}/{(N-1)}$.
\end{claim}

\begin{proof}
We can express $\hat\rho$ as a mixture of its eigenvectors $\ket{\chi_i}$, with 
probabilities that are equal to their eigenvalues $\chi_i$:
\(\hat\rho = \sum_i \chi_i \ket{\chi_i} \bra{\chi_i} .\)
Hence we have
\[
 \Tr (\hat\Pi \hat\rho) = \sum\nolimits_i \chi_i \Tr(\hat\Pi \ket{\chi_i}\bra{\chi_i})
 = \sum\nolimits_i \chi_i\| \hat\Pi \ket{\chi_i} \|^2,
 \]
which is at most
\[
\max\nolimits_{|\chi\>}
\big(\|\hat\Pi|\chi\>\|^2\big/\||\chi\>\|^2\big)
\]
where the maximization is over all eigenvectors of $\hat\rho$ with non-zero eigenvalues.
Due to the symmetry (\ref{eqn:rhoHatSymmetry}), we can calculate the eigenspaces of $\hat\rho$ by inspecting the restriction of $U_{I}$ to the subspace $\cT^1_{I}$, namely,  $\hat{U}_{I}:=\Pi_{I,\Lzero} U_{I}$.
Recall that we defined $\cT_I^1$ to be the space spanned by all
\begin{equation*}
|\psi_{x}\> =
\frac{1}{\sqrt{\binom{N-1}{k-1}}}
\sum_{\substack{z'\in D\\z'_{x}=1}}
|z'\>.
\end{equation*}
We note that $\< \psi_{x_1} | \psi_{x_2} \> = \frac{k-1}{N-1}$ for all $x_1, x_2: x_1\neq x_2$.

Both $U_{I}$ and $\hat U_{I}$ are representations of both $\S_X$ and its subgroup $\hat\S$.
We already studied $U_{I}$ as a representation of $\S_X$ in Section \ref{sec:repTheory}.
Since $\cT^1_{I}=\cH_{I}^{(N)}\oplus\cH_{I}^{(N-1,1)}$, the representation $\hat U_{I}$ of $\S_X$ consists of only two irreps: one-dimensional $(N)$ and $(N-1)$-dimensional $(N-1,1)$, which correspond to the spaces $\cH_{I}^{(N)}$ and $\cH_{I}^{(N-1,1)}$, respectively.

In order to see how 
$\hat{U}_{I}$ decomposes into irreps of $\hat\S$, we need to restrict $(N)$ and $(N-1,1)$ from $\S_N$ to $\S_2\times\S_{N-2}$. The Littlewood-Richardson rule gives us the decomposition of these restrictions:
\begin{alignat*}{2}
& (N)\downarrow(\S_2\times\S_{N-2})&&=((2)\times(N-2)); \\
& (N-1,1)\downarrow(\S_2\times\S_{N-2})&&=((2)\times(N-2)) \+ ((1,1)\times(N-2)) \+ ((2)\times(N-3,1)).
\end{alignat*}
Hence, Schur's lemma and (\ref{eqn:rhoHatSymmetry}) imply that that eigenspaces of $\hat\rho$ 
are invariant under $U_{I,\sigma}$ for all $\sigma\in\hat\S$, and they
 have one of the following forms:
\begin{enumerate}
\item %
one-dimensional subspace spanned by 
$\ket{\psi(\alpha,\beta)} = \alpha(\ket{\psi_1}+\ket{\psi_2})+\beta\sum_{x=3}^N \ket{\psi_x}$ for some coefficients $\alpha,\beta$;
\item %
one-dimensional subspace spanned by $\ket{\psi_1}-\ket{\psi_2}$;
\item %
$(N-3)$-dimensional subspace consisting of all 
$\sum_{i=3}^{N} \alpha_x \ket{\psi_x}$ with $\sum_x \alpha_x =0$ (spanned by all
$\ket{\psi_x}-\ket{\psi_{x'}}$, $x, x'\in \{3, \ldots, N\}$);
\item a direct sum of subspaces of the above form.
\end{enumerate}

In the first case, 
\[ \hat\Pi \ket{\psi(\alpha,\beta)} = \frac{2\alpha+(k-2)\beta}{\sqrt{{N-1 \choose k-1}}} 
\sum_{\substack{z'_3, \ldots, z'_N\in\{0,1\}\\z'_3+\ldots+z'_N=k-2}} \ket{1,1,z'_3,\ldots, z'_N} .\]
Therefore,
\[
 \|\hat\Pi\ket{\psi(\alpha,\beta)}\|^2 = \frac{{N-2 \choose k-2}}{{N-1 \choose k-1}}
 \babs{2\alpha+(k-2)\beta}^2 = \frac{k-1}{N-1}
 \babs{2\alpha+(k-2)\beta}^2 .
\]
We also have
\begin{multline} \label{eqn:zetaabnorm}
 \| |\psi(\alpha,\beta)\rangle\|^2 = \langle \psi(\alpha,\beta) |
\psi(\alpha,\beta) \rangle \\
= 2\left(1+\frac{k-1}{N-1} \right) |\alpha|^2 +
(N-2)\left(1 + (N-3) \frac{k-1}{N-1} \right) |\beta|^2 +
2(N-2) \frac{k-1}{N-1}(\alpha \beta^*+\beta \alpha^*) \\
 \geq \frac{|2\alpha+(k-2)\beta|^2}{2}.
\end{multline}
If $\alpha \beta^*\geq 0$, the inequality in (\ref{eqn:zetaabnorm})
follows by showing that
coefficients of $|\alpha|^2$, $|\beta|^2$, and $\alpha \beta^*$ on the
left hand side
are all larger than corresponding coefficients on the right hand side.
Otherwise, without loss of generality, we can assume that $\alpha=1$
and $\beta< 0$, and the inequality follows by inspecting the extreme
point of the quadratic
 polynomial (in $\beta$) that is obtained by subtracting the right
hand side from the left hand side.
Therefore, 
\[ \frac{\|\hat\Pi \ket{\psi({\alpha, \beta})} \|^2}{\|\ket{\psi({\alpha, \beta})} \|^2}
\leq \frac{2(k-1)}{N-1}.\]

In the second case, 
$\Pi (\ket{\psi_{1}}-\ket{\psi_2}) = 0$ because basis states 
$\ket{1,1,z'_3,\ldots, z'_N}$ 
have the same amplitude in $\ket{\psi_1}$ and $\ket{\psi_2}$. 

In the third case, it suffices to consider a state of the form $\ket{\psi_{3}}-\ket{\psi_4}$, 
because $\{U_{I,\sigma}(\ket{\psi_3}-\ket{\psi_4})\colon \sigma\in \hat\S\}$ spans the whole eigenspace and $\hat\Pi$ and $U_{I,\sigma}$ commute.
Then, 
\[ \hat\Pi (\ket{\psi_3}-\ket{\psi_4}) 
= \frac{1}{\sqrt{{N-1 \choose k-1}}} \sum_{\substack{z'_5, \ldots, z'_N\in\{0,1\}\\z'_5+\ldots+ z'_N=k-3}} 
(\ket{1,1,1,0, z'_5,\ldots, z'_N} - \ket{1,1,0,1, z'_5,\ldots ,z'_N}) \]
and 
\[ \| \hat\Pi (\ket{\psi_3}-\ket{\psi_4}) \|^2 = 2 \frac{{N-4 \choose k-3}}{{N-1 \choose k-1}} = 
2\frac{(k-1)(k-2)(N-k)}{(N-1)(N-2)(N-3)}
.\] 
We also have 
\[ \| \ket{\psi_3}-\ket{\psi_4} \|^2 = 2- \lbra \psi_3 |\psi_4\rket = 2- 2\frac{k-1}{N-1} 
= 2 \frac{N-k}{N-1} .\]
Hence, 
\[ \frac{\|\hat\Pi (\ket{\psi_3}-\ket{\psi_4}) \|^2}{\| \ket{\psi_3}-\ket{\psi_4} \|^2} = 
\frac{(k-2)(k-3)}{(N-2)(N-3)}
 = 
O\left(\frac{k^2}{N^2} \right)  .\]
\end{proof}

\subsection{Reduction of Lemma \ref{lem:StepBound} to the $|Y|=1$ case}

First, instead of the oracle $\cO_V$ given by \eqref{eqn:quantumOV}, we define
\[
\cO_V(z)|y,x,b\>:= (-1)^{b\cdot z_{y,x}}|y,x,b\>.
\]
Both definitions are equivalently powerful as one is obtained from another by two Hadamard gates on the register $\cH_{B}$.

For all $z_y\in D$, let 
\[
\cO''_V(z_y):=\Id_{Q_X}-2\sum_{\substack{x\in X\\z_{y,x}=1}}|x\>\<x|_{Q_X}
\qquad\text{and}\qquad
\cO''_F(z_y):=\Id_{Q_X}-2|\Psi(z_y)\>\<\Psi(z_y)|_{Q_X}
\]
act on $\cH_{Q_X}$,
so that we have
\begin{equation} \label{eqn:OVOF}
\begin{split}
&
\cO'_V=
\sum\nolimits_{z\in D^M}|z\>\<z|_\bI\otimes
\sum\nolimits_{y\in Y}|y\>\<y|_{Q_Y}
\otimes\cO''_V(z_y)\otimes|1\>\<1|_B + \Id_{\bI Q}\otimes|0\>\<0|_B, \\
&
\cO'_F=
\sum\nolimits_{z\in D^M}|z\>\<z|_\bI\otimes
\sum\nolimits_{y\in Y}|y\>\<y|_{Q_Y}
\otimes\cO''_F(z_y)\otimes \Id_B.
\end{split}
\end{equation}
Let
\[
\rho'''_q =
\rho'_{q,00}\otimes|0\>\<1|_B +
\rho'_{q,01}\otimes|0\>\<1|_B +
\rho'_{q,10}\otimes|1\>\<0|_B +
\rho'_{q,11}\otimes|1\>\<1|_B
\]
be the state of the algorithm corresponding to the $\cH_\bI$, $\cH_Q$, and $\cH_B$ registers right before the $(q+1)$-th oracle call ($\cO_V$ or $\cO_F$). Note that $\rho_q=\Tr_{QB}(\rho'''_q)$ and, since oracles are the only gates of the algorithm that interact with the $\cH_\bI$ register,
\(
\rho_{q+1}=\Tr_{QB}(\cO'\rho'''_q\cO').
\)

Notice that $|p_{\Lone,q}-p_{\Lone,q+1}|=|p_{\Lzero,q}-p_{\Lzero,q+1}|$, therefore let us deal with 
$p_{\Lzero,q}$ instead. We have 
\begin{equation}\label{eqn:pChange}
|p_{\Lzero,q}-p_{\Lzero,q+1}| = \Tr\big(\Pi_{\bI,\Lzero}(\rho_q-\rho_{q+1})\big) = 
\Tr\big((\Pi_{\bI,\Lzero}\otimes \Id_{QB})(\rho'''_q-\cO'\rho'''_q\cO')\big),
\end{equation}
which for the oracle $\cO_V$ equals
\[ 
\Tr\big((\Pi_{\bI,\Lzero}\otimes \Id_{Q})(\rho'_{q,11}-
\tilde\cO'_V
\rho'_{q,11}
\tilde\cO'_V
)\big),
\]
where $\tilde\cO'_V=(\Id_{\bI Q}\otimes\<1|_B)\cO'_V(\Id_{\bI Q}\otimes|1\>_B)$.
Therefore, without loss of generality, we assume that the state of $\cH_B$ is always $|1\>$ throughout the execution of the algorithm. 
In turn, we assume that $\cO'_V$ and $\cO'_F$ in (\ref{eqn:OVOF}) act only on $\cH_\bI\otimes\cH_Q$, and we take $\rho'_q$ instead of $\rho'''_q$ and $\Id_Q$ instead of $\Id_{QB}$ in (\ref{eqn:pChange}).

Since  $(\tau(z))_{\tau(y,x)}=1$ if and only if $z_{x,y}=1$, we have
\(
U_{\bI Q,\tau} \cO' U_{\bI Q,\tau}^{-1} = \cO'
\)
for all $\tau\in \W$, and recall that the same symmetry holds for $\rho'_q$, namely, (\ref{eqn:stepSym}). Hence, for all $y\in Y$,
\[
\rho'_{q,y} = M
(\Id_{\bI Q_X}\otimes\<y|_{Q_Y})
\rho'_q
(\Id_{\bI Q_X}\otimes|y\>_{Q_Y})
\]
has trace one
and (\ref{eqn:pChange}) equals
\begin{multline}
\label{eqn:pChange2}
M \Tr\big(
(\Id_{\bI Q_X}\otimes\<y|_{Q_Y})
(\Pi_{\bI,\Lzero}\otimes \Id_{Q})(\rho'_q-\cO'\rho'_q\cO')
(\Id_{\bI Q_X}\otimes|y\>_{Q_Y})
\big) \\
=
 \Tr\bigg(
(\Pi_{\bI,\Lzero}\otimes \Id_{Q_X})
\Big(\rho'_{q,y}-
\big(\sum_{z\in D^M}|z\>\<z|_\bI\otimes\cO''(z_y)\big)
\rho'_{q,y}
\big(\sum_{z\in D^M}|z\>\<z|_\bI\otimes\cO''(z_y)\big)
\Big)\bigg).
\end{multline}
Without loss of generality, let $y=1$, and let us write
\[
\sum\nolimits_{z\in D^M} |z\>\<z|_\bI= 
\sum\nolimits_{z_1\in D}|z_1\>\<z_1|_{I(1)}
\otimes \Id_{I}^{\otimes(M-1)}.
\]
Recall that $\Pi_{\bI,\Lzero}=\Pi_{I,\Lzero}^{\otimes M}$.
Therefore, for
\[
\hat\rho'_{q,1} := \Tr_{I(2),\ldots,I(M)}
\big(
(\Id_{I(1)} \otimes \Pi_{I,\Lzero}^{\otimes(M-1)} \otimes \Id_{Q_X})
\rho'_{q,y}
\big),
\]
(\ref{eqn:pChange}) and (\ref{eqn:pChange2}) are equal to
\begin{equation} \label{eqn:pChange3}
 \Tr\bigg(\!
(\Pi_{I(1),\Lzero}\otimes \Id_{Q_X})
\Big(\hat\rho'_{q,1}-
\big(\sum_{z_1\in D}|z_1\>\<z_1|_{I(1)}\otimes\cO''(z_1)\big)
\hat\rho'_{q,1}
\big(\sum_{z_1\in D}|z_1\>\<z_1|_{I(1)}\otimes\cO''(z_1)\big)
\Big)\bigg).
\end{equation}
Since $\hat\rho'_{q,1}$ is a positive semidefinite operator of trace at most one and it acts on $\cH_{I(1)}\otimes\cH_{Q_X}$, we have reduced the lemma to the case when $|Y|=1$.
We consider this case in Section \ref{sec:SoundnessLemma}.

\section{Proof of Lemma \ref{lem:StepBound} when
  $|Y|=1$} \label{sec:SoundnessLemma}
\label{sec:hardness-last-sec}

Since $|Y|=\{y\}$, let us use notation $\cH_Q$ instead of $\cH_{Q_X}$ to denote the register corresponding to the query index $x\in X$. Also, now we have $z=(z_y)$, so let us use $z$ instead of $z_y$ and $z_x$ instead $z_{y,x}$. Also, now we denote the permutations in $\S_N$ with $\pi$ instead of $\sigma$.

We will consider the following representations of $\S_N$:
\begin{enumerate}
\item
The computational basis of $\cH_Q$ is labeled by $x\in\{1,\ldots,N\}=X$. We define the 
action of $\pi\in \S_N$ on $\cH_Q$ via the unitary $U_{Q,\pi}|x\rangle:=|\pi(x)\rangle$. $U_Q$ is known as the {\em natural representation} of $\S_N$, and we can decompose $\cH_{Q}=\cH_Q^{(N)}\oplus\cH_Q^{(N-1,1)}$ so that $U_{Q}$ restricted to $\cH_{Q}^{(N)}$ and $\cH_{Q}^{(N-1,1)}$ are irreps of $\S_N$ isomorphic to $(N)$ and $(N-1,1)$, respectively. 
\item
The computational basis of $\cH_I$ is labeled by $z\in D$, that is, $z=(z_1,\ldots,z_N)\in\{0,1\}^N$ such that $\sum_{x=1}^Nz_i=k$. In Section \ref{sec:repTheory} we already defined and studied the representation $U_I$: for $\pi\in\S_N$,
\[
U_{I,\pi}|z_1\ldots z_N\> = U_{I,\pi}|z_{\pi^{-1}(1)}\ldots z_{\pi^{-1}(N)}\>.
\]
We showed that we can decompose
\(
\cH_{I} = \bigoplus\nolimits_{i=0}^k \cH_{I}^{(N-i,i)}
\)
so that $U_{I}$ restricted to $\cH_{I}^{(N-i,i)}$ is an irrep of $\S_N$ isomorphic to $(N-i,i)$. 
\item 
Finally, let $U:=U_Q\otimes U_I$, which
acts on $\cH:=\cH_Q\otimes\cH_I$ and is also a representation of $\S_N$.
\end{enumerate}
Let $\pro{Q}{(N)}$ and $\pro{Q}{(N-1,1)}$ denote, respectively, the projectors on $\cH_Q^{(N)}$ and $\cH_Q^{(N-1,1)}$.
$\pro{Q}{(N)}$ is the $N$-dimensional matrix with all entries equal to $1/N$, and $\pro{Q}{(N-1,1)}$ is the $N$-dimensional matrix with $1-1/N$ on the diagonal and $-1/N$ elsewhere. 
Let $\pro{I}{(N)}$, $\pro{I}{(N-1,1)}$, \ldots, $\pro{I}{(N-k,k)}$ denote, respectively, the projectors on $\cH_I^{(N)}$, $\cH_I^{(N-1,1)}$, \ldots, $\cH_I^{(N-k,k)}$. The entries of these $\binom{N}{k}$-dimensional matrices can %
be calculated using the fact that they project on the eigenspaces
of the Johnson scheme (see \cite{godsil:assoc1}). %

Let us also denote
\begin{align}
&\label{eq:Pgeq}
\pro{\cH_Q\otimes \cS_{\geq 2}}{} := 
\Id_{Q}\otimes \sum\nolimits_{j=2}^k \pro{I}{(N-j,j)} =
\big(\pro{Q}{(N)}+\pro{Q}{(N-1,1)}\big)\otimes \sum\nolimits_{j=2}^k \pro{I}{(N-j,j)},
\\
&\label{eq:Ple}
\pro{\cH_Q\otimes \cS_{< 2}}{} := 
\Id_{QI}-\pro{\cH_Q\otimes \cS_{\geq 2}}{} =
\big(\pro{Q}{(N)}+\pro{Q}{(N-1,1)}\big)\otimes 
\big(\pro{I}{(N)}+\pro{I}{(N-1,1)}\big),
\end{align}
which are equal to $\Id_Q\otimes\Pi_{I,\Lone}$ and $\Id_Q\otimes\Pi_{I,\Lzero}$, respectively.

\subsection{Statement of the lemma}

For the oracles, let us write $\cO$ instead of $\cO'$ (where $\cO=\cO_V$ or $\cO=\cO_F$). Similarly to (\ref{eqn:OVOF}), we have to consider
\begin{align*}
&\cO_V = \sum_{z\in D} \Big(\sum_{\substack{x\in X\\z_x=0}}|x\>\<x|-\sum_{\substack{x\in X\\z_x=1}}|x\>\<x|\Big)_{\!Q}\!\otimes |z\>\<z|_I, \\
&\cO_F = \sum_{z\in D} \Big(\Id-|\Psi(z)\>\<\Psi(z)|\Big)_{\!Q}\!\otimes |z\>\<z|_I,
\end{align*}
where $|\Psi(z)\>=\sum_{x\colon z_x=1}|x\>\big/\sqrt{k}$.
Note that $\cO$ acts on $\cH$ and is satisfies $U_\pi\cO U_\pi^{-1}=\cO$ for all $\pi\in\S_N$.
Equivalently to (\ref{eqn:pChange3}), it suffices to prove that
\[
\big|\Tr\big(
\pro{\cH_Q\otimes \cS_{< 2}}{}
(\rho-\cO\rho\cO)\big)\big| \leq O(\max\{\sqrt{k/N},\sqrt{1/k}\})
\]
for every density operator $\rho$ on $\cH$ that satisfies $U_\pi\rho U_\pi^{-1}=\rho$ for all $\pi\in\S_N$ and both oracles $\cO=\cO_V$ and $\cO=\cO_F$.

For a subspace $\cH'\subset\cH$ such that $\cH'$ is invariant under $U$ (i.e., under $U_\pi$ for all $\pi\in \S_N$), let
$U|_{\cH'}$ be $U$ restricted to this subspace (note: $U|_{\cH'}$ is a representation of $\S_N$).
Let $\Pi_{\cH'}$ denote the projector on $\cH'$. 
Due to Schur's lemma, there is a spectral decomposition
\[
\rho = \sum\nolimits_\mu\chi_\mu \frac{\Pi_\mu}{\dim \mu},
\] 
where $\sum_\mu\chi_\mu=1$, every $\mu$ is invariant under $U$, and $U|_\mu$ in an irrep of $\S_N$.
Hence, it suffices to show the following.
\begin{lemma}\label{lem:conj3}
For every subspace $\mu\subset \cH$ such that $U|_\mu$ is an irrep and for $\mu'$ being the subspace that $\mu$ is mapped to by $\cO_V$ or $\cO_F$, we have
\begin{equation}\label{eq:conj3}
 \frac{1}{\dim\mu}\big|\Tr(\Pi_{\cH_Q\otimes \cS_{\geq 2}}(\Pi_\mu-\Pi_{\mu'}))\big| \leq O(\max\{\sqrt{k/N},\sqrt{1/k}\}).
\end{equation}
\end{lemma}

In order to prove Lemma \ref{lem:conj3}, we need to inspect the representation $U$ in more detail.

\subsection{Decomposition of $U$}

Let us decompose $U$ into irreps. We consider two approaches how to do that. 
 That is, the list of irreps contained in $U$ cannot depend on which approach we take,
 but we can choose the way we address individual instances of irreps. For example, we will
 show that $U$ contains four instances of $(N-1,1)$, and we have as much freedom in choosing
 a projector on a single instance of $(N-1,1)$ as in choosing (up to global phase) a unit vector 
 in $\C^4$.

For an irrep $\theta$ present in $U$, let $\hat \Pi_\theta$ be a projector on the space 
 corresponding to all instances of $\theta$ in $U$.

\paragraph{Approach 1: via the tensor product of irreps.} \label{ssec:app1}

We know that $U=U_Q\otimes U_I$ and we already know how $U_Q$ and $U_I$ decomposes into 
irreps. Thus, all we need to see is how, for $j\in\{0,\ldots,k\}$, $(N)_Q\otimes(N-j,j)_I$ and $(N-1,1)_Q\otimes(N-j,j)_I$ decompose into 
 irreps (we use subscripts $Q$ and $I$ here to specify which spaces these irreps act on, namely, $\cH_Q$ and $\cH_I$, respectively, but we will drop these subscripts most of the time later). Note that $(N)\otimes(N-j,j)\cong(N-j,j)$ and $(N-1,1)\otimes(N)\cong(N-1,1)$ 
 as $(N)$ is the trivial representation. And, for $j\in\{1,\ldots,k\}$, the decomposition of 
 $(N-1,1)\otimes(N-j,j)$ is given by the following claim.
\begin{claim}
For $j\in\{1,\ldots,k\}$, we have
\begin{multline*}
(N-1,1)\otimes(N-j,j) \\ 
=(N-j+1,j-1) \+ (N-j,j) \+ (N-j,j-1,1) \+ (N-j-1,j+1) \+ (N-j-1,j,1),
\end{multline*}
where we omit the term $(N-j,j-1,1)$ when $j=1$.
\end{claim}
\begin{proof}
We use Expression 2.9.5 of \cite{james:symmetric}, which, for $j\in\{2,\ldots,k\}$, gives us
\begin{align*}
&(N-1,1)\otimes(N-j,j)
= (N-j,j)\downarrow(\S_{N-1}\times\S_1)\uparrow\S_N \ominus (N-j,j)\downarrow\S_N\uparrow\S_N \\
&\qquad= 
((N-j,j-1)\times(1))\uparrow\S_N \+ ((N-j-1,j)\times(1))\uparrow\S_N  \ominus (N-j,j) \\
 &\qquad =
(N-j+1,j-1)\+(N-j,j)\+(N-j,j-1,1) 
 \+ (N-j,j) \\
&\qquad\qquad
\+ (N-j-1,j+1) \+ (N-j-1,j,1)  \ominus (N-j,j) \\
&\qquad= 
(N-j+1,j-1)\+(N-j,j)\+(N-j,j-1,1) \+ (N-j-1,j+1) \\
&\qquad\qquad\+ (N-j-1,j,1)
\end{align*}
and, similarly, for $j=1$, gives us
\begin{align*}
(N-1,1)\otimes(N-1,1)
= & \,
(N-1,1)\downarrow(\S_{N-1}\times\S_1)\uparrow\S_N \ominus (N-1,1)\downarrow\S_N\uparrow\S_N 
 \\
= & \,
(N)\+(N-1,1)\+(N-2,2) \+ (N-2,1,1).
\mathqed
\end{align*}
\end{proof}
We can see that, for every $\ell\in\{0,1\}$ and $j\in\{0,\ldots,k\}$, the representation $(N-\ell,\ell)_Q\otimes(N-j,j)_I$ is {\em multiplicity-free}, that is, it contains each irrep at most once. 
For an irrep $\theta$ present in $(N-\ell,\ell)_Q\otimes(N-j,j)_I$, let
\[
\Pi_\theta^{(N-\ell,\ell)_Q \otimes (N-j,j)_I} := \hat \Pi_\theta\big(\pro{Q}{(N-\ell,\ell)}\otimes \pro{I}{(N-j,j)}\big),
\]
which is the projector on the unique instance of $\theta$ in $(N-\ell,\ell)_Q\otimes(N-j,j)_I$. For example, 
for $\theta=(N-1,1)$, we have projectors $\Pi_{(N-1,1)}^{(N)_Q\otimes(N-1,1)_I}$, 
$\Pi_{(N-1,1)}^{(N-1,1)_Q\otimes(N)_I}$, $\Pi_{(N-1,1)}^{(N-1,1)_Q\otimes(N-1,1)_I}$, and
$\Pi_{(N-1,1)}^{(N-1,1)_Q\otimes(N-2,2)_I}$.

\paragraph{Approach 2: via spaces invariant under queries $\cO_V$ and $\cO_F$.} \label{ssec:app2}

Let us decompose $\cH$ as the direct sum of four subspaces, each invariant under the action of $U$,
$\cO_V$, and $\cO_F$. First, let $\cH=\cH^{(0)}\+\cH^{(1)}$, where  $\cH^{(0)}$ and $\cH^{(1)}$
 are spaces corresponding to, respectively, the subsets 
\[
  H_0=\big\{(x,z)\in X\times D\,:\,z_x=0\big\}
 \qquad\text{and}\qquad
  H_1=\big\{(x,z)\in X\times D\,:\,z_x=1\big\},
\]
of the standard basis $X\times D$. Let us further decompose $\cH^{(0)}$ and 
$\cH^{(1)}$ as
\[
 \cH^{(0)} = \cH^{(0,s)}\oplus \cH^{(0,t)}
 \qquad\text{and}\qquad
 \cH^{(1)} = \cH^{(1,s)}\oplus \cH^{(1,t)},
\]
where
\[
 \cH^{(0,s)} := \spn\Big\{\sum_{x\colon z_x=0}|x,z\rangle\,:\,z\in D\Big\}
 \qquad\text{and}\qquad
\cH^{(1,s)} := \spn\Big\{\sum_{x\colon z_x=1}|x,z\rangle\,:\,z\in D\Big\},
\]
and $\cH^{(0,t)}:=\cH^{(0)}\cap(\cH^{(0,s)})^\bot$ and $\cH^{(1,t)}:=\cH^{(1)}\cap(\cH^{(1,s)})^\bot$.

Note that, for a given $z$, $\sum_{z\colon z_x=1}|x\rangle=\sqrt{k}\,|\Psi(z)\rangle$. Therefore,
the query $\cO_F$ acts on $\cH^{(1,s)}$ as the minus identity and on $\cH^{(0)}\+\cH^{(1,t)}$ as the
identity. Meanwhile, $\cO_V$ acts on $\cH^{(1)}$ as the minus identity and on $\cH^{(0)}$ as the
identity.

For every superscript $\sigma\in\{(0),(1),(0,s),(0,t),(1,s),(1,t)\}$, let $\Pi^\sigma$ be 
the projector on the space $\cH^\sigma$, and let $U^\sigma$ be the restriction of $U$ to $\cH^\sigma$. Let $V^\sigma_\pi$ be $U^\sigma_\pi$ restricted to $\pi\in \S_k\times \S_{N-k}$ and the space
\[
  \tilde\cH^\sigma:=\cH^\sigma\cap(\cH_Q\otimes|1^k0^{N-k}\rangle_I).
\]
$V^\sigma$ is a representation of $\S_{k}\times \S_{N-k}$.
One can see that 
\[
{|\S_N|}\big/{|\S_k\times \S_{N-k}|}
 = {\dim\cH^\sigma}\big/{\dim\tilde\cH^\sigma},
\]
so we have $U^\sigma=V^\sigma\uparrow \S_N$. In order to see how $U^\sigma$ decomposes
 into irreps, we need to see how $V^\sigma$ decomposes into irreps, and then apply
 the Littlewood-Richardson rule.

We have $\dim\tilde\cH^{(0,s)}=\dim\tilde\cH^{(1,s)}=1$, and it is easy to see that $V^{(0,s)}$ and $V^{(1,s)}$ act trivially on $\tilde\cH^{(0,s)}$ and $\tilde\cH^{(1,s)}$, respectively. That is,
$V^{(0,s)}\cong V^{(1,s)}\cong (k)\times (N-k)$. Now, note that
\[
\tilde\cH^{(0)} = \spn\big\{|x\rangle\otimes|1^k0^{N-k}\rangle\,:\,x\in\{k+1,\ldots,N\}\big\}.
\]
The group $\S_k$ (in $\S_k\times \S_{N-k}$) acts trivially on $\tilde\cH^{(0)}$, while and the action of 
 $\S_{N-k}$ on $\tilde\cH^{(0)}$ defines the natural representation of $\S_{N-k}$. Hence,
 $V^{(0)}\cong(k)\times((N-k)\+(N-k-1,1))$, and $V^{(0)}=V^{(0,s)}\+V^{(0,t)}$, in turn, gives us $V^{(0,t)}\cong(k)\times(N-k-1,1)$.
 Analogously we obtain $V^{(1,t)}\cong(k-1,1)\times(N-k)$.
 The decompositions of $U^{(0,s)}=V^{(0,s)}\uparrow \S_N$ and $U^{(1,s)}=V^{(1,s)}\uparrow \S_N$ into irreps are given via (\ref{eqn:JohnsonIrreps}).
For $U^{(0,t)}=V^{(0,t)}\uparrow \S_N$ and $U^{(1,t)}=V^{(1,t)}\uparrow \S_N$, the Littlewood-Richardson rule gives us, respectively,
\begin{align*}
&((k)\times(N-k-1,1))\uparrow\S_N = (N-1,1) \+ (N-2,2) \+ (N-2,1,1) \\
&\hspace{40pt}\+ (N-3,3) \+ (N-3,2,1) \+ (N-4,4) \+ (N-4,3,1) \+\ldots\\
&\hspace{80pt}\+ (N-k,k) \+ (N-k,k-1,1) \+ (N-k-1,k+1) \+ (N-k-1,k,1)
\end{align*}
and
\begin{align*}
&((k-1,1)\times (N-k))\uparrow\S_N = (N-1,1) \+ (N-2,2) \+ (N-2,1,1) \\
& \hspace{40pt} \+ (N-3,3)  \+ (N-3,2,1) \+ (N-4,4) \+ (N-4,3,1)\\
& \hspace{80pt} \+\ldots\+ (N-k+1,k-1) \+ (N-k+1,k-2,1)\+ (N-k,k-1,1).
\end{align*}

Note that all $U^{(0,s)}$, $U^{(0,t)}$, $U^{(1,s)}$, and $U^{(1,t)}$ are multiplicity-free.
For a superscript $\sigma\in\{(0,s),(0,t),(1,s),(1,t)\}$ and an irrep $\theta$ present in $U^\sigma$, let $\Pi_\theta^\sigma := \hat \Pi_\theta \Pi^\sigma$, which is the projector on the unique instance of $\theta$ in $U^\sigma$.
For example, 
for $\theta=(N-1,1)$, we have all the projectors $\Pi_{(N-1,1)}^{(0,s)}$, $\Pi_{(N-1,1)}^{(0,t)}$, $\Pi_{(N-1,1)}^{(1,s)}$, and $\Pi_{(N-1,1)}^{(1,t)}$.

\subsection{Significant irreps} \label{sec:theirreps}

We noted in Section \ref{ssec:app2} that $\cO_F$ acts on $\cH^{(1,s)}$ as the minus identity and on $\cH^{(0)}\+\cH^{(1,t)}$ as the identity and $\cO_V$ acts on $\cH^{(1)}$ as the minus identity and on $\cH^{(0)}$ as the identity. This means that, if $\mu$ is a subspace of one of the spaces $\cH^{(0)}$, $\cH^{(1,s)}$, or $\cH^{(1,t)}$, then $\mu'=\mu$. In turn, 
 even if that is not the case, we still have that $U|_\mu$ and $U|_{\mu'}$ are isomorphic
 irreps.

Also note that 
\begin{equation}\label{eq:geqle_equal}
\big|\Tr(\Pi_{\cH_Q\otimes \cS_{\geq 2}}(\Pi_\mu-\Pi_{\mu'}))\big| = 
\big|\Tr(\Pi_{\cH_Q\otimes \cS_{< 2}}(\Pi_\mu-\Pi_{\mu'}))\big|.
\end{equation}
Hence we need to consider only $\mu$ such that $U|_\mu$ is isomorphic to 
 an irrep present in both
\[
\big((N)\oplus(N-1,1)\big)_Q\otimes \big((N)\oplus(N-1,1)\big)_I
\qquad\text{and}\qquad
\big((N)\oplus(N-1,1)\big)_Q\otimes \bigoplus_{j=2}^k(N-j,j)_I,
\]
as otherwise the expression (\ref{eq:conj3}) equals $0$.
From Section \ref{ssec:app1} we see that the only such irreps are $(N-1,1)$, $(N-2,2)$, 
 and $(N-2,1,1)$.

The representation $U$ contains four instances of irrep $(N-1,1)$, four of $(N-2,2)$, and 
two of $(N-2,1,1)$. Projectors on them, according to Approach 1 in Section \ref{ssec:app1}, are
\begin{equation}\label{eq:irrepsplit}
\begin{split}
& \Pi_{(N-1,1)}^{(N)_Q\otimes(N-1,1)_I},\,  \Pi_{(N-1,1)}^{(N-1,1)_Q\otimes(N)_I},\, \Pi_{(N-1,1)}^{(N-1,1)_Q\otimes(N-1,1)_I},\, \Pi_{(N-1,1)}^{(N-1,1)_Q\otimes(N-2,2)_I}, \\
& \Pi_{(N-2,2)}^{(N)_Q\otimes(N-2,2)_I},\,  \Pi_{(N-2,2)}^{(N-1,1)_Q\otimes(N-1,1)_I},\, \Pi_{(N-2,2)}^{(N-1,1)_Q\otimes(N-2,2)_I},\, \Pi_{(N-2,2)}^{(N-1,1)_Q\otimes(N-3,3)_I}, \\
& \Pi_{(N-2,1,1)}^{(N-1,1)_Q\otimes(N-1,1)_I},\,  \Pi_{(N-2,1,1)}^{(N-1,1)_Q\otimes(N-2,2)_I},
\end{split}
\end{equation}
or, according to Approach 2 in Section \ref{ssec:app2}, are
\begin{equation*}
\begin{split}
&\Pi_{(N-1,1)}^{(0,s)},\,\Pi_{(N-1,1)}^{(0,t)},\,\Pi_{(N-1,1)}^{(1,s)},\,\Pi_{(N-1,1)}^{(1,t)},\\
&\Pi_{(N-2,2)}^{(0,s)},\,\Pi_{(N-2,2)}^{(0,t)},\,\Pi_{(N-2,2)}^{(1,s)},\,\Pi_{(N-2,2)}^{(1,t)},\\
&\Pi_{(N-2,1,1)}^{(0,t)},\,\Pi_{(N-2,1,1)}^{(1,t)}.
\end{split}
\end{equation*}
One thing we can see from this right away is that, if $U|_\mu\cong(N-2,1,1)$, then
$\mu\subset\cH^{(0)}\oplus\cH^{(1,t)}$, so the application of the query $\cO_F$ fixes $\mu$,
and the expression (\ref{eq:conj3}) equals $0$.

\subsection{Necessary and sufficient conditions for irrep $(N-1,1)$}

We would like to know what are necessary and sufficient conditions for inequality (\ref{eq:conj3}) to
hold. 
First, let us consider the irrep $(N-1,1)$; later, the argument for the other two irreps will be very similar.

\paragraph{Transporters as the standard basis for irreps.}

For $a_1,a_2\in\{0,1\}$ and $b_1,b_2\in\{s,t\}$, let $\Pi_{(N-1,1)}^{(a_1,b_1)\leftarrow(a_2,b_2)}$
be, up to a global phase, the unique operator of rank $\dim(N-1,1)$ such that
\[
\big(U_{(N-1,1)}^{(a_1,b_1)}\big)_\pi = 
\Pi_{(N-1,1)}^{(a_1,b_1)\leftarrow(a_2,b_2)}
\big(U_{(N-1,1)}^{(a_2,b_2)}\big)_\pi
\big(\Pi_{(N-1,1)}^{(a_1,b_1)\leftarrow(a_2,b_2)}\big)^*.
\]
for all  $\pi\in \S_N$. 
We call  $\Pi_{(N-1,1)}^{(a_1,b_1)\leftarrow(a_2,b_2)}$ the {\em transporter} from irrep $U_{(N-1,1)}^{(a_2,b_2)}$ to $U_{(N-1,1)}^{(a_1,b_1)}$. One can see that all non-zero singular values of $\Pi_{(N-1,1)}^{(a_1,b_1)\leftarrow(a_2,b_2)}$ are $1$. We also have
\[
\Pi_{(N-1,1)}^{(a_1,b_1)\leftarrow(a_2,b_2)}
\big(\Pi_{(N-1,1)}^{(a_1,b_1)\leftarrow(a_2,b_2)}\big)^* = \Pi_{(N-1,1)}^{(a_1,b_1)}
,\qquad
\big(\Pi_{(N-1,1)}^{(a_1,b_1)\leftarrow(a_2,b_2)}\big)^*
\Pi_{(N-1,1)}^{(a_1,b_1)\leftarrow(a_2,b_2)} = \Pi_{(N-1,1)}^{(a_2,b_2)}.
\]
We can and we do choose global phases of these transporters in a consistent manner so that
\[
 \big(\Pi_{(N-1,1)}^{(a_1,b_1)\leftarrow(a_2,b_2)}\big)^*=\Pi_{(N-1,1)}^{(a_2,b_2)\leftarrow(a_1,b_1)}
\qquad\text{and}\qquad
\Pi_{(N-1,1)}^{(a_1,b_1)\leftarrow(a_2,b_2)}\Pi_{(N-1,1)}^{(a_2,b_2)\leftarrow(a_3,b_3)} =
\Pi_{(N-1,1)}^{(a_1,b_1)\leftarrow(a_3,b_3)}
\]
for all $a_3\in\{0,1\}$ and $b_3\in\{s,t\}$.
Together they imply $\Pi_{(N-1,1)}^{(a_1,b_1)\leftarrow(a_1,b_1)}=\Pi_{(N-1,1)}^{(a_1,b_1)}$.

Fix $a_3$ and $b_3$, and note that
\[
\big(\Pi_{(N-1,1)}^{(a_3,b_3)\leftarrow(a_1,b_1)}\big)^*
\Pi_{(N-1,1)}^{(a_3,b_3)\leftarrow(a_2,b_2)}
=
\Pi_{(N-1,1)}^{(a_1,b_1)\leftarrow(a_2,b_2)}
\]
is independent of our choice of $(a_3,b_3)$. Therefore, let us introduce the notation
\[
\Pi_{(N-1,1)}^{\leftarrow(a_1,b_1)} := \Pi_{(N-1,1)}^{(a_3,b_3)\leftarrow(a_1,b_1)}.
\]

\begin{fact}
Let $\mu\subset \cH$ be such that $U|_\mu$ is an irrep isomorphic to $(N-1,1)$ and let $\Pi_\mu$ be the projector on this subspace. There exists, up to a global phase, a unique vector $\gamma=(\gamma_{0,s},\gamma_{0,t},\gamma_{1,s},\gamma_{1,t})$ such that $\Pi_\mu=\bar \Pi_\gamma^*\bar \Pi_\gamma$, where 
\[
\bar \Pi_\gamma=
\big(
\gamma_{0,s}\Pi_{(N-1,1)}^{\leftarrow(0,s)} +
\gamma_{0,t}\Pi_{(N-1,1)}^{\leftarrow(0,t)} +
\gamma_{1,s}\Pi_{(N-1,1)}^{\leftarrow(1,s)} +
\gamma_{1,t}\Pi_{(N-1,1)}^{\leftarrow(1,t)}
\big).
\]
The norm of the vector $\gamma$ is $1$.
 The converse also holds: for any unit vector $\gamma$, 
$\bar \Pi_\gamma^*\bar \Pi_\gamma$ is a projector to an irrep isomorphic to $(N-1,1)$.
\end{fact}

From now on, let us work in this basis of transporters, because in this basis, queries $\cO_V$ and $\cO_F$
restricted to $\hat \Pi_{(N-1,1)}$ are, respectively, 
\[
\cO_V|_{(N-1,1)}=
\left(
\begin{array}{cccc}
1&0&0&0\\
0&1&0&0\\
0&0&-1&0\\
0&0&0&-1
\end{array}
\right)
\qquad\text{and}\qquad
\cO_F|_{(N-1,1)}=
\left(
\begin{array}{cccc}
1&0&0&0\\
0&1&0& 0\\
0&0&-1&0\\
0&0&0&1
\end{array}
\right).
\]

\paragraph{Necessary and sufficient condition for the query $\cO_V$.}

In the basis of transporters we have  
\begin{equation}\label{eq:Pmu}
\Pi_\mu =
\left(
\begin{array}{c}
\gamma_{0,s}^* \\
\gamma_{0,t}^* \\
\gamma_{1,s}^* \\
\gamma_{1,t}^* \\ \end{array}
\right)\cdot
\left(
\begin{array}{cccc}
\gamma_{0,s} & \gamma_{0,t} &
\gamma_{1,s} & \gamma_{1,t}
\end{array}
\right) =
\left(
\begin{array}{cccc}
|\gamma_{0,s}|^2 & \gamma_{0,s}^*\gamma_{0,t} &
\gamma_{0,s}^*\gamma_{1,s} & \gamma_{0,s}^*\gamma_{1,t} \\
\gamma_{0,t}^*\gamma_{0,s} & |\gamma_{0,t}|^2 &
\gamma_{0,t}^*\gamma_{1,s} & \gamma_{0,t}^*\gamma_{1,t} \\
\gamma_{1,s}^*\gamma_{0,s} & \gamma_{1,s}^*\gamma_{0,t} &
|\gamma_{1,s}|^2 & \gamma_{1,s}^*\gamma_{1,t} \\
\gamma_{1,t}^*\gamma_{0,s} & \gamma_{1,t}^*\gamma_{0,t} &
\gamma_{1,t}^*\gamma_{1,s} & |\gamma_{1,t}|^2
\end{array}
\right),
\end{equation}
and note that
\[
|\gamma_{a,b}|^2=\Tr\big(\Pi_\mu \Pi_{(N-1,1)}^{(a,b)}\big)\Big/\dim(N-1,1).
\]
From (\ref{eq:irrepsplit}), one can see that 
\[
\hat \Pi_{(N-1,1)}\Pi_{\cH_Q\otimes \cS_{\geq 2}}=\Pi_{(N-1,1)}^{(N-1,1)_Q\otimes(N-2,2)_I}.
\]
Hence, for the space $\mu$, the desired inequality (\ref{eq:conj3}) becomes
\begin{equation} \label{eq:conj3redef}
\frac{1}{\dim(N-1,1)}\Big|\Tr\big(\Pi_{(N-1,1)}^{(N-1,1)_Q\otimes(N-2,2)_I}(\Pi_\mu-\Pi_{\mu'})\big)\Big| \leq O(\max\{\sqrt{k/N},\sqrt{1/k}\}).
\end{equation}
Let us first obtain a necessary condition if we want this to hold for all $\mu$.

In the same transporter basis, let
\begin{equation}\label{Pbeta}
\Pi_{(N-1,1)}^{(N-1,1)_Q\otimes(N-2,2)_I} =
\left(
\begin{array}{cccc}
|\beta_{0,s}|^2 & \beta_{0,s}^*\beta_{0,t} &
\beta_{0,s}^*\beta_{1,s} & \beta_{0,s}^*\beta_{1,t} \\
\beta_{0,t}^*\beta_{0,s} & |\beta_{0,t}|^2 &
\beta_{0,t}^*\beta_{1,s} & \beta_{0,t}^*\beta_{1,t} \\
\beta_{1,s}^*\beta_{0,s} & \beta_{1,s}^*\beta_{0,t} &
|\beta_{1,s}|^2 & \beta_{1,s}^*\beta_{1,t} \\
\beta_{1,t}^*\beta_{0,s} & \beta_{1,t}^*\beta_{0,t} &
\beta_{1,t}^*\beta_{1,s} & |\beta_{1,t}|^2
\end{array}
\right).
\end{equation}
For $b_0,b_1\in\{s,t\}$ and a phase $\phi\in\R$, define the space $\xi_{b_0,b_1,\phi}$ via the projector on it:
\[
  \Pi_{\xi_{b_0,b_1,\phi}} := \frac{1}{2}
\big(
\Pi^{(0,b_0)}_{(N-1,1)} + e^{i\phi}\Pi^{(0,b_0)\leftarrow(1,b_1)}_{(N-1,1)} + e^{-i\phi}\Pi^{(1,b_1)\leftarrow(0,b_0)}_{(N-1,1)} + \Pi^{(1,b_1)}_{(N-1,1)}
\big).
\]
We have
\[
  \Pi_{\xi_{b_0,b_1,\phi}}-\cO_V\Pi_{\xi_{b_0,b_1,\phi}}\cO_V = 
 e^{i\phi}\Pi^{(0,b_0)\leftarrow(1,b_1)}_{(N-1,1)} + e^{-i\phi}\Pi^{(1,b_1)\leftarrow(0,b_0)}_{(N-1,1)},
\]
so, for this space, the inequality (\ref{eq:conj3redef}) becomes
\[
  \big|e^{i\phi}\beta^*_{1,b_1}\beta_{0,b_0}+e^{-i\phi}\beta^*_{0,b_0}\beta_{1,b_1}\big|\leq
 O(\max\{\sqrt{k/N},\sqrt{1/k}\}).
\]
Since this has to hold for all $b_0$, $b_1$, and $\phi$ (in particular, consider $b_0$ and $b_1$ that maximize $|\beta_{1,b_1}^*\beta_{0,b_0}|$), we must have either
\begin{equation}\label{eq:cond:necc1}
|\beta_{1,s}|^2+|\beta_{1,t}|^2 \leq O(\max\{{k/N},{1/k}\})
\qquad\text{or}\qquad
|\beta_{1,s}|^2+|\beta_{1,t}|^2 \geq 1-O(\max\{{k/N},{1/k}\}),
\end{equation}
and note that
\[
|\beta_{1,s}|^2+|\beta_{1,t}|^2 
= {\Tr\big(\Pi_{(N-1,1)}^{(N-1,1)_Q\otimes(N-2,2)_I}\!\cdot\! \Pi^{(1)}\big)}\Big/{\dim(N-1,1)}.
\]

The condition (\ref{eq:cond:necc1}) is necessary, but it is also
sufficient for \eqref{eq:conj3redef}. Because, if it holds,
then $|\beta^*_{1,b_1}\beta_{0,b_0}|\leq O(\max\{\sqrt{k/N},\sqrt{1/k}\})$ for all $b_0,b_1\in\{s,t\}$ and, clearly, $|\gamma^*_{1,b_1}\gamma_{0,b_0}|\in O(1)$ for all unit vectors $\gamma$.
Therefore, if we plug (\ref{eq:Pmu}) and (\ref{Pbeta}) into (\ref{eq:conj3redef}), the inequality
is satisfied.

\paragraph{Necessary and sufficient condition for the query $\cO_F$.}

Almost identical analysis shows that, in order for \autoref{lem:conj3} to hold when $U|_\mu$ is
isomorphic to $(N-1,1)$ and we apply $\cO_F$, it is necessary and sufficient that
\begin{equation}\label{eq:cond:necc2}
|\beta_{1,s}|^2 \leq O(\max\{{k/N},{1/k}\})
\qquad\text{or}\qquad
|\beta_{1,s}|^2 \geq 1-O(\max\{{k/N},{1/k}\}).
\end{equation}
Note that
\[
|\beta_{1,s}|^2 
= {\Tr\big(\Pi_{(N-1,1)}^{(N-1,1)_Q\otimes(N-2,2)_I}\!\cdot\! \Pi^{(1,s)}_{(N-1,1)}\big)}\Big/{\dim(N-1,1)}.
\]

\subsection{Conditions for irreps $(N-2,2)$ and $(N-2,1,1)$}

For irreps $(N-2,2)$ and $(N-2,1,1)$, let us exploit equation (\ref{eq:geqle_equal}). Mainly, we do that because the space $\cH_Q\otimes \cS_{\geq 2}$ contains three instances of irrep $(N-2,2)$, 
while $\cH_Q\otimes \cS_{< 2}$ contains only one.
From (\ref{eq:irrepsplit}) we get 
\[
\hat \Pi_{(N-2,2)}\Pi_{\cH_Q\otimes \cS_{< 2}}=\Pi_{(N-2,2)}^{(N-1,1)_Q\otimes(N-1,1)_I}
\qquad\text{and}\qquad
\hat \Pi_{(N-2,1,1)}\Pi_{\cH_Q\otimes \cS_{< 2}}=\Pi_{(N-2,1,1)}^{(N-1,1)_Q\otimes(N-1,1)_I}.
\]

\paragraph{Condition for the query $\cO_V$.}

An analysis analogous to that of the irrep $(N-1,1)$ shows that, in order for the desired inequality (\ref{eq:conj3}) to hold for query $\cO_V$ and irreps $(N-2,2)$ and $(N-2,1,1)$, it is sufficient to have
\[
\frac{\Tr\big(\Pi_{(N-2,2)}^{(N-1,1)_Q\otimes(N-1,1)_I}\!\cdot\! \Pi^{(1)}\big)}{\dim(N-2,2)}\leq O(k/N)
\qquad\text{and}\qquad
\frac{\Tr\big(\Pi_{(N-2,1,1)}^{(N-1,1)_Q\otimes(N-1,1)_I}\!\cdot\! \Pi^{(1)}\big)}{\dim(N-2,1,1)}\leq O(k/N).
\]
Let us prove this. 
Consider irrep $(N-2,2)$ and the hook-length formula gives us $\dim(N-2,2)=N(N-3)/2$. We have
\[
\Tr\big(\Pi_{(N-2,2)}^{(N-1,1)_Q\otimes(N-1,1)_I}\!\cdot\! \Pi^{(1)}\big) \leq
\Tr\big((\Pi^{(N-1,1)}_Q\otimes \Pi^{(N-1,1)}_I)\!\cdot\! \Pi^{(1)}\big),
\]
and we can evaluate the right hand side of this exactly. $\Pi^{(1)}$ is diagonal (in the standard basis),
and, on the diagonal, it has $(N-k)\binom{N}{k}$ zeros and $k\binom{N}{k}$ ones. The diagonal
entries of $\Pi^{(N-1,1)}_Q$ are all the same and equal to $\frac{N-1}{N}$. The diagonal entries of 
$\Pi^{(N-1,1)}_I$ are also all the same, because $\Pi^{(N-1,1)}_I$ projects to an eigenspace of the Johnson scheme. More precisely, we have $\Tr(\Pi^{(N-1,1)}_I)=\dim(N-1,1)=N-1$, %
therefore the diagonal entries of $\Pi^{(N-1,1)}_I$ are $(N-1)/\binom{N}{k}$. Hence, the diagonal entries of $\Pi^{(N-1,1)}_Q\otimes \Pi^{(N-1,1)}_I$ are
$(N-1)^2/(N\binom{N}{k})$, implying that
\[
\Tr\big((\Pi^{(N-1,1)}_Q\otimes \Pi^{(N-1,1)}_I)\Pi^{(1)}\big) = \frac{k(N-1)^2}{N}
\]
and, in turn,
\[
\frac{\Tr\big(\Pi_{(N-2,2)}^{(N-1,1)_Q\otimes(N-1,1)_I}\Pi^{(1)}\big)}{\dim(N-2,2)}\leq
\frac{2k(N-1)^2}{N^2(N-3)}\in O(k/N)
\]
as required. The same argument works for irrep $(N-2,1,1)$ as, by the hook-length formula, 
$\dim(N-2,1,1)=(N-1)(N-2)/2=\dim(N-2,2)+1$.

\paragraph{Condition for the query $\cO_F$.}

As we mentioned in the very end of Section \ref{sec:theirreps}, $\cO_F$ affects no space $\mu$ such that 
$U|\mu$ is isomorphic to irrep $(N-2,1,1)$. However, the following argument for irrep $(N-2,2)$ actually works for $(N-2,1,1)$ as well. We have
\[
\frac{\Tr\big(\Pi_{(N-2,2)}^{(N-1,1)_Q\otimes(N-1,1)_I}\Pi^{(1,s)}\big)}{\dim(N-2,2)}\leq
\frac{\Tr\big(\Pi_{(N-2,2)}^{(N-1,1)_Q\otimes(N-1,1)_I}\Pi^{(1)}\big)}{\dim(N-2,2)}\leq
 O(k/N),
\]
which, similarly to the condition (\ref{eq:cond:necc2}) for irrep $(N-1,1)$, is sufficient to show that
\autoref{lem:conj3} holds for irrep $(N-2,2)$ and the query $\cO_F$.

\subsection{Solution for irrep $(N-1,1)$}

Recall that conditions (\ref{eq:cond:necc1}) and (\ref{eq:cond:necc2}) are sufficient for \autoref{lem:conj3} to hold for the queries $\cO_V$ and $\cO_F$, respectively. Hence, it suffices for us to show that
\begin{multline*}
\frac{\Tr\big(\Pi_{(N-1,1)}^{(N-1,1)_Q\otimes(N-2,2)_I}\!\cdot\! \Pi^{(1)}\big)}{\dim(N-1,1)} \geq
\frac{\Tr\big(\Pi_{(N-1,1)}^{(N-1,1)_Q\otimes(N-2,2)_I}\!\cdot\! \Pi^{(1,s)}_{(N-1,1)}\big)}{\dim(N-1,1)}
=\\
=
\frac{k-1}{k}\!\cdot\!\frac{N(N-k-1)}{(N-1)(N-2)}
\geq 1-O(\max\{k/N,1/k\}).
\end{multline*}
It is easy to see that both inequalities in this expression hold, and we need to concern ourselves only with the equality in the middle.

Notice that
\[
\Pi_{(N-1,1)}^{(N-1,1)_Q\otimes(N-2,2)_I}\!\cdot\! \Pi^{(1,s)}_{(N-1,1)}
=
(\Id_{Q}\otimes \Pi_I^{(N-2,2)})\!\cdot\! \Pi^{(1,s)}_{(N-1,1)},
\]
and let us evaluate the trace of the latter. We briefly mentioned before that $\Pi_I^{(N)}$, $\Pi_I^{(N-1,1)}$, \ldots, $\Pi_I^{(N-k,k)}$ are orthogonal projectors on the eigenspaces
of the Johnson scheme. Let us now use this fact.

\paragraph{Johnson scheme on $\cH_I$.}

For any two strings $z,z'\in D$, let $|z-z'|$ be the half of the Hamming distance between them (the Hamming distance between them is an even number in the range $\{0,2,4,\ldots,2k\}$).
For every $i\in\{0,1,\ldots,k\}$, let 
\[
   A^I_i=\sum_{\substack{z,z'\in D\\|z-z'|=i}}|z\rangle\langle z'|,
\]
which is a $01$-matrix in the standard basis of $\cH_I$. Matrices $A_0^I,A_1^I,\ldots,A_k^I$ form
an association scheme known as the Johnson scheme (see \cite[Chapter 7]{godsil:assoc1}).

There are matrices $C_0^I,C_1^I,\ldots,C_k^I$ of the same dimensions as $A_i$ that satisfy
\begin{equation}\label{eq:CandA}
C_j^I=\sum_{i=0}^{k-j}\binom{k-i}{j}A_i\quad\text{for all }j
\qquad\text{and}\qquad
A_i^I=\sum_{j=k-i}^{k}(-1)^{j-k+i}\binom{j}{k-i}C_j\quad\text{for all }i.
\end{equation}
These matrices $C_j^I$ simplify the calculation of the eigenvalues of $A_i^I$, as, for all $j\in\{0,1,\ldots,k\}$, we have
\begin{equation}\label{eq:CviaP}
C_j^I=\sum_{h=0}^{j}\binom{N-j-h}{N-k-h}\binom{k-h}{j-h}\Pi_I^{(N-h,h)}\quad\text{for all }j.
\end{equation}
Hence, we can express $A_i^I$ uniquely as a linear combination of orthogonal projectors $\Pi_I^{(N-h,h)}$, and
the coefficients corresponding to these projectors are the eigenvalues of $A_i^I$.

Here, however, we are interested in the opposite: expressing $\Pi_I^{(N-h,h)}$ as a linear combination
of $A_i^I$. From (\ref{eq:CviaP}) one can see that
\begin{equation}\label{eq:PviaC}
\Pi_I^{(N-h,h)} = (N-2h+1)\sum_{j=0}^h(-1)^{j-h}\frac{\binom{k-j}{h-j}}{(k-j+1)\binom{N-j-h+1}{N-k-h}}C^I_j
\end{equation}
for $h=0,1,2$. %
 We are interested particularly in $\Pi_I^{(N-2,2)}$, and from (\ref{eq:PviaC}) and (\ref{eq:CandA}) we get
\begin{equation} \label{PAssoc1}
\Pi_I^{(N-2,2)} = \frac{1}{\binom{N-4}{k-2}}\sum_{i=0}^{k}
\bigg(
\binom{k-i}{2}-\frac{(k-1)^2}{N-2}(k-i)+\frac{k^2(k-1)^2}{2(N-1)(N-2)}
\bigg) A_i^I.
\end{equation}

\paragraph{Johnson scheme on $\cH^{(1,s)}$.}

Recall that, for $z\in D$, we have $|\Psi(z)\rangle=\sum_{x\,:\,z_x=1}|x\rangle/\sqrt{k}$, and let us define
\[
 A_i^{(1,s)} = \sum_{\substack{z,z'\in D\\|z-z'|=i}}|\Psi(z),z\rangle\langle \Psi(z'),z'|
\]
for all $i\in\{0,1,\ldots,k\}$. The matrices $A_i^I$ and $A_i^{(1,s)}$ have the same eigenvalues corresponding to the same irreps. Analogously to the space $\cH_I$, we can define matrices $C_j^{(1,s)}$ to the space $\cH^{(1,s)}$. From (\ref{eq:PviaC}) and (\ref{eq:CandA}) we get
\begin{equation} \label{PAssoc2}
 \Pi_{(N-1,1)}^{(1,s)} = \frac{1}{\binom{N-2}{k-1}}\sum_{i=0}^{k}
\bigg(
  (k-i)-\frac{k^2}{N}
\bigg) A_i^{(1,s)}.
\end{equation}

\paragraph{Both Johnson schemes together.}

Now that we have expressions for both $\Pi_I^{(N-2,2)}$ and $\Pi^{(1,s)}_{(N-1,1)}$, we can
 compute $\Tr\big((\Id_{Q}\otimes \Pi_I^{(N-2,2)})\!\cdot\! \Pi^{(1,s)}_{(N-1,1)}\big)$.
For all $i,i'\in\{0,1,\ldots,k\}$, we have
\begin{equation}\label{AOverlaps}
\Tr\big((\Id_{Q}\otimes A^I_{i})\!\cdot\! A^{(1,s)}_{i'}\big)
=\delta_{i,i'}\binom{N}{k}\binom{k}{i}\binom{N-k}{i}\frac{k-i}{k}.
\end{equation}
Indeed, it is easy to see that this trace is $0$ if $i\neq i'$, and for $i=i'$ we argue as follows. The matrix $A^I_{i}$ has $\binom{N}{k}$ rows, and each row has $\binom{k}{i}\binom{N-k}{i}$ entries $1$. That is, each $z\in D$ has exactly  $\binom{k}{i}\binom{N-k}{i}$ $z'\in D$ such that $|z-z'|=i$. And for such $z$ and $z'$, we have $\langle \psi_z|\psi_{z'}\rangle=(k-i)/k$.

Now, if we put (\ref{PAssoc1}), (\ref{PAssoc2}), and (\ref{AOverlaps}) together, we get
\begin{multline*}
\Tr\big(\Pi_{(N-1,1)}^{(N-1,1)_Q\otimes(N-2,2)_I}\!\cdot\! \Pi^{(1,s)}_{(N-1,1)}\big)
=
\Tr\big((\Id_{Q}\otimes \Pi_I^{(N-2,2)})\!\cdot\! \Pi^{(1,s)}_{(N-1,1)}\big) =\\
=
\sum_{i=0}^{k}
\frac{\Big((k-i)-\frac{k^2}{N}\Big)}
{\binom{N-2}{k-1}} %
\frac{\Big(\frac{(k-i)(k-i-1)}{2}-\frac{(k-1)^2}{N-2}(k-i) + \frac{k^2(k-1)^2}{2(N-1)(N-2)}\Big)}
{\binom{N-4}{k-2}} %
\binom{N}{k}\binom{k}{i}\binom{N-k}{i}\frac{k-i}{k},
\end{multline*}
which, by using the equality
\[
\sum_{i=0}^{k}\binom{k}{i}\binom{N-k}{i}\frac{(k-i)!}{(k-i-l)!}
=
\frac{k!}{(k-l)!}\binom{N-l}{N-k},
\]can be shown to be equal to $\frac{k-1}{k}\!\cdot\!\frac{N(N-k-1)}{(N-2)}$. We get the desired
equality by dividing this by $\dim(N-1,1)=N-1$.

This concludes the proof of \autonameref{theo:pickone.sound}.

\section{Proofs for \autoref{sec:pickone}}

\subsection{Proof of \autoref{theo:pick1.complete}}
\label{sec:proof:theo:pick1.complete}

\begin{proof}[of \autoref{theo:pick1.complete}]
  Algorithm $E_1$ measures
  the first half of $\ketpsi$. This measurement yields a uniformly
  random outcome $y\in Y$ and leaves $\ketpsiy y$ in the second half.

  Let $\OF(y):=I-2\selfbutterpsiy y$. This notation is justified
  because $\OF(y)$ is how $\OF$ operates its the second input when the
  first input is $\ket y$. In particular, given $\OF$ we can implement
  the unitary $\OF(y)$.

  The algorithm $E_2$ is as follows:\\
  {%
    \begin{algorithm}[H]
      initialize register $X$ with $\ketpsiy y$ (given as input)\;
      \For{$i=1$ \KwTo $n+1$}{
        \For{\label{line:for.j}$j=1,\dots,\ceil{\log(\pi/2\sqrt{\delta_{\min}})}$}{
          \For{\label{line:for.k} $k=1$ \KwTo $2^{j-1}$}{
            let $U_P\ket x:=(-1)^{P(x)}\ket x$\;
            apply\label{line:of.up} $\OF(y)U_P$ to register $X$
          }
          let $P_X:=\sum_{P(x)=1}\selfbutter x$\;
          measure\label{line:measure.px} register $X$ with projector $P_X$, outcome $b$\;
          \If{$b=1$}{
            measure\label{line:measure.x} register $X$ in the computational basis, outcome $x$\;
            \Return $x$\label{line:returnx}
          }
        }
      }
    \end{algorithm}}
  We first analyze the one iteration of the $j$-loop (i.e., lines
  \ref{line:for.k}--\ref{line:returnx}).  Let
  $P_y:=\{x\in\Sy y:P(x)=1\}$ and $\Bar P_y:=\{x\in\Sy y:P(x)=0\}$.
  Let
  $\ketyes:=\sum_{x\in P_y}\sqrt{1/\abs{P_y}}\,\ket
  x$\symbolindexmark{\ketyes}
  and
  $\ketno:=\sum_{x\in \Bar P_y}\sqrt{1/\abs{\Bar P_y}}\,\ket
  x$\symbolindexmark{\ketno}.
  For any $\beta\in\setR$, let
  $\ket{\phi_\beta}:=\sin\beta\ketyes+\cos\beta\ketno$.  We check that
  $U_P\ket{\phi_\beta}=\ket{\phi_{-\beta}}$.  Let
  $\gamma:=\arcsin \sqrt{\abs{P_y}/\abs{\Sy y}}$. Then
  $\ketpsiy y=\sin\gamma\ketyes+\cos\gamma\ketno=\ket{\phi_\delta}$.  Hence
  $\OF(y)\ket{\phi_\beta}=\bigl(I-2\selfbutterpsiy
  y\bigr)\ket{\phi_\beta}=\ket{\phi_{-\beta+2\gamma}}$
  for all $\beta$.  Thus $\OF(y)U_P\ket{\phi_\beta}=\ket{\phi_{\beta+2\gamma}}$.

  Assume that at \autoref{line:for.k}, we have $X=\ket{\phi_\beta}$.  The
  innermost loop (lines \ref{line:for.k}--\ref{line:of.up}) thus
  yields $X=\ket{\phi_{\beta+2^j\gamma}}$. Since $\ketyes\in\im P_X$ and
  $\ketno$ is orthogonal to $\im P_X$, measuring $X$ using $P_X$
  (\autoref{line:measure.px}) yields $b=1$ with probability
  $(\sin(\beta+2^j\gamma))^2$. If $b=1$, $X$ has state $\ketyes$, and
  if $b=0$, $X$ has state $\ketno$. Thus, if $b=1$, measuring $X$ in
  the computational basis (\autoref{line:measure.x}) yields and returns
  $x\in\Sy y$ with $P(x)=1$. 

  Summarizing so far: one iteration of the $j$-loop (i.e., lines
  \ref{line:for.k}--\ref{line:returnx}) returns $x\in\Sy y$ with
  probability $(\sin(\beta+2^j\gamma))^2$ if $X$ has state
  $\ket{\phi_\beta}$ initially. And if no such $x$ is returned, $X$ is
  in state $\ketno=\ket{\psi_0}$.

  In the first execution of the $j$-loop, $X$ contains
  $\ketpsiy y=\ket{\phi_\gamma}$. Thus in all further executions of the
  $j$-loop, $X$ contains $\ketno=\ket{\phi_0}$ and the probability of
  returning $x\in\Sy y$, $P(x)=1$ in the $j$-th iteration is 
  $(\sin 2^j\gamma)^2=1-\bigr(\sin(
  \pi/2-2^j\gamma)\bigl)^2 \geq1-(\pi/2-2^j\gamma)^2$.

  Thus any but the first iteration of the $j$-loop (i.e., lines
  \ref{line:for.k}--\ref{line:returnx}) fails to return $x\in\Sy y$
  with probability at most:
  \[
  \chi:=\min_{1\leq j\leq\ceil{\log(\pi/2\sqrt{\delta_{\min}})}} (\pi/2-2^j\gamma)^2.
  \]
  We distinguish two cases:
  \begin{compactitem}
  \item Case $\gamma>\frac\pi4$: Since also $\gamma\leq 1$, we have that
    $\abs{\pi/2-2\gamma}\leq 2-\pi/2<\frac12$ and thus
    $\chi\leq(\pi/2-2\gamma)^2\leq(\tfrac12)^2\leq\frac12$.
  \item Case $\gamma\leq\frac\pi4$:
    For at least one $1\leq j\leq\ceil{\log(\pi/2\sqrt{\delta_{\min}})}$
    we have $2^j\gamma\leq\pi/2$. And for at least one such $j$ we
    have
    \[
    2^j\gamma\geq 2^{\log\pi/2\sqrt{\delta_{\min}}}\gamma =
    \frac{\pi\gamma}{2\sqrt{\delta_{\min}}} \geq
    \frac{\pi\arcsin\sqrt{\abs{P_y}/\abs{\Sy
          y}}}{2\sqrt{\abs{P_y}/\abs{\Sy y}}} \geq \pi/2.
    \]
    Thus the minimum ranges over some $j,j+1$ such that
    $2^j\gamma\leq\pi/2\leq 2^{j+1}\gamma$.  For any $a\geq 0$,
    $\min\{\abs{\frac\pi2-a},\abs{\frac\pi2-2a}\}\leq\frac\pi6$ if
    $a\leq \frac\pi2\leq 2a$. Thus $\chi\leq(\pi/6)^2\leq\frac12$.
  \end{compactitem}
  Hence in all cases, $\chi\leq\frac12$.

  The algorithm executes the $j$-loop $n+1$ times, and each but the
  first $j$-loop fails to return $x\in\Sy y$, $P(x)=1$ with probability at most
  $\chi\leq\frac12$. Thus the algorithm fails to return $x\in\Sy y$, $P(x)=1$
  with probability at most $\chi^n\leq 2^{-n}$.
\end{proof}

\subsection{Proof of \autoref{coro:pick.one.sound-oall}}
\label{app:proof:coro:pick.one.sound-oall}

\begin{proof}[of \autoref{coro:pick.one.sound-oall}]
  We first show \eqref{item:pick.one.w}.
  Let $P_A:=\Pr[w=w_0: w\ot A^{\Oall}]$.

  In the remainder of the proof, we will make the probabilistic choice
  of oracles explicit, as well as their use by $A$. That is,
  $P_A$ becomes:
  \begin{align*}
    P_A = \Pr[w=w_0:{}&w_0\otR\bits\ellrand,
    (\Sy\com)\ot\$,
    \OS\ot\$,
    \OP\ot\$,\\
    &w\ot A^{\OE,\OP,\OR,\OS,\OF,\Opsi,\OV}].
  \end{align*}
  Here we used the following shorthands: $(\Sy\com)\ot\$$ means that
  the sets $\Sy\com$
  are uniformly random subsets of $\bits\ellch\times\bits\ellresp$
  of size $\kk$.
  $\OS\ot\$$
  means that the oracle $\OS$ is randomly chosen as described in
  \autonameref{def:ora.dist}. $\OP\ot\$$ means that the oracle
  $\OP$
  is randomly chosen as described in \autoref{def:ora.dist}.  Since no
  random choices are involved in the definitions of
  $\OE,\OR,\OF,\Opsi,\OV$,
  we do not write their definitions explicitly here,
  cf.~\multiautoref{defn:twoValues,def:ora.dist}.

  \medskip\noindent\textbf{Removing $\pmb{\OP,\OR}$:} We now remove access to $\OP,\OR$.
  We then have
  \begin{gather}
    \label{eq:pap1}
    P_A \leq 2 (q_P+q_R+1)\sqrt{P_1}, \\
    P_1 := \Pr[w=w_0:w_0\otR\bits\ellrand,
    (\Sy\com)\ot\$,
    \OS\ot\$,
    w\ot A_1^{\OE,\OS,\OF,\Opsi,\OV}]
    \notag
  \end{gather}
  for some $A_1$ by
  \autoref{lemma:remove.wora} (with $\calO_1:=(\OP,\OR)$,
  $w:=w_0$, $\calO_2:=(\OE,\OS,\OF,\Opsi,\OV)$, $\forall w':f(\cdot,w'):=f(w',\cdot,\cdot,\cdot):=w'$).
  Here the algorithm $A_1$
  makes at most as many oracle queries as $A$
  to the remaining oracles. Note that we also removed $\OP\ot\$$
  because $\OP$ is not used any more.

  \medskip\noindent\textbf{Removing $\pmb{\OE}$:}   We now transform $A_1$ not to output $w$, but to output the two
  accepting conversations $(\com,\ch,\resp,\ch',\resp')$ needed for
  extraction. In the following, we write short $\mathsf{Collision}$ for
  $(\ch,\resp)\neq(\ch',\resp')\land(\ch,\resp),(\ch',\resp')\in\Sy\com$.
  \begin{gather}
    \label{eq:p1p2}
    P_1 \leq 
    2q_E\sqrt{P_2}+2^{-\ellrand},\\
    P_2 := 
    \Pr[\mathsf{Collision}:
    (\Sy\com)\ot\$,
    \OS\ot\$,
    (\com,\ch,\resp,\ch',\resp')\ot A_2^{\OS,\OF,\Opsi,\OV}]\notag
  \end{gather}
  for some $A_2$ by \autoref{lemma:remove.xora} (with $w:=w_0$, $\ell:=\ellrand$,
  $\calO_1:=\OE$,
  $\calO_2:=(\OS,\OF,\Opsi,\OV)$,
  and $X:=\{(\com,\ch,\resp,\ch',\resp'):\mathsf{Collision}\}$).
  Here $A_2$ makes at most as many oracles queries as $A_1$.
  We also removed the choice of $w_0$ from the formula because none of the remaining
  oracles depend on it.

  \medskip\noindent\textbf{Removing $\pmb{\Opsi}$:} Fix integers
  $n,m$. We determinate the actual values later.
  By \autonameref{theo:emulate.opsi}, we have:
  \begin{gather}
    \label{eq:p2p3}
    P_2 \leq 
    P_3 + O\Bigl(\frac{q_\Psi}{\sqrt n}+\frac{q_\Psi}{\sqrt m}\Bigr), \\
    P_3 := 
    \Pr[\mathsf{Collision}:
    (\Sy\com)\ot\$,
    \OS\ot\$,
    (\com,\ch,\resp,\ch',\resp')\ot A_3^{\OS,\OF,\OV}(\ket R)]\notag
  \end{gather}
  for some $A_3$. Here $A_3$ makes $q_S,q_F,q_V$ queries to $\OS,\OF,\OV$. And
  $\ket R:=\ketpsi^{\otimes
    m}\otimes\ket{\alpha_1}\otimes\dots\otimes \ket{\alpha_n}$ with 
  $\ket{\alpha_j}:=(\cos\frac{j\pi}{2n})\ketpsi+(\sin\frac{j\pi}{2n})\ketbot$.

  \medskip\noindent\textbf{Removing $\pmb{\OS}$:}
  For given choice of $(\Sy \com)_{\com\in \bits\ellcom}$, let $\calD_Y$ be the distribution of
  $\OS(z)$, i.e., $\calD$ picks $\com\otR \bits\ellcom$ and $(\ch,\resp)\otR\Sy\com$ and returns $(\com,\ch,\resp)$.

  Fix some integers $s$ (we determine the value of $s$ later).
  Then, for fixed choice of $(\Sy\com)_\com$ ($\OV,\OF$ are
  deterministic given $\Sy\com$ anyway), we have by
  \autonameref{theo:small.range} (with $H:=\OS$):
  \begin{align*}
    \Babs{
      &\Pr[\mathsf{Collision}:
      \OS\ot\$,
      (\com,\ch,\resp,\ch',\resp')\ot A_3^{\OS,\OF,\OV}(\ket R)]
      -{}\\
      &\Pr[\mathsf{Collision}:
      G\ot\$,
      (\com,\ch,\resp,\ch',\resp')\ot A_3^{G,\OF,\OV}(\ket R)]
    } \leq 14q^3/s.
  \end{align*}
  Here $G\ot\$$ means
  that $G$
  is chosen as: pick $(\com_1,\ch_1,\resp_1),\dots,(\com_s,\ch_s,\resp_s)\ot\calD_Y$,
  then for all $z$,
  pick $i_z\otR\{1,\dots,s\}$ and set $G(z):=(\com_{i_z},\ch_{i_z},\resp_{i_z})$.

  By averaging over the choice of $(\Sy \com)$, we then get that
  \begin{gather}
    \label{eq:p3p4}
    \abs{P_3-P_4}\leq 14q_S^3/s,\\
    P_4:=\Pr[\mathsf{Collision}:      
    (\Sy\com)\ot\$,
    G\ot\$,
    (\com,\ch,\resp,\ch',\resp')
    \ot A_3^{G,\OF,\OV}].\notag
  \end{gather}
  We construct the adversary $A_4$: Let
  $A_4^{\OF,\OV}(\com_1,\ch_1,\resp_1,\dots,\com_s,\ch_s,\resp_s,\ket R)$ pick $G$ himself as:
  for all $z$, $i_z\otR\{1,\dots,s\}$, $G(z):=(\com_{i_z},\ch_{i_z},\resp_{i_z})$. Then $A_4$
  executes $A_3^{G,\OF,\OV}(\ket R)$.  Then
  \begin{align*}
    P_4=\Pr[&\mathsf{Collision}:
    (\Sy\com)\ot\$,
    (\com_1,\ch_1,\resp_1),\dots,(\com_s,\ch_s,\resp_s)\ot\calD_Y,\\
    &(\com,\ch,\resp,\ch',\resp')\ot A_4^{\OF,\OV}(\com_1,\ch_1,\resp_1,\dots,\com_s,\ch_s,\resp_s,\ket R)].
  \end{align*}
  (Note that the distribution
  $\calD_Y$ depends on the choice of $\Sy y$.)

  Let $A_5^{\OF,\OV}(\ketpsi^{\otimes s},\ket R)$ be the algorithm that does the following: For
  each each $i$, it takes one copy of the state $\ketpsi$ (given as input) and
  measures it in the computational basis to get $(\com_i,\ch_i,\resp_i)$. Then
  $A_5$ runs $A_4^{\OF,\OV}(\com_1,\ch_1,\resp_1,\dots,\com_s,\ch_s,\resp_s,\ket R)$.

  By definition of $\ketpsi$ (\autoref{defn:twoValues}), each $(\com_i,\ch_i,\resp_i)$ chosen by $A_5$ is
  independently distributed according to $\calD_Y$. Thus 
  \begin{gather}\label{eq:p4p5}
    P_4=P_5,\\
    P_5:=\Pr[\mathsf{Collision}:
    (\Sy\com)\ot\$,
    (\com,\ch,\resp,\ch',\resp')\ot A_5^{\OF,\OV}(\ketpsi^{\otimes s},\ket R)].
    \notag
  \end{gather}

  \medskip\noindent\textbf{Converting the $\pmb{\ket{\alpha_i}}$:}  The
  adversary $A_5$ is almost an adversary as in
  \autonameref{theo:pickone.sound}, with one exception: the input to
  $A_5$ is a state
  $\ket R=\ketpsi^{\otimes m}\otimes\ket{\alpha_1}\otimes\dots\otimes
  \ket{\alpha_n}$
  with
  $\ket{\alpha_j}:=(\cos\frac{j\pi}{2n})\ketpsi+(\sin\frac{j\pi}{2n})\ketbot$.
  \autoref{theo:pickone.sound} on the other hand assumes an
  adversary that takes as input states in the span of $\ketpsi$ and
  $\ket{\Sigma\Phi}:=\sum_{\com,\ch,\resp}2^{-(\ellcom+\ellch+\ellresp)/2}\ket{\com,\ch,\resp}$.
  Let
  $\ket{\tilde\alpha_j}:=(\cos\frac{j\pi}{2n})\ketpsi+(\sin\frac{j\pi}{2n})\ket{\Sigma\Phi}$.
  $\ket{\Tilde R}=\ketpsi^{\otimes
    m}\otimes\ket{\tilde\alpha_1}\otimes\dots\otimes
  \ket{\tilde\alpha_n}$
  Let $U_\alpha\ket{\Sigma\Phi}:=\ketbot$ and
  $U_\alpha\ketbot:=\ket{\Sigma\Phi}$ and
  $U_\alpha\ket{\Phi}:=\ket{\Phi}$ for $\ket\Phi$ orthogonal to
  $\ketbot,\ket{\Sigma\Phi}$.

  Let $A_6^{\OF,\OV}(\ketpsi^{\otimes s},\ket{\Tilde R})$ be the
  algorithm that runs
  $A_5^{\OF,\OV}(\ketpsi^{\otimes s},(I^{\otimes m}\otimes
  U_\alpha^{\otimes n})\ket{\Tilde R})$. Then
  \begin{gather}
    \label{eq:p5p6.pre}
    P_5\leq P_6+
    \TD\bigl((I^{\otimes m}\otimes
    U_\alpha^{\otimes n})\ket{\Tilde R},\ \ket R\bigl),\\
    P_6:=\Pr[\mathsf{Collision}:
    (\Sy\com)\ot\$,
    (\com,\ch,\resp,\ch',\resp')\ot A_6^{\OF,\OV}(\ketpsi^{\otimes s},\ket{\Tilde R})].
    \notag
  \end{gather}
  Write $\ketpsi$ as $\ketpsi=\gamma\ket{\Sigma\Phi}+\delta\ket{\Sigma\Phi^\bot}$ with
  $\ket{\Sigma\Phi^\bot}$ a state orthogonal to
  $\ket{\Sigma\Phi}$. Write short $c:=(\cos\frac{j\pi}{2n})$ and $s:=(\sin\frac{j\pi}{2n})$.
  Then 
  \begin{align*}
    \chi &:=\bra{\alpha_j}U_\alpha\ket{\tilde\alpha_j} \\
    &=
    (c\ketpsi+s\ketbot)^\dagger U_\alpha
    (c\ketpsi+s\ket{\Sigma\Phi}) \\
    &=c^2\bra{\Sigma\Psi} U_\alpha\ketpsi + s^2\inner\bot\bot
    + cs\inner{\Sigma\Psi}{\bot} + cs\bra\bot U_\alpha\ketpsi \\
    &\starrel=c^2\abs{\delta^2} + s^2 + cs\cdot 0 + cs\gamma
    = c^2(1-\abs{\gamma^2}) + s^2 + cs\gamma
    = 1-c^2\abs{\gamma^2}+cs\gamma.
  \end{align*}
  In $(*)$ we use that
  $\ketbot,\ket{\Sigma\Phi},\ket{\Sigma\Phi^\bot}$ are orthogonal.
  Furthermore, 
  \begin{align*}
    \gamma=\inner{\Sigma\Phi}{\Sigma\Psi}
    &= \!\!\! \sum_{\substack{\com,\ch,\resp\\(\ch,\resp)\in\Sy\com}}  \!\!\!\!
    2^{-(\ellcom+\ellch+\ellresp)/2}\cdot 2^{-\ellcom/2}/\sqrt\kk \\
    &= 2^{\ellcom}\kk \cdot  2^{-(\ellcom+\ellch+\ellresp)/2}\cdot 2^{-\ellcom/2}/\sqrt\kk
    = {2^{-(\ellch+\ellresp)/2}}\sqrt{\kk} \geq 0.
  \end{align*}
  Thus
  \begin{align*}
    \chi = 1-c^2\abs{\gamma^2}+cs\gamma
    \geq 1-c^2\gamma^2 \geq 1-\gamma^2
  \end{align*}
  and hence
  \begin{align*}
    \TD(\ket{\alpha_j},U_\alpha\ket{\tilde\alpha_j})
    =
    \sqrt{1-\chi^2}
    \leq \sqrt{1-(1-\gamma^2)}
    \leq \sqrt{2\gamma^2}
    = {2^{-(\ellch+\ellresp-1)/2}}\sqrt{\kk}.
  \end{align*}
  With \eqref{eq:p5p6.pre}, we get
  \begin{align}
    P_5&\leq P_6+
    \TD\bigl((I^{\otimes m}\otimes
    U_\alpha^{\otimes n})\ket{\Tilde R},\ \ket R\bigr)    
    = P_6+\sum_{i=1}^n\TD(\ket{\alpha_j},U_\alpha\ket{\tilde\alpha_j})\notag\\
    &\leq P_6 + n{2^{-(\ellch+\ellresp-1)/2}}\sqrt{\kk}.
    \label{eq:p5p6}
  \end{align}
  
  \medskip\noindent\textbf{Wrapping up:} Note that $A_6$ is an
  adversary as in \autonameref{theo:pickone.sound}. Thus by
  \autoref{theo:pickone.sound} (with $h:=n+m+s$), we have:
  \begin{align}
    P_6 \leq 
O\left(
\frac{(n+m+s)}{2^{\ellcom/2}} + \frac{(q_V+q_F)^{1/2}\kk^{1/4}}{2^{(\ellch+\ellresp)/4}}
+ \frac{(q_V+q_F)^{1/2}}{\kk^{1/4}}
\right)
.    \label{eq:p6}
  \end{align}
  Let
    $n,m,s:=\floor{\min\{2^{\ellresp/4},2^{\ellcom/3}\}}$. Since
    $\ellresp$ and $\ellcom$ are superlogarithmic, $n,m,s$ are
    superpolynomial.  The first summand in \eqref{eq:p6} is negligible
    since $n+m+s\leq 3\cdot 2^{\ellcom/3}$. The second summand is
    negligible because $q_V,q_F$ are polynomially-bounded and
    $\kk=2^{\ellch+\floor{\ellresp/3}}$ and $\ellresp$ is
    superlogarithmic. The third summand is negligible because
    $q_V,q_F$ are polynomially-bounded and $\kk$ is superlogarithmic.
    Thus by \eqref{eq:p6}, $P_6$ is negligible.

    Using $n\leq 2^{\ellresp/4}$ and $k\leq 2^{\ellch+\ellresp/3}$, we
    get that the second summand in \eqref{eq:p5p6} is upper bounded by
    $2^{\ellresp/4}\cdot 2^{-\ellch/2-\ellresp/2-1/2}\cdot
    2^{\ellch/2+\ellresp/6}=2^{-1/2-\ellresp/12}$
    which is negligible. Since $P_6$ is negligible, \eqref{eq:p5p6}
    implies that $P_5$ is negligible. By \eqref{eq:p4p5}, $P_4$ is
    negligible. Since $q_S$ is polynomially-bounded and $s$ is
    superpolynomial, $14q_S^3/s$ is negligible. Thus by
    \eqref{eq:p3p4}, $P_3$ is negligible. Since $q_\Psi$ is
    polynomially-bounded and $n,m$ are superpolynomial, the second
    summand in \eqref{eq:p2p3} is negligible, so $P_2$ is negligible.
    Since $\ellrand$ is superlogarithmic, $q_E$ is
    polynomially-bounded, and $P_2$ is negligible, \eqref{eq:p1p2}
    implies that $P_1$ is negligible. And since $q_P,q_R$ are
    polynomially-bounded and $P_1$ is negligible, \eqref{eq:pap1}
    implies that $P_A$ is negligible.
    This shows part
    \eqref{item:pick.one.w} of the lemma.

  \bigskip

  We now show part \eqref{item:pick.one.xx} of the lemma. For an
  adversary $A$ outputting $(\com,\ch,\resp,\ch',\resp')$, let $B$ be
  the adversary that runs $(\com,\ch,\resp,\ch',\resp')\ot A$, then
  invokes $w\ot\OE(\com,\ch,\resp,\ch',\resp')$ and returns $w$. Note
  that $B$ makes $q_E+1$ queries to $\OE$, and the same number of
  queries to the other oracles as $A$.  By definition of $\OE$, we
  have
  \begin{align*}
    \Pr[&(\ch,\resp)\neq (\ch',\resp')\land
    (\ch,\resp),(\ch',\resp')\in\Sy \com:\\
    &(\com,\ch,\resp,\ch',\resp')\ot A^{\Oall}] \leq
    \Pr[w=w_0: w\ot B^{\Oall}].
  \end{align*}
  By \eqref{item:pick.one.w} the rhs is negligible, thus the lhs is,
  too. This proves  \eqref{item:pick.one.xx}.
\end{proof}

\section{Proofs for \autoref{sec:com}}

\subsection{Proof for \autoref{lemma:com.props}}
\label{sec:proof:lemma:com.props}

\begin{proof}[of \autoref{lemma:com.props}]
  \textbf{Perfect completeness:} By definition of $\OS$, we have that
  $x_i\in\Sy{y_i}$ for all $(y_i,x_i):=\OS(z_i)$. Hence
  $\OV(y_i,x_i)=1$ for all $i$. Thus $\COMverify(c,m,u)=1$ for
  $(c,u)\ot\COM(m)$. Hence we have perfect completeness.

  \medskip

  \noindent\textbf{Computational strict binding:} Consider an adversary $A^{\Oall}$
  against the computational strict binding property.  
  Let $\mu$ be the
  probability that $A^{\Oall}$ outputs $(c,m,u,m',u')$
  such that $(m,u)\neq(m',u')$ and $\ok=\ok'=1$ with
  $\ok=\COMverify(c,m,u)$ and $\ok'=\COMverify(c,m',u')$. We need to
  show that $\mu$ is negligible. Let
  $c=:(p_1,\dots,p_{\abs m},y_1,\dots,y_\abs m,b_1,\dots,b_{\abs m})$ and
  $u=:(x_1,\dots,x_{\abs m})$ and $u'=:(x'_1,\dots,x'_{\abs m})$. Then
  $(m,u)\neq(m',u')$ implies that for some $i$,
  $(x_i,m_i)\neq(x_i',m_i')$. If $x_i=x_i'$, then from $\ok=\ok'=1$ we
  have $m_i=b_i\oplus\ibit{p_i}(x_i)=b_i\oplus\ibit{p_i}(x_i')=m_i'$, in
  contradiction to $(x_i,m_i)\neq(x_i',m_i')$. So $x_i\neq x_i'$. 
  Furthermore, $\ok=\ok'=1$ implies that
  $\OV(y_i,x_i)=\OV(y_i,x_i')=1$, i.e., $x_i,x_i'\in\Sy{y_i}$.
  So $A^{\Oall}$ finds $x_i\neq x_i'$ with $x_i,x_i'\in\Sy{y_i}$ with
  probability $\mu$.
  By \autonameref{coro:pick.one.sound-oall}, this implies that $\mu$ is
  negligible. 

  \medskip

  \noindent\textbf{Computational binding:} This is implied by computational strict
  binding.

\medskip

  \noindent
  \textbf{Statistical hiding:} Fix $m,m'\in\bit$. Let $(y,x):=\OS(z)$,
  $z\otR\bits\ellrand$, $p\ot\{1,\dots,\ellch+\ellresp\}$,
  $b:=m\oplus\ibit p(x)$.  Let $\hat y\otR\ellcom$,
  $\hat x\otR\Sy{\hat y}$.  Define analogously
  $y',x',z',p',b',\hat y',\hat x'$.

Let $\calD$ be the
distribution that returns $(\hat y,\hat x)$ with
$\hat y\otR\bits\ellcom$, $\hat x\otR\Sy{\hat y}$.  Note that by
definition of $\OS$, $\OS(z)$ is initialized according to $\calD$.
By \autoref{lemma:empirical}, for
fixed choice of the sets $\Sy y$,
$\SD\bigl((\OS,y,x);(\OS,\hat y,\hat x)\bigr)\leq 2^{(\ellcom-\ellrand)/2-1}\sqrt\kk=:\mu_1$. (With
$X:=\bits\ellrand$, $Y:=\{(y,x):y\in\bits\ellcom,x\in\Sy y\}$, and
$\calO:=\OS$.)
Thus for random $\Sy y$ and random $p$, $\SD\bigl((\Oall,p,y,\ibit p(x)\oplus
m);(\Oall,p,\hat y,\ibit p(\hat
x)\oplus m)\bigr)\leq\mu_1$.
Let $b^*\otR\bit$.
For fixed $\hat y$ and $p$ and random sets $\Sy y$ and random $p$, $\SD\bigl((\Sy{\hat
  y},\hat y,\ibit p(\hat x));(\Sy{\hat y},\hat
y,b^*)\bigr)\leq1/2\sqrt k=:\mu_2$ by
\autoref{lemma:s.lsb}.
Thus for random $\hat y$ and $p$, $\SD\bigl((\Oall,p,\hat y,\ibit p(\hat x)\oplus
m);(\Oall,p,\hat y,b^*\oplus m)\bigr)\leq\mu_2$.
And $(\Oall,p,\hat y,b^*\oplus m)$ has the same distribution as 
$(\Oall,p,\hat y,b^*)$ since $b^*\in\bit$ is uniform and independently
chosen from $\Oall,\hat y$.
Hence $\SD\bigl((\Oall,p,y,\ibit p(x)\oplus m);(\Oall,p,\hat y,b^*)\bigr)\leq\mu_1+\mu_2$.
Analogously,  $\SD\bigl((\Oall,p',y',\ibit p(x')\oplus m');(\Oall,p',\hat
y',b^{*\prime})\bigr)\leq\mu_1+\mu_2$ with $b^{*\prime}\otR\bit$.
Since $(\Oall,p,\hat y,b^*)$ and $(\Oall,p',\hat y',b^{*\prime})$ have the same
distribution, this implies 
\begin{equation}\label{eq:O.lsb}
\SD\bigl((\Oall,p,y,\ibit p(x)\oplus m);(\Oall,p',y',\ibit p(x')\oplus m')\bigr)\leq2(\mu_1+\mu_2).
\end{equation}

Fix $m_1,m_2$ with
$\abs{m_1}=\abs{m_2}$. Let $z_i\otR\bits\ellrand$,  $(y_i,x_i):=\OS(z_i)$,
$p_i\otR\{1,\dots,\ellch+\ellresp\}$, $b_i:=m_i\oplus\ibit{p_i}(x_i)$ and analogously
$y_i',x_i',p_i',z_i',b_i'$. By induction over $n$, and using
\eqref{eq:O.lsb}, we get for all $1\leq n\leq \abs{m_1}$:
\begin{align*}
\SD\bigl(&(\Oall,(p_i)_{i=1,\dots,n},(y_i)_{i=1,\dots,n},(\ibit{p_i}(x_i)\oplus
m_i)_{i=1,\dots,n});\\
&(\Oall,(p_i')_{i=1,\dots,n},(y'_i)_{i=1,\dots,n},(\ibit{p_i}(x'_i)\oplus m'_i)_{i=1,\dots,n})\bigr)\leq2n(\mu_1+\mu_2).
\end{align*}
For $n=\abs{m_1}$, this becomes
\[
\SD\bigl((\Oall,c),(\Oall,c')\bigr)\leq 2{\abs{m_1}}(\mu_1+\mu_2)=:\mu\qquad
(\text{with }c\ot\COM(m),\ c'\ot\COM(m')).
\]
Since $\abs{m_1}$ is polynomially-bounded, and $\ellrand-\ellcom-\kk$ is
superlogarithmic, and $\kk$ is superpolynomial, $\mu$ is
negligible.
Thus $\COM$ is statistically hiding.
\end{proof}

\subsection{Proof of \autoref{lemma:com.attack}}
\label{app:proof:lemma:com.attack}

\begin{proof}[of \autoref{lemma:com.attack}]
  Our adversary is as follows: 
  \begin{compactitem}
  \item $B_1(\abs{m})$ invokes $E_1$ from~\autonameref{theo:pick1.complete}
    $\abs{m}$ times to get $(y_i,\ketpsiy{y_i})$ for $i=1,\dots,\abs m$.\footnote{$E_1$ expects an input $\ketpsi$.
      $\ketpsi$ can be computed using the oracle $\Opsi$.}
    Let $p_1,\dots,p_{\abs m}\otR\{1,\dots,\ellch+\ellresp\}$.
    Let $b_1,\dots,b_{\abs m}\otR\bit$. Output $c:=(p_1,\dots,p_{\abs m},y_1,\dots,y_\abs m,b_1,\dots,b_\abs m)$.
  \item $B_2(m)$: Let $P_i(x):=1$ iff $\ibit{p_i}(x)= b_i\oplus m_i$. Then,
    for each $i=1,\dots,\abs m$, $B_2$ invokes
    $E_2(n,\delta_{\min},y_i,\ketpsiy{y_i})$
    from~\autoref{theo:pick1.complete} with oracle access to $P:=P_i$
    and with $n:=\ellcom$ and $\delta_{\min}:=1/3$ to get
    $x_i$. Then $B_2$ outputs $u:=(x_1,\dots,x_n)$.
  \end{compactitem}
  By \autoref{theo:pick1.complete}, the probability that the $i$-th
  invocation of $E_2$ fails to return $x_i$ with $x_i\in\Sy y\land
  P_i(x_i)=1$ is at most:
  \begin{align*}
    f&:=2^{-\ellcom} + f_\delta
    \qquad\text{with}\qquad f_\delta:=\Pr\Bigl[\frac{\abs{\{x\in\Sy
        {y_i}:P_i(x)=1\}}}{\abs{\Sy{y_i}}}<\delta_{\min}\Bigr]
  \end{align*}

  Let $P'_0:=\{x:\ibit{p_i}(x)=0\}$ and $P'_1:=\{x:\ibit{p_i}(x)=1\}$.  Since
  $\Sy{y_i}\subseteq X$ is chosen uniformly at random, by
  \autoref{lemma:p.fraction} we have for $b=0,1$:
  \[
  f_\delta^b:=\Pr\bigl[\abs{\Sy{y_i}\cap P_b'}/\abs{\Sy{y_i}}<\delta_{\min}\bigr]
  \leq e^{-2k(\frac12-\delta_{\min})^2} = e^{-k/18}.
  \]
  Since $P_i=P_0'$ or $P_i=P_1'$, we have $f_\delta\leq f_\delta^0+f_\delta^1 \leq 2e^{-k/18}$.
  (Note: we cannot just apply \autoref{lemma:p.fraction} to $P_i$ because $P_i$ might not be independent of $\Sy{y_i}$.)
  
  The probability that $B_2$ fails to return $u$ with
  $\COMverify(c,m,u)$ is then $\abs mf$.  Hence
  $\varepsilon_{\COM}\geq 1-\abs mf \geq 1-
  \abs m2^{-\ellcom} + \abs m2e^{-k/18}$
  which is overwhelming since $\abs m$ is polynomial and $\ellcom$ and $k$
  are superlogarithmic.
\end{proof}

\section{Proofs for \autoref{sec:attack.sigma}}

\subsection{Proof of \autoref{lemma:sigma.sec}}
\label{app:proof:lemma:sigma.sec}

\begin{proof}[of \autoref{lemma:sigma.sec}]
  \textbf{Completeness}: We need to show that with overwhelming probability, \begin{inparaenum}[(a)]\item\label{com}
       $\COMverify(c_\ch,\resp_\ch,u_\ch)=1$ for
    $(c_\ch,u_\ch)\ot\COM(\resp_\ch)$ and \item \label{ov}
 $\OV(\com,\ch,\resp_\ch)=1$
    for uniform $\com,\ch$ and $\resp_\ch:=\OP(w,\com,\ch)$.
  \end{inparaenum} From the completeness of $\COM$
  (\autoref{lemma:com.props}), we immediately get \eqref{com}.
  We prove \eqref{ov}:
  By definition
  of $\OP$ and $\OV$, \eqref{ov} holds iff
  $\exists\resp.(\ch,\resp)\in\Sy\com$. We thus need to show that
  $p_1:=\Pr[\exists\resp.(\ch,\resp)\in\Sy\com]$ is overwhelming.
  $\Sy\com$ is a uniformly random subset of size
  $\kk=2^{\ellch+\floor{\ellresp}/3}$ of
  $X=\bits\ellch\times\bits\ellresp$. Thus $p_1$ is lower bounded by
  the probability $p_2$ that out of $\kk$ uniform independent samples
  from $\bits\ellch$, at least one is $\ch$. Thus
  $p_1\geq p_2=1-(1-2^{-\ellch})^\kk=1-\bigl((1-1/2^{\ellch})^{2^{\ellch}}\bigr)^{2^{\floor\ellresp/3}}\starrel\geq
  1-e^{-2^{\floor\ellresp/3}}$
  where $(*)$ uses the fact that $(1-1/n)^n$ converges from below to
  $1/e$ for integers $n\to\infty$.  Thus $p_1$ is overwhelming
  for superlogarithmic $\ellresp$, and the sigma-protocol is complete.

  \medskip

  \noindent\textbf{Commitment entropy:} We need to show that
  $\com^*\ot P_1(s,w)$ has superlogarithmic min-entropy. Since
  $\com^*=(\com,\dots)$, and $\com$ is uniformly distributed on
  $\bits\ellcom$, the min-entropy of $\com^*$ is at least $\ellcom$
  which is superlogarithmic.

  \medskip

  \noindent \textbf{Perfect special soundness}: Observe that
  $V(s,\com^*,\ch,\resp^*)=V(s,\com^*,\ch',\resp^{*\prime})=1$ and $\ch\neq\ch'$ implies
  $(\ch,\resp),(\ch',\resp')\in\Sy\com$ and $s=s_0$ and $\ch\neq\ch'$ which in turn implies
  $\OE(\com,\ch,\resp,\ch',\resp')=w_0$ and $(s,w_0)\in R$. Thus an extractor $E$ that
  just outputs   $\OE(\com,\ch,\resp,\ch',\resp')$ achieves perfect
  special soundness.

  \medskip

  \noindent\textbf{Computational strict soundness:} We need to show that a
  polynomial-time $A$ will only with negligible probability output
  $(\com^*,\ch,\resp^*,\resp^{*\prime})$ such that
  $\resp^*\neq\resp^{*\prime}$ and
  $V(s,\com^*,\ch,\resp^*)=V(s,\com^*,\ch,\resp^{*\prime})=1$. Assume
  $A$ outputs such a tuple with non-negligible probability. By
  definition of $V$, this implies that $\resp^*=(\resp,u)$,
  $\resp^{*\prime}=(\resp',u')$, and $\com^*$ contains $c_\ch$ such
  that $\COMverify(c_\ch,\resp,u)=1$ and
  $\COMverify(c_\ch,\resp',u')=1$. Since $\resp^*\neq\resp^{*\prime}$,
  this contradicts the computational strict binding property of
  $\COM,\COMverify$ (\autoref{lemma:com.props}). Thus the
  sigma-protocol has computational strict soundness.

  \medskip

  \noindent\textbf{Statistical HVZK}: Let $S$ be the
  simulator that picks $z\otR\bits\ellrand$, computes
  $(\com,\ch,\resp):=\OS(z)$, and $(c_c,u_c)\ot\COM(0^{\ellresp})$ for
  all $c\in\bits\ellch\setminus\{\ch\}$, and
  $(c_\ch,u_\ch)\ot\COM(\resp)$, and returns $(\com^*,\ch,\resp^*)$
  with $\com^*:=(\com,(c_\ch)_{\ch\in\bits\ellch})$ and
  $\resp^*:=(\resp_\ch,u_\ch)$.  We now compute the difference between
  the probabilities from the definition of statistical HVZK
  (\autoref{def:sigma.props}) for $(s,w)\in R$, i.e., for $s=s_0$ and
  $w=w_0$. In the calculation, $\com^*$ always stands short for
  $(\com,(c_\ch)_{\ch\in\bits\ellch})$ and $\resp^*$ for
  $(\resp_\ch,u_\ch)$.
  \begin{align*}
    &\Pr[b=1: \com^*\ot P_1(s,w), \ch\otR\bits\ellch, \resp^*\ot P_2(\ch),b\ot A(\com^*,\ch,\resp^*)]\\
    &=\Pr\bigl[b=1: \com\otR\bits\ellcom,\
    \ch\otR\bits\ellch,\
    [\text{for all $c\in\bits\ellch$}\colon
    z_c\otR\bits\ellrand,\\
    &\qquad\qquad    \resp_c:=\OP(w,\com,c,z_c),\
    (c_c,u_c)\ot\COM(\resp_c)], \
    b\ot A(\com^*,\ch,\resp^*)\bigr]\\
    &\stackrel{\varepsilon_0}\approx
    \Pr\bigl[b=1: \com\otR\bits\ellcom,\
    \ch\otR\bits\ellch,\
    [\text{for all $c\in\bits\ellch\setminus\{\ch\}$}\colon\\
    &\qquad\qquad (c_c,u_c)\ot\COM(0^{\ellresp})], \
    z_\ch\otR\bits\ellrand,
    \resp_\ch:=\OP(w,\com,\ch,z_\ch),\\
    &\qquad\qquad(c_\ch,u_\ch)\ot\COM(\resp_\ch),
    b\ot A(\com^*,\ch,\resp^*)\bigr]
  \end{align*}
  Here $a\stackrel{\varepsilon_0}\approx b$ means that
  $\abs{a-b}\leq\varepsilon_0$ where
  $\varepsilon_0:=2^{\ellch}\varepsilon_{\COM}$ and
  $\varepsilon_{\COM}$ is the statistical distance between commitments
  $\COM(\resp_c)$ and $\COM(0^{\ellresp})$. We have that
  $\varepsilon_{\COM}$ is negligible by \autoref{lemma:com.props}
  (statistical hiding of $\COM$).

  We abbreviate $[\text{for all $c\in\bits\ellch\setminus\{\ch\}$}\colon
     (c_c,u_c)\ot\COM(0^{\ellresp})]$ with $[\COM(0)]$
    and continue our calculation:
    \begin{align*}
    \cdots &=
    \Pr\bigl[b=1: \com\otR\bits\ellcom,\
    \ch\otR\bits\ellch,\
    [\COM(0)], \
    z_\ch\otR\bits\ellrand,\\
    &\qquad\qquad\resp_\ch:=\OP(w,\com,\ch,z_\ch),\
    (c_\ch,u_\ch)\ot\COM(\resp_\ch),\
    b\ot A(\com^*,\ch,\resp^*)\bigr]
    \\
    &\stackrel{\varepsilon_1}\approx
    \Pr\bigl[b=1: \com\otR\bits\ellcom,\
    \ch\otR\bits\ellch,\
    [\COM(0)], \\
    &\qquad\qquad\resp_\ch\ot\calD_{\com,\ch},\
    (c_\ch,u_\ch)\ot\COM(\resp_\ch),\
    b\ot A(\com^*,\ch,\resp^*)\bigr]
  \end{align*}
  Here $\calD_{\com,\ch}$ is the uniform distribution on
  $\{\resp:(\ch,\resp)\in\Sy\com\}$.  (Or, if that set is empty, 
  $\calD_{\com,\ch}$ assigns probability $1$ to $\bot$.)
    And
  $a\stackrel{\varepsilon_1}\approx b$ means that
  $\abs{a-b}\leq\varepsilon_1$ where
  $\varepsilon_1:=\frac12\sqrt{2^{\ellresp} / 2^{\ellrand}}$. The last
  equation follows from \autoref{lemma:empirical}, with
  $X:=\bits\ellrand$ and $Y:=\bits\ellresp$ and
  $\calD:=\calD_{\ch,\com}$, and using the fact that for all $z$,
  $\OP(w_0,\com,\ch,z)$ is chosen according to
  $\calD_{\ch,\com}$. (Note that the adversary $A$ has access to
  $\OP$, but that is covered since $\calO$ occur on both sides of the
  statistical distance in \autoref{lemma:empirical}.)
  We continue the computation:
  \begin{align*}
    \dots \stackrel{\varepsilon_2}\approx
    \Pr[b=1: {}& (\com,\ch,\resp_\ch)\otR\calD',\,
    [\COM(0)],\
    (c_\ch,u_\ch)\ot\COM(\resp_\ch),\\
    &b\ot A(\com^*,\ch,\resp^*)]
  \end{align*}
  Here $\calD'$ is the distribution resulting from choosing
  $\com\otR\bits\ellcom$, $(\ch,\resp)\otR\Sy\com$. By
  \autoref{lemma:pick.dep}, $\varepsilon_2\leq
  \frac{2\kk^2}{2^{\ellch+\ellresp}}
  +
  \frac{2^{\ellch/2}}{2\sqrt \kk}
  $. 
  We continue
  \begin{align*}
    \dots &\stackrel{\varepsilon_3}\approx
    \Pr[b=1: z\otR\bits\ellrand, (\com,\ch,\resp):=\OS(z),\\
    &\qquad\qquad[\COM(0)],\
    (c_\ch,u_\ch)\ot\COM(\resp_\ch),\
    b\ot A(\com^*,\ch,\resp^*)]
  \end{align*}
  Here
  $\varepsilon_3=\sqrt{(2^{\ellcom}\cdot k)/2^{\ellrand}}$.
  This follows from \autoref{lemma:empirical} with $\calD:=\calD'$ and
  $X:=\bits\ellrand$ and
  $Y:=\{(\com,\ch,\resp):(\ch,\resp)\in\Sy\com\}$. (Note that
  $\abs Y=2^{\ellcom}\cdot k$.) We continue
  \[
    \dots =
    \Pr[b=1: (\com^*,\ch,\resp^*):=S(s),b\ot A(\com^*,\ch,\resp^*)].
  \]
  Thus the difference of probabilities from the definition of
  statistical HVZK is bounded by
  $\varepsilon:=\varepsilon_0+\varepsilon_1+\varepsilon_2+\varepsilon_3$. And
  $\varepsilon$ is negligible since $\varepsilon_{\COM}$ is
  negligible, and
  $\kk=2^{\ellch+\floor{\ellresp}/3}$, and $\ellch$ is logarithmic, and
  $\ellresp,\ellcom$ are superlogarithmic, and $\ellrand=\ellcom+\ellresp$.
\end{proof}

\subsection{Proof of \autoref{lemma:sigma-break}}
\label{app:proof:lemma:sigma-break}

\begin{proof}[of \autoref{lemma:sigma-break}]
  According to~\autoref{def:total.break} (specialized to the case of
  the sigma-protocol from~\autoref{def:sigma}) we need to construct a
  polynomial-time quantum adversary $A_1,A_2,A_3$ such that:
  \begin{compactitem}
  \item Adversary success:
    \begin{align}
      P_A:=
      \Pr[&\ok=1:s\ot A_1, \com^*\ot A_2, \ch\otR\bits\ellch, 
      \notag\\&
      \resp^*\ot A_3(\ch), \ok=V(s,\com^*,\ch,\resp^*)] \notag\\
      = \Pr[&\ok_v=1\land\ok_c=1\land s=s_0:s\ot
      A_1,\bigr(\com,(c_\ch)_{\ch\in\bits\ellch}\bigl)\ot A_2, \notag\\&\ch\otR\bits\ellch,(\resp,u)\ot
      A_3(\ch),\ok_v:=\OV(\com,\ch,\resp),\notag\\&\ok_c=\COMverify(c_\ch,\resp,u)]
      \label{eq:advsucc.sigma}
    \end{align}
    is overwhelming.
  \item Extractor failure: For any polynomial-time quantum $E$ (with
    access to the final state of $A_1$),
    $\Pr[s=s_0,w=w_0:s\ot A_1,w\ot E(s)]$ is negligible.
  \end{compactitem}
  Our adversary is as follows: 
  \begin{compactitem}
  \item Let $B_1,B_2$ be the adversary from
    \autonameref{lemma:com.attack}. (That is, $B_1(\abs m)$ produces a
    fake commitment which $B_2(m)$ then opens to $m$.)    
  \item $A_1$ outputs $s_0$.
  \item $A_2$ invokes $E_1$ from~\autonameref{theo:pick1.complete} to get
     $(\com,\ketpsiy\com)$.\footnote{Using $\Opsi$ to get the input
       $\ketpsi$ for $E_1$.} 
     Then $A_2$ invokes $c_c\ot B_1({\ellresp})$ for all
     $c\in\bits\ellch$.
     $A_2$ outputs $\com^*:=(\com,(c_\ch)_{\ch\in\bits\ellch})$.
   \item Let $P_\ch(\ch',\resp'):=1$ iff $\ch'=\ch$.  $A_3(\ch)$
     invokes $E_2(n,\delta_{\min},\com,\ketpsiy\com)$ from~\autoref{theo:pick1.complete} with oracle
     access to $P:=P_\ch$ and with $n:=\ellcom$ and
     $\delta_{\min}:=2^{-\ellch-1}$ to get $\resp$. Then $A_3$ invokes
     $u\ot B_1(\resp)$ to get opening information for $c_\ch$.
     $A_3$ outputs $\resp^*:=(\resp,u)$.
  \end{compactitem}
  \noindent\textbf{Adversary success:} By  \autoref{lemma:com.attack}, $\COMverify(c_\ch,\resp,u)=1$ with
  overwhelming probability. Thus $\ok_c=1$ with overwhelming
  probability in \eqref{eq:advsucc.sigma}.

  By \autoref{theo:pick1.complete}, the probability that  $E_2$ fails to return $(\ch',\resp)$ with $(\ch',\resp)\in\Sy\com\land
  P_\ch(\ch',\resp)=1$ is at most:
  \begin{align*}
    f&:=2^{-\ellcom} + f_\delta
    \qquad\text{with}\qquad
    f_\delta:=\Pr\Bigl[\frac{\abs{\{(\ch',\resp)\in\Sy\com:
        P_\ch(\ch',\resp)=1\}}}{\abs{\Sy\com}}<\delta_{\min}\Bigr]
  \end{align*}
  Let $P':=\{x:P_\ch(x)=1\}$ and $X:=\bits\ellch\times\bits\ellcom$. 
  Then $\abs P'/\abs X=2^{-\ellch}$.
  Since
  $\Sy\com\subseteq X$ is chosen uniformly at random with $\abs{\Sy\com}=\kk$, by
  \autoref{lemma:p.fraction} we have:
  \[
  f_\delta=\Pr\bigl[\abs{\Sy\com\cap P'}/\abs{\Sy\com}<\delta_{\min}\bigr]
  \leq e^{-2k(2^{-\ellch}-\delta_{\min})^2} = e^{-k2^{-2\ellch-1}}.
  \]
  Thus $f\leq2^{-\ellcom} + e^{-k2^{-2\ellch-1}}$ is
  negligible since $\ellcom$ is superpolynomial, $\ellch$ logarithmic,
  and $k$ superpolynomial. Thus with overwhelming probability $E_2$
  returns $(\ch',\resp)\in\Sy\com$ with
  $P_\ch(\ch',\resp)=1$. $P_\ch(\ch',\resp)=1$ implies $\ch'=\ch$. Hence
  $(\ch,\resp)\in\Sy\com$, thus $\OV(\com,\ch,\resp)=1$, thus
  $\ok_v=1$ with overwhelming probability. Since $s=s_0$ by
  construction of $A_1$, it follows that $P_A$ is overwhelming. Thus
  we have adversary success.

  \medskip

  \noindent\textbf{Extractor failure:}   It remains to show extractor failure. Fix some polynomial-time
  $E$. 
  Since $A_1$ only returns a fixed $s_0$ and has a trivial final
  state, without loss of generality we can assume that $E$ does not
  use its input $s$ or $A_1$'s final state. Then 
  \begin{align*}
    P_E := \Pr[s=s_0,w=w_0:s\ot A_1,w\ot E^{\Oall}(s)]
    =
    \Pr[w=w_0:w\ot E^{\Oall}]
  \end{align*}
  is negligible by \autonameref{coro:pick.one.sound-oall}.
  This shows extractor failure.
\end{proof}

\subsection{Proof of \autoref{lemma:sigma.sec.comp}}
\label{app:proof:lemma:sigma.sec.comp}

\begin{proof}[of \autoref{lemma:sigma.sec.comp}]
  \textbf{Completeness} and \textbf{statistical HVZK} and
  \textbf{commitment entropy} hold trivially,
  because they only have to hold for $(s,w)\in R'=\varnothing$.
  \textbf{Computational strict soundness} is shown exactly as in the
  proof of \autonameref{lemma:sigma.sec}. (The definition of computational
  strict soundness is independent of the relation $R'$.)

  \medskip

  \noindent\textbf{Computational special soundness:} Let $\ESigma$ be an
  algorithm that always outputs $\bot$.
  By
  \autonameref{def:sigma.props} we have to show that the following probability is
  negligible:
  \begin{align*}
    P_S &:= \Pr[(s,w)\notin R'\land\ch\neq\ch'\land\ok=\ok'=1:
    (s,\com^*,\ch,\resp^*,\ch',\resp^{*\prime})\ot A^{\Oall},\\
    &\qquad\qquad \ok\ot V(s,\com^*,\ch,\resp^*),\ok'\ot
    V(s,\com^*,\ch',\resp^{*\prime}), \\
    &\qquad\qquad w\ot \ESigma(s,\com^*,\ch,\resp^*,\ch',\resp^{*\prime})]
    \\
    &\leq \Pr[\ch\neq\ch'\land(\ch,\resp),(\ch',\resp')\in\Sy\com:
    (com^*,\ch,\resp^*,\ch',\resp^{*\prime})\ot A^{\Oall}, \\
    &\qquad\qquad
    (\com,\dots):=\com^*,
    (\resp,\dots):=\resp^*,
    (\resp',\dots):=\resp^{*\prime}
    ]
  \end{align*}
  The right hand side is negligible by \autonameref{coro:pick.one.sound-oall}.
  Hence $P_S$ is negligible. This shows that the
  sigma-protocol from \autoref{def:sigma.comp} has computational
  special soundness.
\end{proof}

\subsection{Proof of \autoref{lemma:sigma-break.comp}}
\label{app:proof:lemma:sigma-break.comp}

\begin{proof}[of \autoref{lemma:sigma-break.comp}]
  By \autoref{def:total.break} (specialized to the sigma-protocol from
  \autoref{def:sigma.comp}), we need to construct a polynomial-time
  adversary $A_1,A_2,A_3$ such that:
  \begin{align*}
    P_A := \Pr[&\ok=1\ \land\ s\notin L_{R'}:
    s\ot
    A_1,\com^*\ot A_2,\ch\otR\bits\ellch,\resp^*\ot
    A_3(\ch),\\
    &\ok:=V(\com^*,\ch,\resp^*)
    ]\text{ is overwhelming}
  \end{align*}
  We use the same adversary $(A_1,A_2,A_3)$ as in the proof of
  \autoref{lemma:sigma-break}. Then $P_A$ here is the same as $P_A$ in
  the proof of \autoref{lemma:sigma-break}. (Here we additionally have
  the condition $s\notin L_{R'}$, but this condition is vacuously true
  since $R'=\varnothing$ and thus $L_{R'}=\varnothing$.) And in the
  proof of \autoref{lemma:sigma-break} we showed that $P_A$ is
  overwhelming.
\end{proof}

\section{Proofs for \autoref{sec:fiat}}

\subsection{Proof of \autoref{theo:know.break.fs}}
\label{app:proof:theo:know.break.fs}

\begin{lemma}[Attack on Fiat-Shamir]\label{lemma:fs-break}
  There exists a
  total knowledge break (\autoref{def:total.break}) against the
  Fiat-Shamir construction based on the sigma-protocol from
  \autoref{def:sigma}. (For any~$r$.)
\end{lemma}

\begin{proof}
  According to~\autoref{def:total.break} (specialized to the case of
  the Fiat-Shamir construction based on the sigma-protocol from
  \autoref{def:sigma}) we need to construct a polynomial-time quantum
  adversary $\Hat A_1,\Hat A_2$ such that:
  \begin{compactitem}
  \item Adversary success:
    \begin{align*}
      \Hat P_A:=\Pr[&\forall i.\ok_i=1:
      s\ot \Hat A_1^{H,\Oall},\
      \bigl((\com^*_i)_i,(\resp^*_i)_i\bigr)\ot\Hat A_2^{H,\Oall},\penalty0
      \\&
      \ch_1\Vert\dots\Vert\ch_r:= H(s,(\com^*_i)_i),\penalty0
      \ok_i:=V(\com^*_i,\ch_i,\resp^*_i)]
    \end{align*}
    is overwhelming. Here $V$ is the verifier of the sigma-protocol (\autoref{def:sigma}).
  \item Extractor failure: For any polynomial-time quantum $E$ (with
    access to the final state of $\Hat A_1$),
    $\Pr[s=s_0,w=w_0:s\ot \Hat A_1^{H,\Oall},w\ot E^{H,\Oall}(s)]$ is
    negligible.
  \end{compactitem}
  Let $A_1,A_2,A_3$ be the adversary from the proof
  of~\autonameref{lemma:sigma-break}.
  Our adversary is then as follows: 
  \begin{compactitem}
  \item $\Hat A_1$ outputs $s_0$. (Identical to $A_1$.)
  \item $\Hat A_2$ invokes the adversary $A_2$ $r$ times to get
    $\com^*_1,\dots,\com^*_r$. Then $\Hat A_2$ computes
    $\ch_1\Vert\dots\Vert\ch_r:=H(s,(\com_i^*)_i)$. Then $\Hat A_2$
    invokes $A_3$ $r$
    times to get
    $\resp^*_1\ot A_3(\ch_1),\dots,\resp^*_r\ot A_3(\ch_r)$.
    Then $\Hat A_2$ outputs $((\com^*_i)_i,(\resp^*_i)_i)$.
  \end{compactitem}
  \noindent\textbf{Adversary success:} We have
  \begin{align*}
    1-\Hat P_A &= 
    \begin{aligned}[t]
      \Pr[&\exists i.\ok_i=0:
      s\ot A_1^{\Oall},\
      \forall i. \com^*_i\ot A_2^{\Oall},\\
      &\ch_1\Vert\dots\Vert\ch_r:= H(s,(\com^*_i)_i),\
      \forall i. \resp^*_i\ot A_3^{\Oall}(\ch_i),\\
      & \forall i. \ok_i\ot V(\com^*_i,\ch_i,\resp^*_i)] 
    \end{aligned}\\
    &\starrel= 
    \begin{aligned}[t]
      \Pr[&\exists i.\ok_i=0:
      s\ot A_1^{\Oall},\
      \forall i. \com^*_i\ot A_2^{\Oall},\\
      &\forall i.\ch_i\otR\bits\ellch,\
      \forall i. \resp^*_i\ot A_3^{\Oall}(\ch_i),\
      \\&
      \forall i. \ok_i\ot V(\com^*_i,\ch_i,\resp^*_i)]
    \end{aligned} \\
    &\starstarrel\leq  \vphantom{\sum^r} \smash{\sum_{i=1}^r}
    \begin{aligned}[t]
      \Pr[&\ok_i=0:
      s\ot A_1^{\Oall},\
      \com^*_i\ot A_2^{\Oall},\\
      &\ch_i\otR\bits\ellch,\
      \resp^*_i\ot A_3^{\Oall}(\ch_i),\
      \\&
      \ok_i\ot V(\com^*_i,\ch_i,\resp^*_i)]
    \end{aligned} \\
    &\tristarrel=   \sum_{i=1}^r (1-P_A) = r(1-P_A).
  \end{align*}
  Here $(*)$ uses the fact that $H$ is only queried once (classically), and thus
  $H(s,(\com^*_i)_i)$ is uniformly random. And $(**)$ is a union
  bound.  And $(*\mathord**)$ is by definition of $P_A$ in the proof
  of \autoref{lemma:sigma-break}. There is was also shown that $P_A$
  is overwhelming. Thus $1- \Hat P_A\leq r(1-P_A)$ is negligible and
  hence $\Hat P_A$ overwhelming. Thus we have adversary success.

  \medskip

  \noindent\textbf{Extractor failure}:
  Extractor failure was already shown in the proof of
  \autoref{lemma:sigma-break}. ($A_1$ here is defined exactly as $\Hat
  A_1$ in
  the proof of \autoref{lemma:sigma-break}, and the definition of
  extractor failure depends only on $\Hat A_1$, not on $\Hat A_2$ or the
  protocol being attacked.)

  Note that we have actually even shown extractor failure in the case
  that the extractor is allowed to choose the random oracle $H$ before
  and during the execution of $A_1$, because $A_1$ does not access~$H$.
\end{proof}

Now \autoref{theo:know.break.fs} follows from
\autonameref{lemma:sigma.sec} and \autoref{lemma:fs-break}. (The fact that the
Fiat-Shamir protocol is a classical argument of knowledge is shown in
\cite{faust12fiat}.\footnote{Actually, \cite{faust12fiat} requires
  perfect completeness instead of completeness as defined here (we
  allow a negligible error).  However, it is straightforward to see
  that their proof works unmodified for completeness as defined here.

  Also, \cite{faust12fiat} assumes that $\ellch$ is superlogarithmic,
  and considers the case $r=1$. But \cite{faust12fiat} can be applied
  to our formulation by first parallel composing the sigma-protocol
  $r$ times (yielding a protocol with challenges of length $r\ellch$),
  and then applying the result from \cite{faust12fiat}.})

\subsection{Proof of \autoref{theo:break.fs.comp}}
\label{app:proof:theo:break.fs.comp}

\begin{lemma}[Attack on Fiat-Shamir, computational]\label{lemma:fs-break.comp}
  Then there exists a
  total break (\autoref{def:total.break}) against the
  Fiat-Shamir construction based on the sigma-protocol from
  \autoref{def:sigma.comp}. (For any~$r$.)
\end{lemma}

\begin{proof}
  By \autoref{def:total.break} (specialized to the case of
  the Fiat-Shamir construction based on the sigma-protocol from
  \autoref{def:sigma.comp}), we need to construct a polynomial-time
  adversary $A_1,A_2$ such that:
    \begin{align*}
      \Hat P_A:=\Pr[&\forall i.\ok_i=1\land s\notin L_{R'}:
      s\ot \Hat A_1^{H,\Oall},\
      \bigl((\com^*_i)_i,(\resp^*_i)_i\bigr)\ot\Hat A_2^{H,\Oall},\penalty0
      \\&
      \ch_1\Vert\dots\Vert\ch_r:= H(s,(\com^*_i)_i),\penalty0
      \ok_i:=V(\com^*_i,\ch_i,\resp^*_i)]\text{ is overwhelming}
    \end{align*}
    Here $V$ is the verifier of the sigma-protocol (\autoref{def:sigma.comp}).
    
    We use the same adversary $(\Hat A_1,\Hat A_2)$ as in the proof of
    \autonameref{lemma:fs-break}. Then $\Hat P_A$ here is the same as
    $\Hat P_A$ in
    the proof of \autoref{lemma:fs-break}. (Here we additionally have
    the condition $s\notin L_{R'}$, but this condition is vacuously true
    since $R'=\varnothing$ and thus $L_{R'}=\varnothing$.) And in the
    proof of \autoref{lemma:fs-break} we showed that $\Hat P_A$ is
    overwhelming.
\end{proof}

Now \autoref{theo:break.fs.comp} follows from
\multiautoref{lemma:sigma.sec.comp,lemma:fs-break.comp}. (The fact that the
Fiat-Shamir protocol is a classical argument of knowledge is shown in
\cite{faust12fiat}.\footnote{Actually, \cite{faust12fiat}
  requires
  perfect special soundness
  instead of computational special soundness, as well as
  perfect completeness instead of completeness as defined here (we
  allow a negligible error).  However, it is
  straightforward to see that their proof works unmodified for
  computational special soundness and 
  completeness as defined here.

  Also, \cite{faust12fiat} assumes that $\ellch$ is superlogarithmic,
  and considers the case $r=1$. But \cite{faust12fiat} can be applied
  to our formulation by first parallel composing the sigma-protocol
  $r$ times (yielding a protocol with challenges of length $r\ellch$),
  and then applying the result from \cite{faust12fiat}.})

\section{Proofs for \autoref{sec:fischlin}}

\subsection{Proof of \autoref{theo:know.break.fischlin}}
\label{app:proof:theo:know.break.fischlin}

\begin{lemma}[Attack on Fischlin's construction]\label{lemma:fischlin-break}
  There exists a
  total knowledge break (\autoref{def:total.break}) against the
  Fischlin construction based on the sigma-protocol from
  \autonameref{def:sigma}.
\end{lemma}

\begin{proof}
  According to~\autonameref{def:total.break} (specialized to the case of
  Fischlin's construction based on the sigma-protocol from
  \autoref{def:sigma}) we need to construct a polynomial-time quantum
  adversary $A_1,A_2$ such that:
  \begin{compactitem}
  \item Adversary success:
  \begin{align}
    P_A:=\Pr\Bigl[{}&\forall i.\ok_i=1\land \mathit{\sigma\leq \fisS}\land s=s_0:
    s\ot A_1^{H,\Oall},\notag\\
    &(\com^*_i,\ch_i,\resp^*_i)_{i=1\dots \fisr}\ot A_2^{H,\Oall},\
    \ok_i:=V(\com^*_i,\ch_i,\resp^*_i),\notag\\
    &\sigma:=\sum_{i=1}^\fisr H(x,(\com^*_i)_i,i,\ch_i,\resp^*_i)\Bigr]\text{ is overwhelming.}
    \label{eq:adv.succ.fischlin}
  \end{align}
  \item Extractor failure: For any polynomial-time quantum $E$ (with
    access to the final state of $A_1$),
    $\Pr[s=s_0,w=w_0:s\ot A_1^{H,\Oall},w\ot E^{H,\Oall}(s)]$ is
    negligible.
  \end{compactitem}

  \noindent\textbf{Adversary success:} 
  At the first glance, it may seem that it is immediate how to
  construct an adversary that has adversary success: Using
  \autonameref{theo:pick1.complete}, we can for each $i$ search
  $(\ch_i,\resp_i)\in\Sy{\com_i}$ such that
  $H(x,(\com^*_i)_i,i,\ch_i,\resp^*_i)=0$. However, there is a
  problem: $\com_i^*$ contains commitments $c^i_\ch$ to all
  responses. Thus, after finding $\ch_i,\resp_i$, we need to open
  $c^i_{\ch_i}$ as $\resp_i$. This could be done with the adversary
  against $\COM$ from \autonameref{lemma:com.attack}. But the problem is,
  the corresponding openings have to be contained in $\resp_i^*$. So
  we need to know these openings already when searching for
  $\ch_i,\resp_i$. But at that point we do not know yet to what value
  the commitments $c_{\ch_i}^i$ should be opened! To avoid this
  problem, we use a special fixpoint property of the commitment scheme
  $\COM$ that allows us to commit in a way such that we can use the
  $(\ch_i,\resp_i)$ themselves as openings for the commitments.

  The fixpoint property is the following: There are functions
  $\COMstar$\symbolindexmark{\COMstar},
  $\COMopenstar$\symbolindexmark{\COMopenstar}
  such that for any $\com\in\bits\ellcom$,
  and any $(\ch,\resp)\in\Sy\com$, we have
  \begin{equation}
    \COMverify(c,\resp,u)=1
    \text{ for }
    c:=\COMstar(\com)
    \text{ and }
    u:=\COMopenstar(\ch,\resp).
    \label{eq:comstar}
  \end{equation}
  These functions are defined as follows:
  $\COMstar(\com)=(p_1,\dots,p_{\ellresp},y_1,\dots,y_{\ellresp},b_1,\dots,b_{\ellresp})$ with
  $p_i:=\ellch+i$, $y_i:=\com$, $b_i:=0$. And
  $\COMopenstar(\ch,\resp):=(x_1,\dots,x_{\ellresp})$ with
  $x_i:=(\ch,\resp)$ for all $i$. It is easy to verify from the definition of
  $\COMverify$ (\autoref{def:com}) that \eqref{eq:comstar} holds if
  $(\ch,\resp)\in\Sy\com$.

  Our adversary is as follows: 
  \begin{compactitem}
  \item $A_1$ outputs $s_0$.
  \item
    $A_2$ invokes $E_1$ from~\autonameref{theo:pick1.complete} $\fisr$
    times to get $\bigl(\com_i,\ketpsiy{\com_i}\bigr)$ for $i=1,\dots,r$.
    $A_2$ sets $c^i_\ch:=\COMstar(\com_i)$ for all $i$ and all $\ch\in\bits\ellch$.
    And $\com_i^*:=\bigl(\com_i,(c_\ch^i)_\ch\bigr)$.

    Let $P_i(\ch',\resp'):=1$ iff
    $H(s,(\com_i^*)_i,i,\ch',(\resp',\COMopenstar(\ch',\resp')))=0$.
    Then, for each $i=1,\dots,\fisr$, $A_2$ invokes $E_2(n,\delta_{\min},\com_i,\ketpsiy{\com_i})$
    from~\autoref{theo:pick1.complete} with oracle access to
    $P:=P_{i}$ and with $n:=\ellcom$ and $\delta_{\min}:=2^{-b-1}$ to get $\ch_i,\resp_i$. 
    Let $\resp^*_i:=(\resp_i,\COMopenstar(\ch_i,\resp_i))$.
    Then $A_2$ outputs
    $\pi:=(\com^*_i,\ch_i,\resp^*_i)_{i=1,\dots,r}$.
  \end{compactitem}
  Consider an execution of $A_1,A_2$ as in
  \eqref{eq:adv.succ.fischlin}. Let $\mathsf{Succ}_i$ denote the event
  that $(\ch_i,\resp_i)\in\Sy{\com_i}\land P_i(\ch_i,\resp_i)=1$ in that execution. We have
  \begin{align}
    \Pr[\mathsf{Succ}_i] =
      \Pr[&(\ch,\resp)\in\Sy{\com_i}\land P(\ch,\resp)=1:
      \forall j.\bigl(\com_j,\ketpsiy{\com_j}\bigr)\ot E_1,\notag\\
      &\forall j.\com_j^*:=\bigl(\com_j,(\COMstar(\com_j))_\ch\bigr),
      H\otR(\bits*\to\bits b),\notag\\
      &\forall\ch'\resp'.P(\ch',\resp'):=1\text{ iff }
      H(s,(\com_j^*)_j,i,\ch',(\resp',\COMopenstar(\ch',\resp')))=0,\notag\\
      &(\ch,\resp)\ot E_2(n,\delta_{\min},\com_i,\ketpsiy{\com_i})].
    \label{eq:succ_i}
  \end{align}
  Hence by \autonameref{theo:pick1.complete}, 
  \[
  \Pr[\mathsf{Succ}_i] \geq 1-2^{-\ellcom} - 
  \underbrace{\Pr\Bigl[\tfrac{\abs{\{(\ch,\resp)\in\Sy{\com_j}:P(\ch,\resp)=1\}}}{\abs{\Sy{\com_j}}}<\delta_{\min}\Bigr]}_{=:p_\delta}.
  \]
  Here $P$ and $\com$ are chosen as in the rhs of \eqref{eq:succ_i}.

  In the rhs of \eqref{eq:succ_i}, $H$ is chosen after
  $\Sy{\com_j},s,\com_j^*$, and $i$ are fixed. Thus for every,
  $(\ch,\resp)\in \Sy{\com_i}$ it is independently chosen whether
  $P(\ch,\resp)=1$ or $P(\ch,\resp)=0$, where
  $\Pr[P(\ch,\resp)=1]=2^{-\fisb}$.  
  Thus
  \begin{align*}
    p_\delta&=\Pr\Bigl[\sum\nolimits_{i\in S} \tfrac{X_i}{\abs{S}}\geq 1-\delta_{\min}\Bigr]
    =\Pr\Bigl[\sum\nolimits_{i\in S} \tfrac{X_i}{\abs{S}}-(1-2^{-\fisb})\geq 1-\delta_{\min}-(1-2^{-\fisb})\Bigr]\\
    &\starrel\leq e^{-2\abs{S}(1-\delta_{\min}-(1-2^{-\fisb}))^2} = e^{-2\kk(2^{-\fisb}-\delta_{\min})^2} = 
    e^{-\kk(2^{-2\fisb-1})}
  \end{align*}
  where $X_{\ch,\resp}:=1-P(\ch,\resp)$ and $S:=\Sy{\com_i}$. And
  $(*)$ follows from Hoeffding's inequality \cite{JASA1963:Hoeffding}.

  We thus have
  \[
  \Pr[\forall i=1\dots \fisr.\ \mathsf{Succ}_i] \geq 1-2^{-\ellcom} \fisr -\fisr e^{-\kk(2^{-2\fisb-1})} =: p_s
  \]
  Since $r$ is polynomially bounded and $b$ is logarithmic and $\ellcom,\kk$
  are superpolynomial, $p_s$ is overwhelming.
  
  For adversary success, it remains to show that $P_A\geq p_s$ where
  $P_A$ is as in \eqref{eq:adv.succ.fischlin}. For this, we show that
  $\forall i.\mathsf{Succ}_i$ implies
  $\forall i.\ok_i=1\land\sigma\leq\fisS\land s=s_0$. First, note that
  $s=s_0$ always holds by definition of $A_1$. Furthermore, 
  $\forall i.\mathsf{Succ}_i$ implies (by definition of $P_i$) that 
  \begin{align*}
    \sigma &= \sum\nolimits_i H(s,(\com_i^*)_i,i,\ch_i,\resp^*_i) \\
    &=\sum\nolimits_i H\Bigl(s,(\com_i^*)_i,i,\ch_i,\bigl(\resp_i,\COMopenstar(\ch_i,\resp_i)\bigr)\Bigr)=
    \sum\nolimits_i 0 \leq \fisS.
  \end{align*}
  Finally, if $\mathsf{Succ}_i$ holds, then $(\ch_i,\resp_i)\in\Sy{\com_i}$, thus 
  \begin{multline*}
    \COMverify(c_{\ch_i}^i,\resp_i,\COMopenstar(\ch_i,\resp_i))\\
    =\COMverify(\COMstar(\com_i),\resp_i,\COMopenstar(\ch_i,\resp_i))\eqrefrel{eq:comstar}=1.
  \end{multline*}
  And $\OV(\com_i,\ch_i,\resp_i)=1$. Thus
  $\ok_i=V(\com^*_i,\ch_i,\resp^*_i)=1$. Summarizing,
  $\forall i.\mathsf{Succ}_i$ implies
  $\forall i.\ok_i=1\land\sigma\leq\fisS\land s=s_0$ and thus
  $P_A\geq p_s$. Since $p_s$ is overwhelming, so is $P_A$, thus we have adversary success.

  \medskip\noindent\textbf{Extractor failure:}
  Extractor failure was already shown in the proof of
  \autoref{lemma:sigma-break}. ($A_1$ here is defined exactly as in
  the proof of \autoref{lemma:sigma-break}, and the definition of
  extractor failure depends only on $A_1$, not on $A_2$ or the
  protocol being attacked.)

  Note that we have actually even shown extractor failure in the case
  that the extractor is allowed to choose the random oracle $H$ before
  and during the execution of $A_1$, because $A_1$ does not access~$H$.
\end{proof}

\medskip

\noindent
Now \autoref{theo:know.break.fischlin} follows from
\autonameref{lemma:sigma.sec} and \autoref{lemma:fischlin-break}. (The fact that
Fischlin's construction is a classical argument of knowledge is shown
in \cite{fischlin05online}.\footnote{Actually, \cite{fischlin05online}
  requires perfect completeness instead of completeness as defined
  here (we allow a negligible error). However, it is straightforward
  to see that their proof works unmodified for completeness as defined
  here.})

\subsection{Proofs for \autoref{theo:break.fischlin.comp}}
\label{app:proof:theo:break.fischlin.comp}

\begin{lemma}[Attack on Fischlin's construction, computational]\label{lemma:fischlin-break.comp}
  Then there exists a total break (\autoref{def:total.break}) against
  Fischlin's construction based on the sigma-protocol from
  \autonameref{def:sigma.comp}.
\end{lemma}

\begin{proof}
  By \autoref{def:total.break} (specialized to the case of
  Fischlin's construction based on the sigma-protocol from
  \autoref{def:sigma.comp}), we need to construct a polynomial-time
  adversary $A_1,A_2$ such that:
  \begin{align*}
    P_A:=\Pr\Bigl[{}&\forall i.\ok_i=1\land \mathit{\sigma\leq \fisS}\land s=s_0\land s\notin L_{R'}:
    s\ot A_1^{H,\Oall},\\
    &(\com^*_i,\ch_i,\resp^*_i)_{i=1\dots \fisr}\ot A_2^{H,\Oall},\
    \ok_i:=V(\com^*_i,\ch_i,\resp^*_i),\\
    &\sigma:=\sum_{i=1}^\fisr H(x,(\com^*_i)_i,i,\ch_i,\resp^*_i)\Bigr]\text{ is overwhelming.}
  \end{align*}
  Here $V$ is the verifier of the sigma-protocol (\autoref{def:sigma.comp}).

  We use the same adversary $(A_1,A_2)$ as in the proof of
  \autonameref{lemma:fischlin-break}. Then $P_A$ here is the same as $P_A$ in
  the proof of \autoref{lemma:sigma-break}. (Here we additionally have
  the condition $s\notin L_{R'}$, but this condition is vacuously true
  since $R'=\varnothing$ and thus $L_{R'}=\varnothing$.) And in the
  proof of \autoref{lemma:fischlin-break} we showed that $P_A$ is
  overwhelming.
\end{proof}

\medskip\noindent
Now \autoref{theo:break.fischlin.comp} follows from
\autonameref{lemma:sigma.sec.comp} and \autoref{lemma:fischlin-break.comp}. (The
fact that Fischlin's construction is a classical argument of knowledge
is shown in \cite{fischlin05online}.\footnote{Actually, \cite{fischlin05online} requires
 perfect special soundness
  instead of computational special soundness, as well as
  perfect completeness instead of completeness as defined here (we
  allow a negligible error).  However, it is
  straightforward to see that their proof works unmodified for
  computational special soundness and 
  completeness as defined here.})

\end{fullversion}

\end{document}